\documentclass[11pt,a4paper]{article}
\usepackage{fullpage}

\usepackage[english]{babel}
\usepackage[T1]{fontenc}
\usepackage[tt=false]{libertine}
\usepackage{microtype}

\usepackage{amsmath}
\usepackage{amssymb}
\usepackage{amsthm}
\usepackage{mathtools}
\usepackage{bm}  
\usepackage[libertine,vvarbb]{newtxmath}
\allowdisplaybreaks%

\newtheorem{theorem}{Theorem}[section]
\newtheorem{corollary}[theorem]{Corollary}
\newtheorem{lemma}[theorem]{Lemma}

\theoremstyle{definition}
\newtheorem{definition}[theorem]{Definition}

\newtheorem{remark}[theorem]{Remark}
\newtheorem{example}[theorem]{Example}

\usepackage{float}
\usepackage{color}
\usepackage[svgnames]{xcolor}
\usepackage{graphicx}

\usepackage{hyperref}
\hypersetup{colorlinks={true},urlcolor={DarkBlue},linkcolor={DarkBlue},citecolor=[named]{DarkGreen}}
\usepackage{footnotebackref}  

\usepackage{caption}
\usepackage{subcaption}
\usepackage{tikz}
\usetikzlibrary{arrows.meta,positioning,fit,backgrounds}
\usepackage{enumitem}

\usepackage{booktabs}
\usepackage{multirow, makecell, colortbl, hhline}
\usepackage{diagbox}
\usepackage{makecell}

\usepackage{algorithm2e}

\SetAlFnt{\small}
\SetAlCapFnt{\small}
\SetAlCapSty{normalfont}  
\SetAlCapNameFnt{\small}
\SetAlCapHSkip{0pt}
\IncMargin{-\parindent}

\newcommand*{\vecc}[1]{\bm{#1}}
\newcommand*{\N}{\mathbb{N}}
\newcommand*{\R}{\mathbb{R}}
\newcommand*{\card}[1]{\left|#1\right|}

\definecolor{TUMBlue}{HTML}{0065BD}
\definecolor{TUMGreen}{HTML}{64A12D}
\definecolor{TUMOrange}{HTML}{E37222}
\usepackage{adjustbox}

\usepackage{orcidlink}

\begin{document}

\title{Stability in Combinatorial Markets with Side Payments\footnote{An extended abstract of this work has been accepted at SAGT'26 and will be published in the proceedings of the conference (LNCS, volume 16911).}}
\author{Alexander Grosz\thanks{Chair of Operations Research, Technical University of Munich. \texttt{\{alexander.grosz,chiara.vanoli\}@tum.de}}\; \orcidlink{0000-0001-9744-0253}
\and Chiara Vanoli\footnotemark[2]\; \orcidlink{0009-0006-0094-5209}}

\maketitle
\thispagestyle{empty}

\begin{abstract}
  Combinatorial markets provide a general framework for trading bundles of indivisible goods. Building on the combinatorial market models of Bikhchandani \& Ostroy \cite{BikhchandaniOstroy:ThePackageAssignmentModel} and Bichler \& Waldherr~\cite{BichlerWaldherr:CorePricingCombExchange}, we introduce explicit side payments, thereby allowing restricted transfers of utility among subsets of agents and capturing different forms of financial collusion.

  This extension results in a large number of seemingly distinct market settings. However, we establish a systematic classification of their expressive power. Most notably, we show that higher-order market settings (i.e., those with more personalized prices and less clearinghouse power) are essentially equivalent in terms of market properties when side payments are permitted between buyers. An analogous collapse occurs for the third- and higher-order settings when side payments are permitted between sellers. In contrast, we identify a fundamental structural separation between the second- and third-order settings.

  To analyze stability in these environments, we generalize the classical concepts of stability with and without transferable utility (TU and NTU) to a partition-based notion of restricted transferability (the $\mathcal{T}$-core). We relate stable outcomes across different settings to appropriate stability notions, identify market instances in which partially transferable utility yields a better (or any) stable outcome, and show that under personalized pricing, all stability notions collapse to NTU-stability.

\end{abstract}

\newpage
\thispagestyle{empty}
\setcounter{tocdepth}{2}
\tableofcontents
\thispagestyle{empty}
\newpage

\setcounter{page}{1}

\section{Introduction}

Explaining the behavior of agents in markets where goods are exchanged for money has been an important field of research for the better part of the last century (starting with the seminal contribution by Arrow \& Debreu \cite{ArrowDebreu:MarketEquilibria}).
Such market models can be applied to plenty of real-life scenarios.
Our object of interest is a combinatorial market in which agents trade non-divisible items offered and requested as bundles.
Such markets have applications, e.g., in logistics for cargo transfer, airport time-slot allocation, spectrum auctions, and energy markets (see, e.g., \cite{BichlerWaldherr:CorePricinginFinancialConstraintsCombExchange} for references).

The theory of combinatorial markets generalizes the simpler case of assignment markets, in which each buyer and seller trades only a single indivisible item at a time \cite{ShapleyShubik:Core}.
In combinatorial markets, agents exchange bundles or packages of items, and both the valuation functions and prices can differ structurally from the simple linear case.
Computing a socially optimal trade in such markets is NP-hard, even for a single seller, which underscores the structural complexity of combinatorial markets. Additionally, the exhaustive elicitation of preferences poses an interesting problem and has seen various proposed solutions in terms of communication models for valuations \cite{Sandholm:WinnerDetermination}, including fully expressive OR language or other logically described languages \cite{Nisan:Bidding}, and compact, problem-specific languages \cite{EmadikhiavDay:CombinatorialMarketsTruckload,GBSD:CompactBidLanguages}.

Bikhchandani \& Ostroy \cite{BikhchandaniOstroy:ThePackageAssignmentModel} establish four settings for combinatorial markets, distinguished by the degree of control exercised by a central clearinghouse, which can modify the exchange by opening and rebundling packages and relabeling them. These settings also differ in their pricing restrictions, from linear item prices shared across the market to fully individualized package prices. These markets are formulated as integer programs (and the corresponding dual pricing problems) and, as such, allow the identification of pricing equilibria and integral solutions. In general, when pricing equilibria exist, they belong to the core, but not every element of the core may represent an equilibrium. However, in the special case of a single seller (i.e., the combinatorial auction setting), the set of pricing equilibria does coincide with the core, thereby establishing their stability against coalitional deviations.

Building on the work of \cite{BikhchandaniOstroy:ThePackageAssignmentModel}, Bichler \& Waldherr \cite{BichlerWaldherr:CorePricingCombExchange} introduce a model that combines two extremes from previous settings: the clearinghouse can rearrange packages freely, but bundle prices remain fully personalized. The underlying integer programming model is more verbose than in \cite{BikhchandaniOstroy:ThePackageAssignmentModel}, in that it defines binary variables for every feasible package combination in the full market, rather than modeling the occurrence of each package directly. In this paradigm, they establish that the core coincides with the set of pricing equilibria, whenever any exist. In more recent work, Bichler \& Waldherr \cite{BichlerWaldherr:CorePricinginFinancialConstraintsCombExchange} augment their setting with budget constraints and show that computing welfare-optimal stable outcomes is $\Sigma_2^p$-hard.

In related work,
\cite{BDO:EndowmentEffect,EFF:FrameworkforEndowmentEffects}
discuss the existence of equilibria for special valuation structures with the endowment effect, both under strong assumptions that admit equilibria and under weaker assumptions that do not. In \cite{FGL:CombinatorialWalrasianEquilibrium} (see also \cite{FeldmanLucier:ClearingMarketsViaBundles}), a relaxation of the standard equilibrium notion is introduced in a reduced market to make the problem tractable.
There is substantial literature on pricing methods for combinatorial and iterative (dynamic pricing) auctions; see \cite{Cramton2006} for a comprehensive overview.  A different type of combinatorial market is discussed in \cite{DGMVY:MarketswithCoveringConstraints}, where goods are divisible, but certain covering constraints must be satisfied. In \cite{FadaeiBichler:LimitsofApproximationinCombinatorialMarkets} a mechanism design perspective to address the lack of truthful, revenue-maximizing mechanisms is discussed by relaxing to approximately optimal revenue.

\subsubsection*{Our Contributions.}
Our work builds directly on the combinatorial market models of Bikhchandani \& Ostroy \cite{BikhchandaniOstroy:ThePackageAssignmentModel} and their extension by Bichler \& Waldherr~\cite{BichlerWaldherr:CorePricingCombExchange}.
Motivated by market environments in which regulators restrict cooperation on one side of the market (e.g., through competition or antitrust laws), or in which informational frictions prevent agents from engaging in strategic coordination, we extend existing market models by introducing explicit side payments. We consider different degrees of transferability, depending on which agents may exchange utility: only buyers, only sellers, buyers and sellers independently within their respective sides (but not across sides), and finally all agents.

To systematically compare the five original settings now extended with (restricted) side payments, we introduce hierarchical equivalence levels. These notions indicate which properties of a given reference setting can also be represented in another setting.
While in some cases we can only find a corresponding market outcome that trades the same cumulative bundle for every agent (\emph{item equivalence}), sufficiently verbose settings can realize not just the cumulative prices each agent pays or receives (\emph{price equivalence}), but even the individual prices paid for every traded bundle (\emph{market equivalence}). We consider \emph{payoff equivalence} to be the most relevant property, as it describes when market outcomes exist with the same set of traded items and the same utility for each agent and thereby captures their strategic incentive.

Our first set of results establishes a near-monotonicity principle: increasing the power of the pricing system -- either by moving to higher-order market settings, i.e., simultaneously increasing price personalization while reducing the power of the clearinghouse, or by allowing richer side payments -- generally expands the set of implementable assignments (Theorem~\ref{thm:setting-equivalences}, Remark~\ref{remark:marketEquivDiagonal}, Lemmas~\ref{lemma:1to234(M)} to \ref{lemma:1234to5(P)}).
The main exception arises in the second-order setting, where a buyer-seller pair can effectively trade multiple bundles, leading to a qualitative change once personalized prices and labels are introduced (Example~\ref{example:34noto2(M)}).

This structural peculiarity of the second-order setting breaks payoff equivalence when moving to higher-order settings under restrictive transfer regimes. However, payoff equivalence is recovered if sufficiently strong side payments are allowed, including separate transfers between buyers or sellers, or seller-side transfers under compatible restrictions (Lemmas~\ref{lemma:12345to1(Pi)FullSep} and \ref{lemma:245to3(Pi)Sel}). Moreover, every second-order assignment remains price equivalent to an assignment in the fourth-order setting (Lemma~\ref{lemma:2to4(P)}), although full market equivalence need not hold (Example~\ref{example:34noto2(M)}).

Beyond this case, tightening side-payment restrictions while increasing the order of the market typically makes it impossible to exceed item-equivalence between settings.
The only exception is going from full to separate side payments, which preserves payoff equivalence across all order settings (Lemma~\ref{lemma:12345to1(Pi)FullSep}).

When moving to a lower-order market setting, payoff equivalence is always preserved provided that full or separate side payments are allowed. Intuitively, in these cases, most or all monetary transfers can be implemented through side payments, effectively circumventing the pricing restrictions (Lemma~\ref{lemma:12345to1(Pi)FullSep}). Two further instances where payoff equivalence holds are the following: any assignment with sellers’ single or no side payments in the fourth- or fifth-order setting is payoff equivalent to an assignment in the third-order setting with sellers’ single side payments (Lemma~\ref{lemma:245to3(Pi)Sel}), and any assignment in the fifth-order setting with single or no side payments is payoff equivalent to an assignment in the fourth-order setting with the same or more restrictive side-payment structure (Lemma~\ref{lemma:5to4(Pi)Sing}).

We identify settings that are fully equivalent in a strong sense: for certain combinations of market rules and side-payment regimes, payoff equivalence holds in both directions. In these cases, not only the assignments but also the deviations coincide, implying that the core is identical. We identify the following equivalences, which follow directly from the Lemmas outlined above.

\begin{theorem}\label{thm:setting-equivalences}~
  \begin{itemize}
    \item All order settings with full and separate side payments are payoff equivalent.
    \item The fourth-order and the fifth-order settings with buyers' single side payments are payoff equivalent.
    \item The third-, fourth-, and fifth-order settings with sellers' single side payments are payoff equivalent.
  \end{itemize}
\end{theorem}

Thus, under sufficiently permissive side-payment regimes, distinctions between several higher-order market rules become immaterial for both payoff implementability and stability.
Notice the asymmetry between the buyers' and sellers' side results: this is due to the restriction of the third-order setting to only one of the two possible sides where the labels are disregarded (see Section~\ref{sec:prelims} for more detail).

To investigate the solution properties of these new market settings, we turn to the standard notions of TU-stability (transferable utility) as well as NTU-stability (non-transferable utility).
While either notion describes solutions that are stable against coalitional deviations, they differ in how such improvements can be achieved: under TU, agents may freely transfer utility, whereas under NTU, each agent’s payoff is fixed by the market outcome.

We generalize these notions by the definition of $\mathcal{T}$-RTU-stability (restricted TU) in Section~\ref{sec:stability-notions}, where $\mathcal{T}$ denotes a partition of the agents, and transfers of money are allowed  only between agents in the same partition block. The new notion generalizes both TU-stability and NTU-stability in that it additionally allows us to model stable outcomes that include collusion on one side or two different sides of the market, while restricting all money between the two sides to flow through the market via payments.

These new stability notions are closely related to the side-payment restrictions outlined above.
In Section~\ref{sec:stability_relations}, we show that finding $\mathcal{T}$-RTU-stable outcomes is equivalent to finding an NTU-stable outcome in the market setting that admits a corresponding set of side payments within the market itself.
Thus, we can effectively simplify the analysis of different stability notions to NTU-stability with the respective side-payment restrictions.

\begin{theorem}\label{thm:allisNTU}
  For a given $h^\mathrm{th}$-order market $\mathcal{E}_h^r$ with side-payment restriction $r$ and a side-payment partition $\mathcal{T}$, its $\mathcal{T}$-core $\mathcal{C}^{\mathcal{T}}(\mathcal{E}_h^r)$ corresponds to the NTU-core $\mathcal{C}^{\mathcal{T}^\mathrm{no}}(\mathcal{E}_h^{s})$ of the same economy with side-payment restriction $s$, i.e.,
  $$\mathcal{C}^{\mathcal{T}}(\mathcal{E}_h^r) = \mathcal{C}^{\mathcal{T}^\mathrm{no}}(\mathcal{E}_h^{s}),$$ where $s$ is the side-payment restriction induced by the join $\mathcal{S}$ of the respective partitions, i.e., $\mathcal{S} = \mathcal{T} \vee \mathcal{T}^r$.
\end{theorem}

We analyze several properties of stable outcomes under the different stability notions introduced above. By establishing separations between these notions -- or, equivalently, between different side-payment regimes -- we show that they can lead to stable outcomes with strictly different levels of social welfare (Examples~\ref{example:NTUDifferentPayoffs} to \ref{example:RTUUniquePayoff}). This demonstrates that allowing restricted forms of collusion on one side of the market can fundamentally alter the set of stable outcomes. Thus, our results emphasize the importance of accurately modeling the extent of collusion in combinatorial markets with correlated pricing systems.

In market settings with personalized prices, on the other hand, we identify strong structural properties of stable solutions: We provide a constructive method to eliminate all side payments from a stable outcome (Lemmas~\ref{lemma:NTUStableAss5} and \ref{lemma:fourthSettingStableAssignments}), showing that in such settings all stability notions provide equivalent outcomes, regardless of which side payments are permitted in the market (Theorem~\ref{thm:45stability_all_NTU}).

\begin{theorem}\label{thm:45stability_all_NTU}
  Let $\mathcal{E}$ be an economy and $\mathcal{T}$ an arbitrary partition of its agents. Then, all of the following $\mathcal{T}$-cores $\mathcal{C}^{\mathcal{T}}(\mathcal{E}_h^r)$ for an $h^\mathrm{th}$-order market with side-payment restriction $r$ coincide:
  \begin{itemize}
    \item for $h \in \{1, \dots, 5\}$, $r \in \{\mathrm{full}, \, \mathrm{sep}, \, \mathrm{buy}, \, \mathrm{sell}, \, \mathrm{no}\}$ and $\mathcal{T}$, such that the join of the respective partitions is at least as coarse as $\mathcal{T}^\mathrm{sep}$, i.e., $\mathcal{T} \vee \mathcal{T}^r \succeq \mathcal{T}^{\mathrm{sep}}$,
    \item for $h \in \{4, 5\}$ and $r \in \{\mathrm{full}, \, \mathrm{sep}, \, \mathrm{buy}, \, \mathrm{sell}, \, \mathrm{no}\}$, and
    \item for $h = 3$ and $\mathcal{T}$, such that the join of the respective partitions has the sellers as a group of transferable utility, i.e., $\mathcal{T} \vee \mathcal{T}^r \succeq \mathcal{T}^{\mathrm{sell}}$.
  \end{itemize}
\end{theorem}

\section{Preliminaries}\label{sec:prelims}

The following setting for combinatorial markets based on a package assignment model was introduced in \cite{BikhchandaniOstroy:ThePackageAssignmentModel}.
Appendix~\ref{secapp:formal_defs} contains a detailed mathematical description of the market model.
The crucial feature of this model of package exchanges is that trades are not directly executed between seller-buyer pairs, but instead routed through a \emph{clearinghouse} that receives bundles from sellers and, under certain restrictions, potentially repackages them into new bundles before shipping them to buyers. It also collects all payments from buyers and distributes them to sellers in accordance with an adequate market pricing system.

The market settings differ in the possible alterations the clearinghouse can make to the packages it receives and in the restrictions on the pricing system.
The first four settings were first defined in \cite{BikhchandaniOstroy:ThePackageAssignmentModel} and the fifth setting was introduced in \cite{BichlerWaldherr:CorePricingCombExchange}.
We provide a quick overview of the different restrictions in the following table.
The formal definition of each setting is deferred to Appendix~\ref{secapp:formal_market-settings}.

\begin{center}{\small
    \begin{tabular}{c|c|c|c}
      Setting & Clearinghouse & Bundle Compatibility & Pricing \\ \hline
      1st order  & Repackage & Item counts & Linear (item based)\\
      2nd order & Relabel & Bundle counts & Per bundle \\
      3rd order  & Relabel sellers & Bundle counts for each buyer & Buyer personalized \\
      4th order  & -- & Direct trade & Fully personalized \\
      5th order  & Repackage & Item counts & Independent
    \end{tabular}
  }
\end{center}

In the first-order setting, the clearinghouse can both repackage and reassign the bundles, while any package labels are completely disregarded. Therefore, only the total item counts on both sides of the market submissions have to coincide. A first-order pricing system is based on \emph{item} prices, independently of the bundle structure. The bundle prices are then derived linearly from the item prices.

In the second-order setting, the clearinghouse is not allowed to open any packages. However, both the sellers' and buyers' labels on the packages can be freely adjusted. The corresponding pricing system is based on common \emph{bundle} prices across all agents.

In the third-order setting, the clearinghouse can only change the labels of one side of the market (and we restrict our discussion to the sellers' labels similarly to \cite{BikhchandaniOstroy:ThePackageAssignmentModel}).
The pricing system becomes \emph{personalized for the buyers}, i.e., the prices for the same bundle sent to any given buyer have to coincide, but a seller may send the same bundle to different buyers at different prices.

For the fourth-order setting, the restrictions on the clearinghouse represent essentially the \emph{direct trade} situation: no labels can be adjusted and the prices are personalized between any buyer-seller pair.

Finally, the fifth-order setting is based on \emph{cumulative bundles}: every agent only offers/requests a cumulative bundle and the prices are personalized to each agent for the entire bundle. This corresponds to the first-order setting in terms of the set of offered items, where only the total count on each market side must coincide. This setting relaxes the interaction between the two sides of the market in the pricing system, so budget balance needs to be explicitly assumed (as compared to the other settings, where it is derived from the interactions given by the pricing system).

We extend the model of \cite{BikhchandaniOstroy:ThePackageAssignmentModel} by introducing \emph{side payments} so as to make it applicable to economies with coordination issues, externalities, or other frictions that prices alone cannot resolve. Depending on the setting, we allow (a subset of) agents to exchange money outside the combinatorial market and its pricing restrictions. In particular, these payments are a priori available to any pair of agents, not restricted to seller-buyer pairs that trade a bundle.

For an economy $\mathcal{E}$ and $h \in \{1, \dots, 5\}$, we indicate by $\mathcal{E}_h$ the market with resale restrictions imposed by the $h^\mathrm{th}$-order setting.
The bundles that are being traded in the market by the agents $N$ are notated as an \emph{allocation} $(\vecc{S}_I, \vecc{Z}_J)$, where each agent specifies the bundle of items they want to offer to or request from each agent on the other side of the market.
A \emph{pricing system} is then given by functions $p_{ij}$, defining the price of bundle $\vecc\omega$ when sold by seller $j \in J$ and bought by buyer $i \in I$ as $p_{ij}(\vecc \omega)$.

We denote by $q_{kl} \in \R$ the amount of money that agent $k \in N$ side-pays to agent $l \in N$.
The case of $q_{kl}<0$ corresponds to $k$ receiving instead of paying the respective side payment.
We require consistency between side payments of two agents $k\neq l$, i.e., $q_{lk} = -q_{kl}$, so that each outgoing payment from $k$ is equivalently received by $l$.
For notational ease, we allow side payments $q_{kk}$ of an agent to themselves; $q_{kk} = 0$ follows from consistency.

For every $k \in N$, the individual side payments are denoted by $\vecc{Q}_k = (q_{kl})_{l \in N}$ and the corresponding cumulative side payments by $q_k = \sum_{l \in N} q_{kl}$.
The vector of all side payments is denoted by $\vecc{Q}_N = (\vecc{Q}_k)_{k\in N}$.

An \emph{assignment} $((\vecc{S}_I, \vecc{Z}_J), \vecc p, \vecc{Q}_N)$ in a market $\mathcal{E}_m$ is the tuple of the allocation, the pricing system, and the side payments. We may omit all side payments with value zero in the notation and use only $\vecc Q_C$ in place of $\vecc Q_N$ for some $C \subseteq N$ or fully omit $\vecc Q_N$, when no side payments are used or allowed.
For a given market setting $\mathcal{E}_m$, the set of \emph{admissible} assignments $\mathcal{A}(\mathcal{E}_m)$ is structurally restricted in one of the various ways outlined above. These restrictions will apply to the set of assignments that are valid in the market, as well as to the structure of the pricing system and the side payments.

We naturally extend the definition of quasi-linear utility to include the side payments and denote the agents' \emph{payoff} associated with an assignment as $\pi_i = v_i(\vecc{s}_i) - p_i - q_i$ for every buyer $i \in I$ and $\pi_j = p_j - v_j(\vecc{z}_j) - q_j$ for every seller $j \in J$.
Any agent's payoff thus consists of the value they gain for receiving, or lose from selling their cumulative bundle, and the total amount of money they send to or receive from other agents.
We will denote the vector of all payoffs as $\vecc{\pi}_N = (\pi_k)_{k\in N}$.

We will specify the different stability notions we use further below (see Section~\ref{sec:stability-notions}), but a common property and fundamental assumption is the individual rationality of every agent in the outcome, i.e., that they would drop out of the market if they received a negative payoff from the trade. For some of our results, it is therefore reasonable to restrict the set of assignments accordingly.

\begin{definition}[Feasible assignment]
  The set $\mathcal{F}(\mathcal{E}_m)$ of \emph{feasible assignments} comprises all assignments $((\vecc{S}_I, \vecc{Z}_J), \vecc p, \vecc{Q}_N)$ that are admissible in the market $\mathcal{E}_m$ and give non-negative payoff to all agents, $\vecc \pi_N \geq \vecc 0$.
\end{definition}

Because the market settings restrict the structure of the pricing system in specific ways, it might not be possible for every admissible allocation $(\vecc{S}_I, \vecc{Z}_J)$ to achieve a market outcome that is also individually rational for every agent.
While the allocations themselves might be used in different settings (under certain modifications), it might not be possible to do the same with the pricing system.

\begin{definition}[Priceability of an allocation]
  An allocation $(\vecc{S}_I, \vecc{Z}_J)$ is \emph{priceable} in the market $\mathcal{E}_m$ if it is admissible and if there exist an admissible pricing system and admissible side payments in $\mathcal{E}_m$ such that $((\vecc{S}_I, \vecc{Z}_J), \vecc p, \vecc{Q}_N) \in \mathcal{F}(\mathcal{E}_m)$.
\end{definition}

As already observed by \cite{BikhchandaniOstroy:ThePackageAssignmentModel}, we can easily deduce the following from the definitions of the first four settings:
\begin{remark}\label{rem:bikh_monotone}
  Because the repackaging restrictions increase in $h \in \{1,2,3,4\}$, it holds that if an allocation $(\vecc{S}_I, \vecc{Z}_J)$ is admissible in the $h^\mathrm{th}$-order setting, then it is also admissible in the $(h - 1)^\mathrm{th}$-order setting for every $h \in \{2, 3, 4\}$.

  Conversely, the pricing systems become more individualized in $h$, so if $\vecc p$ is an $h^\mathrm{th}$-order pricing system, it naturally induces a $(h + 1)^\mathrm{th}$-order pricing system for every $h \in \{1, 2, 3\}$ by applying the same price to all corresponding individual prices.
\end{remark}

\subsection{Side-Payment Schemes}\label{sec:sidePayments}
Independent of how the pricing system is restricted in the five settings, we also differentiate levels of side-payment restrictions in the market.
While the definition covers the extent of side payments among any pair of agents, we introduce the following notions and restrictions:

We denote by \emph{full side payments} the unrestricted setting, in which all agents can side-pay each other.

\emph{Separate side payments} refers to the restriction where the two sides of the market can send and receive side payments among each other, but all side payments between any seller and buyer are zero, i.e., $q_{ij} = 0$ for every $i\in I$ and $j\in J$.
This restriction establishes that all payments between the two sides go through the market pricing system and will allow us to differentiate between market settings where side payments can be equivalently realized by adjusting the pricing system accordingly.

\emph{Buyers'/sellers' single side payments} are an even more severely restricted setting, where we do allow one side of the market to send and receive money among each other, but not the other side.
This setting reflects, e.g., the situation of markets in which this type of collusion is prohibited for sellers, but this may not be necessary or enforceable for buyers.

Finally, we consider the setting with \emph{no side payments}, which represents the original market setting in the previous literature.

We use the notation $\mathcal{E}_h^r$ with $h \in \{1, \dots, 5\}$ and $r \in \{\mathrm{full}, \, \mathrm{sep}, \, \mathrm{buy}, \, \mathrm{sell}, \, \mathrm{no}\}$ to indicate the market $\mathcal{E}_h$ with the respective restriction $r$ on the side payments as outlined above.

The restrictions on the side payments imply a partial ordering of the sets of admissible assignments based on the sets of agents that can exchange side payments:
$ \mathcal{A}(\mathcal{E}_h^{\textrm{no}}) \subseteq \mathcal{A}(\mathcal{E}_h^{\textrm{buy}}) \subseteq \mathcal{A}(\mathcal{E}_h^{\textrm{sep}}) \subseteq \mathcal{A}(\mathcal{E}_h^{\textrm{full}}),$
and similarly,
$ \mathcal{A}(\mathcal{E}_h^{\textrm{no}}) \subseteq \mathcal{A}(\mathcal{E}_h^{\textrm{sell}}) \subseteq \mathcal{A}(\mathcal{E}_h^{\textrm{sep}}) \subseteq \mathcal{A}(\mathcal{E}_h^{\textrm{full}}).$
We use the notation $r \preceq \tilde r$ to refer to the corresponding partial order on the side-payment restrictions themselves when comparing them, where $\tilde r$ admits all side payments that $r$ does.

\subsection{The Role of the Clearinghouse}\label{sec:RoleClearinghouse}
At first glance, the clearinghouse appears to play a crucial role within the model. Indeed, it is responsible for collecting, manipulating, and reallocating the packages, as well as receiving the price payments and redistributing them according to the pricing system. However, a more detailed analysis shows that outside of establishing the structure of the common pricing system (which, of course, is a defining property in the expressiveness of each setting), its effect is typically somewhat limited.

In the fourth-order setting, the clearinghouse's only task is to act as a centralized control instance for the market, and it is given the least power across all settings: it only forwards each bundle without modification to the designated recipient, and also the payments $p_{ij}$ are already directly established between the corresponding agents. Effectively, the agents trade directly with each other.

On the other hand, both the first- and the fifth-order settings crucially rely on the clearinghouse to repackage the bundles, as the only restriction for market clearance is the total occurrence number of each item. Thus, all the trades and price payments are necessarily centralized.

The remaining second-order and third-order settings are based on forwarding bundles, which the clearinghouse can adjust in terms of the labels applied to the bundles.
In fact, the transition from the second-order to the third-order setting establishes a critical transition in the relation between the order settings, where the restrictive pricing system can be circumvented in some circumstances by mislabeling identical bundles, see Example~\ref{example:34noto2(M)}.
In other transitions between order settings, we will conversely show that the power of the clearinghouse is effectively only an application of some permutation of the bundles in the offer/request vector, and its effect is again very limited (Lemmas~\ref{lemma:1to234(M)} and \ref{lemma:3to4(M)}).

\subsection{Stability}\label{sec:stability-notions}
Combinatorial markets can be interpreted as a specific type of \emph{cooperative games}, where agents make binding agreements on which strategy they use and how the payoff gets distributed among them (\cite{PelegSudhoelter:IntroductionCooperativeGames}). For these games, standard solution concepts are stability in the sense of \emph{transferable utility (TU)}, where agents can freely distribute their total utility among each other, and \emph{non-transferable utility (NTU)}, where no payments outside the market are allowed. Both concepts share the solution concept of payoff vectors that admit coalitional stability, i.e., no coalition of agents can strictly increase their individual payoffs by deviating from the agreed-upon global solution. The set of these vectors is commonly referred to as the \emph{core} of the game.

We consider the concepts of TU and NTU as special cases of the following generalization: Let $\mathcal{T} = \{T_1, \dots, T_{|\mathcal{T}|}\}$ be a partition of the set of agents $N$. The agents in a given block $T$ of the partition can transfer money freely among themselves, but not to agents in a different block. We call this concept \emph{$\mathcal{T}$-restricted transferable utility ($\mathcal{T}$-RTU)}. The classical definitions of TU and NTU are then included in this representation by setting, respectively, $\mathcal{T} = \{N\}$ or $\mathcal{T} = \{\{k\} : k \in N\}$.
The notion of stability corresponding to $\mathcal{T}$-RTU is established by the following definition.

\begin{definition}[$\mathcal{T}$-core]\label{def:T-core}
  A vector $\vecc{\rho}_N \in \R^N$ is in the \emph{$\mathcal{T}$-core} $\mathcal{C}^{\mathcal{T}}(\mathcal{E}_h^r)$ of the market $\mathcal{E}_h^r$ if
  \begin{enumerate}
    \item there exists an assignment $((\vecc{S}_I, \vecc{Z}_J), \vecc{p}, \vecc{Q}_N) \in \mathcal{A}(\mathcal{E}_{h}^r)$ with $\sum_{k \in T} \rho_k = \sum_{k \in T} \pi_k$ for every $T \in \mathcal{T}$, and
    \item for every coalition $C \subseteq N$ and for every assignment $((\vecc{\tilde S}_{C_I}, \vecc{\tilde Z}_{C_J}), \vecc{\tilde p}, \vecc{\tilde Q}_C) \in \mathcal{A}(\mathcal{E}_{h}^r[C])$,
      $$\exists T \in \mathcal{T}:  \sum_{k \in C \cap T} \tilde\pi_k \leq \sum_{k \in C \cap T} \rho_k.$$
  \end{enumerate}
  We also call such a vector \emph{$\mathcal{T}$-stable} in the market $\mathcal{E}_h^r$.
\end{definition}

A vector in the core is typically also called a \emph{payoff vector} (due to the interpretation of the stability concept as its own game), a notion we have already used for the utility of each player in the assignment of the market setting. If we require disambiguation between the utility from the market itself and the vector we consider in the context of the core, we will refer to the latter as the \emph{payoff after utility transfers}.
One assumption necessary to avoid non-optimal stable outcomes is that the agents whose cumulative bundle is empty in an assignment that realizes a core vector and allows for free price readjustment have zero payoff.

The first condition in the definition establishes the transfer of utility within each block of the partition $\mathcal{T}$ among its agents.
The second condition can be interpreted in the following way: a vector is $\mathcal{T}$-stable if no coalition $C \subseteq N$ can find an assignment that improves the cumulative payoff for every agent group $C \cap T,\,T \in \mathcal{T}$.
Compare this condition to the standard notion of TU-stability, where this holds for the group of all agents, and to NTU-stability, where it applies to every agent individually.
A property that immediately follows from the definition is that every vector in the core is non-negative: for the singleton coalition $\{k\}$, that agent will receive a payoff of $\tilde\pi_k = 0$ by not trading anything.

For NTU-stability, every vector in $\mathcal{C}^{\mathcal{T}^\mathrm{no}}(\mathcal{E}_h^{r})$ needs to be exactly realized as the payoff vector of some assignment in the market $\mathcal{E}_h^{r}$.
We can therefore equivalently define the notion of stability directly for the respective assignment.

\begin{definition}[NTU-stable assignment]
  An assignment $((\vecc{S}_I, \vecc{Z}_J), \vecc p, \vecc{Q}_N) \in \mathcal{A}(\mathcal{E}_h^r)$ is \emph{NTU-stable} in the market $\mathcal{E}_h^r$ if for every coalition $C \subseteq N$ and for every assignment $((\vecc{\tilde S}_{C_I}, \vecc{\tilde Z}_{C_J}), \vecc{\tilde p}, \vecc{\tilde Q}_C) \in \mathcal{A}(\mathcal{E}_{h}^r[C])$, there exists an agent $k \in C$ such that $\tilde\pi_k \leq \pi_k$.
\end{definition}

\subsection{Relation between Stability Notions and Side-Payment Settings}\label{sec:stability_relations}
Intuitively, the transfer of utility we allow in different settings of $\mathcal{T}$-RTU stability corresponds to the side payments of some specific setting we have established above on a different level of the problem; the former applies the transfer after the market has established its trade, while the latter incorporates it into the strategic considerations within the market.
Indeed, we provide a formal discussion of this relation in this section.
In particular, the following partitions correspond to the side-payment restrictions we have already defined.
Whenever we refer to any of the special partitions defined here, we call the partition $\mathcal{T}$ a \emph{side-payment partition} and conversely emphasize by writing \emph{arbitrary} partition otherwise.
\begin{enumerate}
  \item \emph{TU-core} when $\mathcal{T}^{\mathrm{full}} = \{N\}$,
  \item \emph{sepRTU-core} when $\mathcal{T}^{\mathrm{sep}} = \{I, J\}$,
  \item \emph{bRTU-core} when $\mathcal{T}^{\mathrm{buy}} = \{I\} \cup \{\{j\} : j \in J\}$,
  \item \emph{sRTU-core} when $\mathcal{T}^{\mathrm{sell}} = \{\{i\} : i \in I\} \cup \{J\}$,
  \item \emph{NTU-core} when $\mathcal{T}^{\mathrm{no}} = \{\{k\} : k \in N\}$.
\end{enumerate}
We can consider the natural partial order on the previous partitions based on the refinement, i.e., $\mathcal{T}^{\mathrm{no}} \preceq \mathcal{T}^{\mathrm{buy}} \preceq \mathcal{T}^{\mathrm{sep}} \preceq \mathcal{T}^{\mathrm{full}}$, and similarly $\mathcal{T}^{\mathrm{no}} \preceq \mathcal{T}^{\mathrm{sell}} \preceq \mathcal{T}^{\mathrm{sep}} \preceq \mathcal{T}^{\mathrm{full}}$.
Clearly, this is the same partial order as for the side-payment settings, if we identify the partition with the respective setting it represents.

Following the intuition outlined above, we expect a close relation between an $r$ side-payment setting and the corresponding notion of $\mathcal{T}^r$-RTU stability in the sense that we can replace one by the other, i.e., the $\mathcal{T}^r$-core should give the same stable outcomes as the NTU-core for $r$ side payments.
We formalize this intuition in the following.
In fact, it turns out that an even stronger statement can be proved: When considering $\mathcal{T}$-RTU stability and the partition $\mathcal{T}^r$ that corresponds to the $r$ side-payment setting in the market, we can alternatively consider joining the partitions in the partition lattice to their common least upper bound $\mathcal{S} = \mathcal{T} \vee \mathcal{T}^r$.\footnote{We omit a formal discussion of this operator. For comparable elements, the greater one is the least upper bound. The only pair of incomparable partitions we consider is $\mathcal{T}^{\mathrm{buy}} \lor \mathcal{T}^{\mathrm{sell}} = \mathcal{T}^{\mathrm{sep}}$.}
Then, the original setting corresponds to both $\mathcal{S}$-RTU stability for the same economy without side payments, as well as NTU-stability on the market that allows side payments according to the partition $\mathcal{S}$.

\begin{lemma}\label{lem:TRTUwithr_is_SRTU}
  It holds that $\mathcal{C}^{\mathcal{T}}(\mathcal{E}_h^r) = \mathcal{C}^{\mathcal{S}}(\mathcal{E}_h^{\mathrm{no}}) $, where $\mathcal{S} = \mathcal{T} \vee \mathcal{T}^r$, when
  \begin{itemize}
    \item $\mathcal{T}$ is a side-payment partition, or
    \item $\mathcal{T}$ and $\mathcal{T}^r$ are comparable, i.e., $\mathcal{T} \preceq \mathcal{T}^r$ or $\mathcal{T}^r \preceq \mathcal{T}$.
  \end{itemize}
\end{lemma}

\begin{proof}
  This proof is a generalization of the one provided for Corollary~\ref{cor:RTUisNTUwithSP}. The distribution of the total payoff in each group is more intricate in this case, however, as we might need to use both side-payment edges as well as the groups induced by $\mathcal{T}$ to achieve the desired outcome. To simplify the notation, denote by $\mathcal{T}_S = \{T\in \mathcal{T}: T\subseteq S\}$.

  First, let $\vecc{\rho}_N\in \mathcal{C}^{\mathcal{T}}(\mathcal{E}_h^{r})$ be a vector of payoff after utility transfers.
  By definition, there exists an assignment $((\vecc{S}_I, \vecc{Z}_J), \vecc{p}, \vecc{Q}_N) \in \mathcal{A}(\mathcal{E}_{h}^{r})$ such that $\sum_{k \in T} \rho_k = \sum_{k \in T} \pi_k$ for any $T \in \mathcal{T}$. Because every block $S\in \mathcal{S}$ is either a superset of or disjoint to a given $r$ side-payment block, the total payoff remains the same when we remove the side payments from the assignment. Thus, for $((\vecc{S}_I, \vecc{Z}_J), \vecc{p}) \in \mathcal{A}(\mathcal{E}_{h}^{\mathrm{no}})$ and its payoff vector $\vecc\pi'_N$ it holds that
  \begin{equation}
    \sum_{k\in S}\rho_k = \sum_{T\in \mathcal{T}_S} \sum_{k\in T}\rho_k = \sum_{T\in \mathcal{T}_S} \sum_{k\in T}\pi_k = \sum_{k\in S}\pi_k = \sum_{T^r\in \mathcal{T}^r_S} \sum_{k\in T^r}\pi_k = \sum_{T^r\in \mathcal{T}^r_S} \sum_{k\in T^r} \pi_k', \label{eq:partition_sidepayments}
  \end{equation}
  which establishes the first condition in the definition of $\vecc{\rho}_N\in \mathcal{C}^{\mathcal{S}}(\mathcal{E}_h^{\mathrm{no}})$. For the stability condition, suppose conversely that $ \vecc\rho_N $ is \emph{not} stable in $\mathcal{C}^{\mathcal{S}}(\mathcal{E}_h^\mathrm{no})$, i.e., there exists a coalition $C$ and an assignment $((\vecc{\tilde S}_{C_I}, \vecc{\tilde Z}_{C_J}), \vecc{\tilde p}) \in \mathcal{A}(\mathcal{E}_{h}^{\mathrm{no}}[C])$ with $\sum_{k\in C\cap S}\tilde\pi_k > \sum_{k\in C\cap S}\rho_k$ for all $S\in \mathcal{S}$. Using $r$ side payments, we now construct a deviating coalition that contradicts the stability of $\vecc{\rho}_N\in \mathcal{C}^{\mathcal{T}}(\mathcal{E}_h^{r})$.
  \begin{itemize}
    \item If $\mathcal{T}^r \preceq \mathcal{T} = \mathcal{S}$, the given deviation already violates the stability.
    \item If $\mathcal{T} \preceq \mathcal{T}^r = \mathcal{S}$, the agents can freely distribute the side payments within the full blocks of $\mathcal{S}$. As $\sum_{k\in C\cap S}\tilde\pi_k = \sum_{T\in \mathcal{T}_S}\sum_{k\in C\cap T}\tilde\pi_k$, we can easily distribute the excess payoff in every $S$ such that $\sum_{k\in C\cap T}\rho_k$ is exceeded for every block $T\in \mathcal{T}_S$.
    \item The remaining case is $\mathcal{S} = \{I, J\}$ and $r = \mathrm{buy}, \mathcal{T} = \mathcal{T}^{\mathrm{sell}}$ or $r = \mathrm{sell}, \mathcal{T} = \mathcal{T}^{\mathrm{buy}}$. In this case, we can establish an assignment with $ \pi'_k > \rho_k $ for all agents $k$ on the side-payment side by distributing sufficient excess among the agents through side payments. Those $k$ correspond to the singleton blocks in $\mathcal{T}$. The other side, given by a single block of agents, already fulfills the inequality.
  \end{itemize}

  Conversely, we take $\vecc{\rho}_N\in \mathcal{C}^{\mathcal{S}}(\mathcal{E}_h^{\mathrm{no}})$. To verify the first condition in $\mathcal{C}^{\mathcal{T}}(\mathcal{E}_h^{r})$, we go through the same case distinctions as before for the assignment $((\vecc{S}_I, \vecc{Z}_J), \vecc{p}) \in \mathcal{A}(\mathcal{E}_{h}^{\mathrm{no}})$ that provides $\vecc\rho_N$:
  \begin{itemize}
    \item If $\mathcal{T}^r \preceq \mathcal{T} = \mathcal{S}$, the condition for the partition blocks is trivially fulfilled and $((\vecc{S}_I, \vecc{Z}_J), \vecc{p}) \in \mathcal{A}(\mathcal{E}_{h}^{r})$  holds.
    \item If $\mathcal{T} \preceq \mathcal{T}^r = \mathcal{S}$, the agents can freely distribute the side payments within the full blocks of $\mathcal{S}$. Thus, they also manage to distribute the value to the blocks of $\mathcal{T}$ accordingly, as $\mathcal{T}_S$ is a partition of $S$ and therefore the cumulative payoff in $S$ corresponds to the sum of blocks in $\mathcal{T}_S$ for every $S\in \mathcal{S}$.
    \item The remaining case is $\mathcal{S} = \{I, J\}$ and $r = \mathrm{buy}, \mathcal{T} = \mathcal{T}^{\mathrm{sell}}$ or $r = \mathrm{sell}, \mathcal{T} = \mathcal{T}^{\mathrm{buy}}$. On the side-payment side, we can adjust the payoffs through side-payments to match the corresponding $\rho_k$. On the other side, the cumulative payoff is the only requirement and is unchanged.
  \end{itemize}

  A supposed deviation of coalition $C$ with assignment $((\vecc{\tilde S}_{C_I}, \vecc{\tilde Z}_{C_J}), \vecc{\tilde p}, \vecc{\tilde Q}_N) \in \mathcal{A}(\mathcal{E}_{h}^{r}[C])$ satisfies $\sum_{k\in C\cap T}\tilde\pi_k > \sum_{k\in C \cap T}\rho_k$ for all $T\in \mathcal{T}$.
  Because $\mathcal{T} \preceq \mathcal{S}$ and all side-payment blocks are fully contained in or disjoint to a given $S$ (compare Equation~(\ref{eq:partition_sidepayments})), we simply sum up the inequalities of $\mathcal{T}_S$ for all $S\in\mathcal{S}$ to see that the assignment $((\vecc{\tilde S}_{C_I}, \vecc{\tilde Z}_{C_J}), \vecc{\tilde p}) \in \mathcal{A}(\mathcal{E}_{h}^{\mathrm{no}}[C])$ provides a deviation in contradiction to $\vecc\rho_N \in \mathcal{C}^{\mathcal{S}}(\mathcal{E}_h^{\mathrm{no}})$.
\end{proof}

Using $\mathcal{T} = \mathcal{T}^\mathrm{no} \preceq \mathcal{T}^r$ and $\mathcal{S} = \mathcal{T}^\mathrm{no} \vee \mathcal{T}^r = \mathcal{T}^r$ for some side-payment restriction $r$ gives the following corollary.

\begin{corollary}\label{cor:RTUisNTUwithSP}
  $\mathcal{C}^{\mathcal{T}^{\mathrm{no}}}(\mathcal{E}_h^r) = \mathcal{C}^{\mathcal{T}^r}(\mathcal{E}_h^{\mathrm{no}}) $.
\end{corollary}

\begin{proof}
  We provide the following proof as a simplified version of the proof of Lemma~\ref{lem:TRTUwithr_is_SRTU}.

  Both inclusions are proved following the same structure: First, we show that a given vector of payoffs after utility transfers $\vecc\rho_N$ admits some assignment to provide its existence in the other core by adjusting the side payments.
  Then, a violation of the stability condition would contradict the stability of $\vecc\rho_N$ in the other core as well.

  Consider a vector $\vecc{\rho}_N\in \mathcal{C}^{\mathcal{T}^r}(\mathcal{E}_h^{\mathrm{no}})$. By definition, there exists an assignment $((\vecc{S}_I, \vecc{Z}_J), \vecc{p}) \in \mathcal{A}(\mathcal{E}_{h}^{\mathrm{no}})$ such that $\sum_{k \in T} \rho_k = \sum_{k \in T} \pi_k$ for any $T \in \mathcal{T}^r$. Because for any fixed $T$, the total utility in each block, $\sum_{k \in T} \pi_k$, can be freely distributed using side payments in $\mathcal{E}_{h}^{r}$, we can find side payments $\vecc{Q}_N$ such that the payoff $\vecc\pi'_N$ of assignment $((\vecc{S}_I, \vecc{Z}_J), \vecc{p}, \vecc Q_N)\in \mathcal{A}(\mathcal{E}_{h}^{r})$ satisfies $\pi'_k = \rho_k$ for any $k\in N$. Thus, the first condition in the definition of $\mathcal{C}^{\mathcal{T}^{\mathrm{no}}}(\mathcal{E}_h^r)$ is established for $\vecc{\rho}_N$.
  For the second condition, suppose conversely that $ \vecc\rho_N $ is \emph{not} stable in $\mathcal{C}^{\mathcal{T}^{\mathrm{no}}}(\mathcal{E}_h^r)$, i.e., there exists a coalition $C$ and an assignment $((\vecc{\tilde S}_{C_I}, \vecc{\tilde Z}_{C_J}), \vecc{\tilde p}, \vecc{\tilde Q}_C) \in \mathcal{A}(\mathcal{E}_{h}^{r}[C])$ with $\tilde\pi_k > \rho_k$ for every $k\in C$. Consider the same allocation and prices without the side payments $((\vecc{\tilde S}_{C_I}, \vecc{\tilde Z}_{C_J}), \vecc{\tilde p})\in \mathcal{A}(\mathcal{E}_{h}^{\mathrm{no}}[C])$. Its payoff vector $ \vecc{\tilde\pi}'_C $ satisfies $\tilde\pi'_k = \tilde\pi_k + q_k$ for all $k\in C$. Because $\sum_{k\in C\cap T} q_k = 0$ for any $T \in \mathcal{T}^r$ due to the anti-symmetry of the side payments, it holds that
  \[ \sum_{k\in C\cap T}\tilde\pi'_k = \sum_{k\in C\cap T}\tilde\pi_k + q_k = \sum_{k\in C\cap T}\tilde\pi_k > \sum_{k\in C\cap T}\rho_k, \quad \forall T\in \mathcal{T}^r, \]
  which contradicts the stability of $\vecc\rho_N\in \mathcal{C}^{\mathcal{T}^r}(\mathcal{E}_h^{\mathrm{no}})$.

  For the opposite inclusion, consider $\vecc{\rho}_N\in \mathcal{C}^{\mathcal{T}^{\mathrm{no}}}(\mathcal{E}_h^{r})$ which is provided as payoff \emph{without} utility transfers by some assignment $((\vecc{S}_I, \vecc{Z}_J), \vecc{p}, \vecc{Q}_N) \in \mathcal{A}(\mathcal{E}_{h}^{r})$. By removing the side payments from this assignment, we get an assignment $((\vecc{S}_I, \vecc{Z}_J), \vecc{p}) \in \mathcal{A}(\mathcal{E}_{h}^{\mathrm{no}})$, whose payoff vector satisfies $\sum_{k \in T} \rho_k = \sum_{k \in T} \pi_k$, as the side payments distribute the block's total payoff only among its agents. Thus, the first condition in the definition of $\mathcal{C}^{\mathcal{T}^r}(\mathcal{E}_h^{\mathrm{no}})$ is fulfilled. If this vector were not stable, we can find a deviating assignment $((\vecc{\tilde S}_{C_I}, \vecc{\tilde Z}_{C_J}), \vecc{\tilde p}) \in \mathcal{A}(\mathcal{E}_{h}^{\mathrm{no}}[C])$ such that $\sum_{k\in T\cap C}\tilde\pi_k > \sum_{k\in T\cap C}\rho_k $ for every $T\in \mathcal{T}^r$. Because this total payoff in $T\cap C$ can be redistributed within this set freely using $r$ side payments, we could find an assignment $((\vecc{\tilde S}_{C_I}, \vecc{\tilde Z}_{C_J}), \vecc{\tilde p}, \vecc{Q}'_C) \in \mathcal{A}(\mathcal{E}_{h}^{r}[C])$ with payoff vector $\vecc{\tilde\pi'}$ with $\tilde\pi'_k > \rho_k$ for every $k \in C$, contradicting the stability of $\vecc{\rho}_N$ in $\mathcal{C}^{\mathcal{T}^{\mathrm{no}}}(\mathcal{E}_h^{r})$.
\end{proof}

Combining Lemma~\ref{lem:TRTUwithr_is_SRTU} and Corollary~\ref{cor:RTUisNTUwithSP} establishes Theorem~\ref{thm:allisNTU}.
The previous results let us conclude the discussion on stability settings in the following way: By using side payments, every $\mathcal{T}$-RTU stability notion can be equivalently represented by allowing a certain type of side payments in the market and discussing NTU-stability only, or by discussing the appropriate RTU-stability in the market without any side payments.
The stability setting itself is then identified by the specific restrictions on side payments.
Most importantly, TU-stability immediately corresponds to NTU-stability in the market with the full side payment setting.

\section{Equivalences}\label{sec:equivalences}

With the introduction of side payments, the relation between the different order settings becomes more involved.
Originally, the pricing systems provided defining restrictions on the structure of the payoff outcomes of the problem.
In this section, we investigate how the market settings actually differ from each other after the introduction of side payments, establishing various one-directional results by identifying situations where a market setting can represent all outcomes of another setting, and conversely provide examples when a certain setting's restrictions do not allow another setting's outcomes.

The combination of the results from Lemmas~\ref{lem:increasing_order} and \ref{lemma:12345to1(Pi)FullSep} to \ref{lemma:245to3(Pi)Sel} below immediately establishes Theorem~\ref{thm:setting-equivalences}, which states the full equivalence between certain settings in the sense that we can always find assignments using the same items and each agent achieves the exact same utility from the total trade.
See Appendix~\ref{secapp:tables} for a structured overview of the equivalences we establish.

\subsection{Types of Equivalences}\label{sec:formal_equivs}

A fundamental requirement for comparing two market settings is whether we can even find an allocation in which each agent offers or requests the same total amount of items.

\begin{definition}[Item-equivalent allocations]
  The allocations $(\vecc{S}_I, \vecc{Z}_J)$ and $(\vecc{\tilde{S}}_I, \vecc{\tilde{Z}}_J)$ are \emph{item equivalent} if $\vecc{s}_i = \vecc{\tilde{s}}_i$ for every buyer $i \in I$ and $\vecc{z}_j = \vecc{\tilde{z}}_j$ for every seller $j \in J$, i.e., each agent offers or requests the same cumulative bundle in both allocations.
\end{definition}

Already \cite[Lemma 2.1]{BikhchandaniOstroy:ThePackageAssignmentModel} have established that when an allocation is admissible in one of the first four order settings, we can always find an item-equivalent allocation in any of the other settings by greedily reassigning and repackaging the items so as to obtain a fourth-order admissible allocation (which in turn is trivially admissible in all lower order settings). This result can be easily generalized to the fifth-order setting by applying the same reasoning.

\begin{lemma}[extending \cite{BikhchandaniOstroy:ThePackageAssignmentModel}]\label{lemma:12345to4(I)Allocation}
  Let $(\vecc{S}_I, \vecc{Z}_J)$ be an $h^{\mathrm{th}}$-order admissible allocation, then there exists an item-equivalent $\tilde h^{\mathrm{th}}$-order admissible allocation $(\vecc{\tilde{S}}_I, \vecc{\tilde{Z}}_J)$, for every $h, \tilde h \in \{1, \dots, 5\}$.
\end{lemma}

In particular, due to budget balance, this implies that the maximum social welfare of assignments in a market depends only on the economy, not on the resale restrictions we impose on it.

We now define the following equivalence properties between two assignments. Note that we will always assume that the assignments are defined in the same economy, but they can (and typically will) be admissible in different market settings.

\begin{definition}[Item equivalence]
  The assignments $((\vecc{S}_I, \vecc{Z}_J), \vecc p, \vecc{Q}_N)$ and $((\vecc{\tilde S}_I, \vecc{\tilde Z}_J), \vecc{\tilde p}, \vecc{\tilde Q}_N)$ are \emph{item equivalent} if their respective allocations are item equivalent.
\end{definition}

While this property is fundamental and required for any reasonable comparison between two market settings, it is far from sufficient to discuss any meaningful properties about the strategic outcome of the market, as it does not relate the prices and payoffs of the assignments in any way.

However, even when the pricing systems differ, it might still be possible for each agent to achieve the same utility from the same cumulative bundle in another setting:

\begin{definition}[Payoff equivalence]
  The assignments $((\vecc{S}_I, \vecc{Z}_J), \vecc p, \vecc{Q}_N)$ and $((\vecc{\tilde S}_I, \vecc{\tilde Z}_J), \vecc{\tilde p}, \vecc{\tilde Q}_N)$ are \emph{payoff equivalent} if they are item equivalent and $\pi_k = \tilde{\pi}_k$ for every $k \in N$, i.e., every agent individually has the same payoff in both assignments.
\end{definition}

A slightly stronger notion of equivalence is the situation where each agent also pays the same total amount of money to the clearinghouse.
Note that this does not necessarily mean that the prices of the traded bundles coincide.

\begin{definition}[Price equivalence]
  The assignments $((\vecc{S}_I, \vecc{Z}_J), \vecc p, \vecc{Q}_N)$ and $((\vecc{\tilde S}_I, \vecc{\tilde Z}_J), \vecc{\tilde p}, \vecc{\tilde Q}_N)$ are \emph{price equivalent} if they are payoff equivalent and $p_k = \tilde{p}_k$ for every $k \in N$, i.e., every agent individually pays or receives the same cumulative price in both assignments.
\end{definition}

And finally, we introduce the strongest notion, in which the assignments realize essentially the same trade, all the way down to the individual prices paid for every bundle.
Because the representation of the pricing system for the fifth-order setting differs from the other settings in that it considers cumulative prices directly, the following notion applies to it only when both assignments are fifth-order admissible.

\begin{definition}[Market equivalence]
  The assignments $((\vecc{S}_I, \vecc{Z}_J), \vecc p, \vecc{Q}_N)$ and $((\vecc{\tilde S}_I, \vecc{\tilde Z}_J), \vecc{\tilde p}, \vecc{\tilde Q}_N)$ are \emph{market equivalent} if they are price equivalent and $p_{ij}(\vecc{\omega}) = \tilde{p}_{ij}(\vecc{\omega})$ for every $\vecc{\omega} \leq \vecc{\Omega}$, $i \in I$ and $j \in J$, i.e., the prices are the same for each bundle, buyer and seller.\footnote{For a pair of fifth-order admissible assignments, the condition on the pricing systems becomes $p_k(\vecc{\omega}) = \tilde{p}_{k}(\vecc{\omega})$ for every $\vecc{\omega} \leq \vecc{\Omega}$ and $k \in N$.}
\end{definition}

\subsection{Increasing Order Setting}\label{sec:increasing_order}
The first results we present consider going from a given order setting to a higher one.
Thereby, the pricing system becomes more individualized, but the clearinghouse gets more restricted in terms of how it can reassemble or reassign packages.
While the former is clearly relaxing the restrictions for higher settings, it turns out that the latter is in most cases only a formal issue, which does not impede equivalence properties between the settings.

We start by considering the cases in which, going from one setting to the other, the side-payment restrictions remain the same or are relaxed. If the order setting does not change, market equivalence can be directly established.

\begin{remark}\label{remark:marketEquivDiagonal}
  As we have already observed in Section~\ref{sec:sidePayments}, any assignment $((\vecc{S}_I, \vecc{Z}_J), \vecc p, \vecc{Q}_N) \in \mathcal{F}(\mathcal{E}_h^r)$ is also feasible in a market with the same order setting but with fewer side-payment restrictions, i.e., $((\vecc{S}_I, \vecc{Z}_J), \vecc p, \vecc{Q}_N) \in \mathcal{F}(\mathcal{E}_h^{\tilde r})$ for every $\tilde r \succeq r$. Therefore, we can immediately deduce that a market-equivalent assignment always exists when the only change to the setting is a relaxation of the side-payment restrictions.
\end{remark}

We summarize our results for increasing the order setting, when the side-payment restrictions are not refined, as follows:

\begin{lemma}\label{lem:increasing_order}
  When the side-payments of the original setting remain directly applicable, i.e., $\tilde r \succeq r$, the following equivalences hold:
  \begin{itemize}
    \item For ${\vecc{A}} \in \mathcal{F}(\mathcal{E}_1^r)$, there exists a market-equivalent assignment ${\vecc{\tilde{A}}} \in \mathcal{F}(\mathcal{E}_{\tilde h}^{\tilde r})$ with $\tilde h \in \{2, 3, 4\}$.
    \item For ${\vecc{A}} \in \mathcal{F}(\mathcal{E}_3^r)$, there exists a market-equivalent assignment ${\vecc{\tilde{A}}} \in \mathcal{F}(\mathcal{E}_{4}^{\tilde r})$.

    \item For ${\vecc{A}} \in \mathcal{F}(\mathcal{E}_h^r)$ with $h \in \{1, \dots, 4\}$, there exists a price-equivalent assignment ${\vecc{\tilde{A}}} \in \mathcal{F}(\mathcal{E}_5^{\tilde r})$.

    \item For ${\vecc{A}} \in \mathcal{F}(\mathcal{E}_2^r)$, there exists a price-equivalent assignment ${\vecc{\tilde{A}}} \in \mathcal{F}(\mathcal{E}_4^{\tilde r})$.\footnote{Note that for the fifth-order setting, the notion of market-equivalence doesn't apply due to the difference in representation of this model.}
  \end{itemize}
\end{lemma}

The proofs of these results are all based on identifying an appropriate pricing system in the higher-order setting to achieve the same market outcome as before, while identifying the required structure to compile admissible packages with restricted clearinghouse power.
In the following, we present the individual statements and their proofs.

In the first-order setting, the linearity of the pricing system implies that all item-equivalent assignments with the same market prices also have the same cumulative price for each agent. Hence, since by Lemma~\ref{lemma:12345to4(I)Allocation} we can always find an item-equivalent allocation in any setting, and first-order pricing systems are admissible in any of the first four order settings, we can derive the following result.

\begin{lemma}\label{lemma:1to234(M)}
  Let $((\vecc{S}_I, \vecc{Z}_J), \vecc p, \vecc{Q}_N) \in \mathcal{F}(\mathcal{E}_1^r)$. For $\tilde h \in \{2, 3, 4\}$ and $\tilde r \succeq r$, there exists a market-equivalent assignment $((\vecc{\tilde S}_I, \vecc{\tilde Z}_J), \vecc{\tilde p}, \vecc{\tilde Q}_N) \in \mathcal{F}(\mathcal{E}_{\tilde h}^{\tilde r})$.
\end{lemma}

\begin{proof}
  Starting with a feasible first-order admissible assignment $((\vecc{S}_I, \vecc{Z}_J), \vecc p, \vecc{Q}_N)$, we only need to take care of the allocation, as the more verbose pricing system in the $\tilde h^{\mathrm{th}}$-order can accommodate our requirements by setting each price to $\tilde{p}_{ij}(\vecc{\omega}) = p_{ij}(\vecc{\omega}) = \vecc{p} \cdot \vecc{\omega}$ for every $\vecc{\omega} \leq \vecc{\Omega}$, $i \in I$ and $j \in J$.
  From Lemma~\ref{lemma:12345to4(I)Allocation}, we can consider an $\tilde h^{\mathrm{th}}$-order admissible allocation $(\vecc{\tilde S}_I, \vecc{\tilde Z}_J)$ that is item equivalent to $(\vecc{S}_I, \vecc{Z}_J)$, but whose individual bundles might differ. However, since $\vecc{\tilde p}$ is linear, the cumulative prices remain the same, i.e., for every buyer $i \in I$,
  $$\tilde{p}_i = \sum_{j \in J} \tilde{p}_{ij}(\vecc{\tilde{s}}_{ij}) = \sum_{j \in J} \vecc{p} \cdot \vecc{\tilde{s}}_{ij} = \vecc{p} \cdot \vecc{\tilde{s}}_i = \vecc{p} \cdot \vecc{s}_i = \sum_{j \in J} \vecc{p} \cdot \vecc{s}_{ij} = p_i,$$
  and similarly, $\tilde p_j = p_j$ for every seller $j \in J$. By maintaining the same individual side payments $\vecc{\tilde Q}_N = \vecc{Q}_N$, we finally establish $\tilde \pi_k = \pi_k$ for every agent $k \in N$.
  Therefore, we conclude that $((\vecc{\tilde S}_I, \vecc{\tilde Z}_J), \vecc{\tilde p}, \vecc{\tilde Q}_N)$ is market equivalent to $((\vecc{S}_I, \vecc{Z}_J), \vecc p, \vecc{Q}_N)$.
\end{proof}

A similar result holds for the third-order setting. In this case, the effect of the clearinghouse is to permute the order of the buyers' individual bundles while leaving the sellers' unaffected. Hence, by rearranging the buyers' requests so that they correspond to the sellers' offers, and maintaining the same pricing system, we can deduce the following.

\begin{lemma}\label{lemma:3to4(M)}
  Let $((\vecc{S}_I, \vecc{Z}_J), \vecc p, \vecc{Q}_N) \in \mathcal{F}(\mathcal{E}_3^r)$. For $\tilde r \succeq r$, there exists a market-equivalent assignment $((\vecc{\tilde S}_I, \vecc{\tilde Z}_J), \vecc{\tilde p}, \vecc{\tilde Q}_N) \in \mathcal{F}(\mathcal{E}_{4}^{\tilde r})$.
\end{lemma}

\begin{proof}
  As in Lemma~\ref{lemma:1to234(M)}, we only need to adjust the allocation, since we can define the new pricing system simply as $\tilde{p}_{ij}(\vecc{\omega}) = p_{ij}(\vecc{\omega}) = p_{i}(\vecc{\omega})$ for every $\vecc{\omega} \leq \vecc{\Omega}$, $i \in I$ and $j \in J$, and keep the same individual side payments $\vecc{\tilde Q}_N = \vecc{Q}_N$.

  We remind that in the third-order setting, the clearinghouse can change the seller's label but not the buyer's label on each bundle it receives. Therefore, every seller $j \in J$ must offer each individual bundle $\vecc{z}_{ij}$ to a buyer $i \in I$ who actually requests it, whereas for the buyers we might have that some bundles $\vecc{s}_{ij} = \vecc{z}_{ij'}$ are offered to them by a seller $j' \in J$ that is different from the seller $j$ indicated on the label. Since each buyer $i$ receives exactly one bundle from every seller, the buyer's allocation is effectively a permutation of the bundles offered to them, i.e.,
  $$\vecc{S}_i = (\vecc{s}_{ij_1}, \dots, \vecc{s}_{ij_{|J|}}) = (\vecc{z}_{i\sigma(j_1)}, \dots, \vecc{z}_{i\sigma(j_{|J|})})$$
  for some permutation $\sigma \in S_{|J|}$. By rearranging the requests of the buyers so that every $i \in I$ asks each individual bundle from the seller who offers it to buyer $i$ (i.e., $\vecc{\tilde{s}}_{ij} = \vecc z_{ij} = \vecc s_{ij'}$ for every $j \in J$), and leaving the allocation of each seller $j \in J$ unaffected (i.e., $\vecc{\tilde{z}}_{ij} = \vecc z_{ij}$ for every $i \in I$), we get a new allocation $(\vecc{\tilde S}_I, \vecc{\tilde Z}_J)$ that is fourth-order admissible and item equivalent to $(\vecc{S}_I, \vecc{Z}_J)$.

  Since the set of individual bundles that each agent receives is the same in both allocations and $\tilde{p}_{ij}(\vecc{\tilde s}_{ij}) = \tilde{p}_{ij}(\vecc{\tilde z}_{ij}) = p_{ij}(\vecc{z}_{ij})$ for every $i \in I$ and $j \in J$, the cumulative price is also the same for each agent.
  We can then conclude that $((\vecc{\tilde S}_I, \vecc{\tilde Z}_J), \vecc{\tilde p}, \vecc{\tilde Q}_N)$ is market equivalent to the initial assignment.
\end{proof}

In the fifth-order setting, allocations and prices are defined cumulatively, and the only conditions they must satisfy are, respectively, that only the requested items are offered and that budget balance holds. But these conditions are true for all admissible assignments in any of the settings, so by considering the cumulative equivalent of any assignment, we can always obtain a fifth-order admissible assignment in the same or less restricted side-payment setting. Analogously to above, we can then deduce the following result. Note that market equivalence is not well-defined between assignments in the fifth-order setting and assignments in the other four order settings.

\begin{lemma}\label{lemma:1234to5(P)}
  Let $((\vecc{S}_I, \vecc{Z}_J), \vecc p, \vecc{Q}_N) \in \mathcal{F}(\mathcal{E}_h^r)$ with $h \in \{1, \dots, 4\}$, then there exists a price-equivalent assignment $((\vecc{\tilde s}_I, \vecc{\tilde z}_J), \vecc{\tilde p}, \vecc{\tilde Q}_N) \in \mathcal{F}(\mathcal{E}_5^{\tilde r})$ with $\tilde r \succeq r$.
\end{lemma}

\begin{proof}
  Starting from assignment $((\vecc{S}_I, \vecc{Z}_J), \vecc p, \vecc{Q}_N) \in \mathcal{F}(\mathcal{E}_h^r)$ in one of the first four order settings, define the assignment $((\vecc{\tilde s}_I, \vecc{\tilde z}_J), \vecc{\tilde p}, \vecc{\tilde Q}_N)$ so that:
  \begin{enumerate}
    \item the allocation is $(\vecc{\tilde{s}}_I, \vecc{\tilde{z}}_J) = (\vecc{s}_I, \vecc{z}_J)$ (the cumulative bundles),
    \item the pricing system $\vecc{\tilde{p}}$ is such that $\tilde{p}_i(\vecc{\tilde{s}}_i) = \tilde p_i = p_i$ for every $i \in I$ and $\tilde{p}_j(\vecc{\tilde{z}}_j) = \tilde p_j = p_j$ for every $j \in J$ (the prices of the non-traded bundles can be set arbitrarily), and
    \item the side payments are $\vecc{\tilde Q}_N = \vecc{Q}_N$.
  \end{enumerate}
  Then, $((\vecc{\tilde s}_I, \vecc{\tilde z}_J), \vecc{\tilde p}, \vecc{\tilde Q}_N) \in \mathcal{F}(\mathcal{E}_5^{\tilde r})$ for every $\tilde r \succeq r$, and it is price equivalent to the initial assignment.
\end{proof}

For the second-order setting, the situation is different. In particular, no market equivalence result holds because the relabeling of the clearinghouse on both sides of the market allows some buyer-seller pairs to effectively trade more than one individual bundle and, thereby, circumvent the price restrictions that would usually apply to the cumulative bundles (see Example~\ref{example:34noto2(M)} for the full discussion of this behavior, including further negative results on payoff and price equivalence to higher-order settings).

However, in the last positive result of this section, we prove that while market equivalence cannot be established, price-equivalent assignments for second-order assignments can still be found in the fourth-order setting.
In this case, we can obtain an item-equivalent allocation by having each buyer-seller pair trade all the individual bundles they effectively exchange after relabeling as a single bundle. If we define the new pricing system by setting the price of each traded bundle to the sum of the initial prices of the corresponding individual bundles, we obtain an assignment that is price equivalent to the first one.

\begin{lemma}\label{lemma:2to4(P)}
  Let $((\vecc{S}_I, \vecc{Z}_J), \vecc p, \vecc{Q}_N) \in \mathcal{F}(\mathcal{E}_2^r)$. For $\tilde r \succeq r$, there exists a price-equivalent assignment $((\vecc{\tilde S}_I, \vecc{\tilde Z}_J), \vecc{\tilde p}, \vecc{\tilde Q}_N) \in \mathcal{F}(\mathcal{E}_4^{\tilde r})$.
\end{lemma}

\begin{proof}
  In the second-order setting, the clearinghouse can alter both the buyer's and the seller's labels of each bundle. Therefore, we can have a situation where the bundle $\vecc{z}_{ij}$ offered by seller $j \in J$ to buyer $i\in I$ is requested by some other buyer $i' \in I$ from some different seller $j' \in J$, i.e., $\vecc z_{ij} = \vecc s_{i'j'}$. This can, in general, lead to some buyer-seller pairs effectively exchanging more than a single bundle, as we have observed in Example~\ref{example:34noto2(M)}.

  Let then $((\vecc{S}_I, \vecc{Z}_J), \vecc p, \vecc{Q}_N) \in \mathcal{F}(\mathcal{E}_2^r)$. For every buyer-seller pair $(i, j)$, consider the (possibly empty) individual bundles $\vecc\omega^1_{ij}, \dots, \vecc\omega^{m_{ij}}_{ij}$ that they exchange after relabeling (so that for every $l \in \{1, \dots, m_{ij}\}$, $\vecc\omega^l_{ij} = \vecc z_{i'j} = \vecc{s}_{ij'}$ for some unique $i' \in I$ and $j' \in J$), and define their union as
  \begin{equation*}
    \vecc{\tilde s}_{ij} = \vecc{\tilde z}_{ij} = \sum_{l = 1}^{m_{ij}} \vecc\omega^l_{ij}.
  \end{equation*}
  The allocation $(\vecc{\tilde S}_I, \vecc{\tilde Z}_J)$, where each buyer-seller pair $(i, j)$ trades the bundle defined above, is fourth-order feasible and item equivalent to the initial one.

  We can define the related fourth-order feasible pricing system $\vecc{\tilde p}$ as
  $$\tilde{p}_{ij}(\vecc{\tilde s}_{ij}) = \tilde{p}_{ij}(\vecc{\tilde z}_{ij}) = \sum_{l = 1}^{m_{ij}} p(\vecc\omega^l_{ij}),$$
  for every $i \in I$ and $j \in J$.\footnote{In general, $\sum_{l = 1}^{m_{ij}} p(\vecc\omega^l_{ij}) \neq p\left(\sum_{l = 1}^{m_{ij}} \vecc\omega^l_{ij}\right)$ since second-order pricing systems are not necessarily linear. Hence, we can have that for some traded bundle $\vecc{\tilde s}_{ij}$, $\tilde{p}_{ij}(\vecc{\tilde s}_{ij}) \neq p(\vecc{\tilde s}_{ij}) = p_{ij}(\vecc{\tilde s}_{ij})$ and market equivalence is not guaranteed.} For the bundles that are not traded, we can set the prices arbitrarily, for example, $\tilde p_{ij}(\vecc\omega) = p(\vecc\omega)$ for every $\vecc\omega \neq \vecc{\tilde s}_{ij}$, $i \in I$ and $j \in J$.
  Then, the cumulative price of every buyer $i \in I$ is such that
  $$\tilde{p}_i = \sum_{j \in J} \tilde{p}_{ij}(\vecc{\tilde s}_{ij}) = \sum_{j \in J} \sum_{l = 1}^{m_{ij}} p(\vecc\omega^l_{ij}) = \sum_{j \in J} p(\vecc{s}_{ij}) = p_i,$$
  and analogously for the sellers.
  If we maintain the same individual side payments $\vecc{\tilde Q}_N = \vecc{Q}_N$, the payoff of each agent $k \in N$ is $\tilde\pi_k = \pi_k$. We can then conclude that $((\vecc{\tilde S}_I, \vecc{\tilde Z}_J), \vecc{\tilde p}, \vecc{\tilde Q}_N)$ is price equivalent to the initial assignment.
\end{proof}

If we restrict the side payments or change the side of the market where they are allowed, we can't establish the same rich set of equivalence implications. In fact, in most cases, not even payoff equivalence is guaranteed. See Examples~\ref{example:34noto2(M)} and \ref{example:FullnottoSep(P)} to \ref{example:1SellSepFullnotto...} for thorough discussions.

\subsection{Full and Separate Side Payments}\label{sec:equiv_fullside}
The full side-payment setting is the least restrictive one, and every possible distribution of social welfare in a trade is realizable.
As a consequence, for any assignment, we can determine a payoff equivalent outcome in any order setting with full side payments by having all money exchanges go through side payments and setting the prices to zero.

This fact can also be generalized to the separate side-payment setting. In this case, we can transfer the money from the buyers to the sellers through the price of a single traded item (while setting the prices of all the other items to zero), and then redistribute the money separately among the two groups so that each agent still attains the same payoff.
The following statement summarizes these observations.

\begin{lemma}\label{lemma:12345to1(Pi)FullSep}
  Let $((\vecc{S}_I, \vecc{Z}_J), \vecc p, \vecc{Q}_N) \in \mathcal{F}(\mathcal{E}_h^r)$, then for every $\tilde h \in \{1, \dots, 5\}$ there exists a payoff-equivalent assignment $((\vecc{\tilde S}_I, \vecc{\tilde Z}_J), \vecc{\tilde p}, \vecc{\tilde Q}_N) \in \mathcal{F}(\mathcal{E}_{\tilde{h}}^{\mathrm{sep}})$.
\end{lemma}

\begin{proof}
  Consider a generic feasible assignment $((\vecc{S}_I, \vecc{Z}_J), \vecc p, \vecc{Q}_N) \in \mathcal{F}(\mathcal{E}_h^r)$, and suppose that $(\vecc{S}_I, \vecc{Z}_J) \neq \vecc 0$ (otherwise, the statement is trivial). Then, there exists an item $\hat{a} \in A$ that is traded in total a positive $m \geq 1$ number of times. Consider an item-equivalent fourth-order admissible allocation $(\vecc{\tilde S}_I, \vecc{\tilde Z}_J)$ (from Lemma~\ref{lemma:12345to4(I)Allocation}), and define $m_i = (\vecc{s}_i)_{\hat{a}} = (\vecc{\tilde s}_i)_{\hat{a}} \geq 0$ the number of units of item $\hat{a}$ in $i$'s cumulative bundle $\vecc{s}_i = \vecc{\tilde s}_i$ and $m_j = (\vecc{z}_j)_{\hat{a}} = (\vecc{\tilde z}_j)_{\hat{a}} \geq 0$ the number of units of item $\hat{a}$ in $j$'s cumulative bundle $\vecc{z}_j = \vecc{\tilde z}_j$. Since we always assume that no offered item gets discarded,
  $$\sum_{i \in I} m_i = \sum_{j \in J} m_j = m.$$
  The first-order pricing system $\tilde{\vecc p}$ is defined so that all the money
  $$\sum_{i \in I} p_i + \sum_{i \in I}\sum_{j \in J} q_{ij} \geq 0$$
  that is sent from the buyers to the sellers in the initial assignment (which is non-negative due to feasibility) is exchanged through the payment of this item. Therefore, the new price of item $\hat{a}$ is the total exchange of money divided by the number of times the item is traded, i.e.,
  $$\tilde{p}_{\hat{a}} = \frac{\sum_{i \in I} p_i + \sum_{i \in I}\sum_{j \in J} q_{ij}}{m} \geq 0,$$
  whereas the price of all the other items $a \neq \hat{a}$ is $\tilde p_a = 0$. The new cumulative prices depend on how many units of $\hat{a}$ each agent has in their own cumulative bundle, i.e., $\tilde{p}_i = \vecc{\tilde p} \cdot \vecc{\tilde s}_i = \tilde{p}_{\hat{a}} m_i$ for every $i \in I$ and $\tilde{p}_j = \vecc{\tilde p} \cdot \vecc{\tilde z}_j = \tilde{p}_{\hat{a}} m_j$ for every $j \in J$.

  To preserve payoff equivalence, we must ensure that each agent pays or receives the payoff difference through side payments. In particular, each buyer's new cumulative side payment has to be $\tilde q_i = q_i + p_i - \tilde{p}_i$ and each seller's new cumulative side payment has to be $\tilde q_j = q_j - p_j + \tilde{p}_j$. These are well-defined since
  $$\sum_{i \in I} \tilde q_i = \sum_{i \in I} (q_i + p_i - \tilde{p}_i) = \sum_{i \in I} \left(\sum_{j \in J} q_{ij} + p_i - \tilde{p}_i\right) + \sum_{i,i' \in I} q_{ii'} = 0$$
  (analogously for the sellers' total side payments).
  To determine the individual side payments $\vecc{\tilde Q}_N$, we can then redistribute them separately to the buyers and the sellers, since each side-payment graph is complete.

  We can then conclude that $((\vecc{\tilde S}_I, \vecc{\tilde Z}_J), \vecc{\tilde p}, \vecc{\tilde Q}_N) \in \mathcal{F}(\mathcal{E}_{\tilde{h}}^{\mathrm{sep}})$ for every $\tilde h \in \{1, \dots, 5\}$ and it is payoff equivalent to the initial assignment.
\end{proof}

Price equivalence does not hold in general between the full and separate side-payment settings, unless we go to the same or higher order setting with the same or less restricted side payments, as outlined in Section~\ref{sec:increasing_order}. See Example~\ref{example:FullnottoSep(P)} for a discussion.

\subsection{Single Side Payments}\label{sec:equiv_single}
When only single side payments are allowed, some equivalence properties hold only for sufficient levels of price personalization. We first observe that, if the prices depend only on the items in each bundle or on the bundle itself, then no equivalence relations can be found starting from a setting in which prices are personalized for at least one side of the market.

\begin{remark}
  When going from a higher order setting to the first- or second-order settings with at most single side payments, the rigidity of the pricing system and the restrictions on the side payments do not allow any equivalence result. In fact, we can always consider an economy where there exists a feasible assignment whose pricing system is admissible only in the higher order settings, and any corresponding item-equivalent assignment is not priceable in the first- and second-order settings. In particular, the following modifications of settings do not admit payoff equivalence:
  \begin{itemize}
    \item Going from the $h^{\mathrm{th}}$-order setting with $h \in \{2, \dots, 5\}$ and any type of side payments to the first-order setting with at most single side payments (see Example~\ref{example:2345noto1(I)SingNo}).
    \item Going from the $h^{\mathrm{th}}$-order setting with $h \in \{3, 4, 5\}$ and any type of side payments to the second-order setting with at most single side payments (see Example~\ref{example:345noto2(IPPiM)SingNo}).
  \end{itemize}
\end{remark}

However, we provide positive results for the more personalized settings in the following.
The first result establishes that going from the fifth-order to the fourth-order setting with single side payments preserves payoff equivalence. This is because the money each agent receives or pays in the initial assignment can be distributed through side-payments on the side where this is allowed, and on the other side is redistributed using individualized prices.

\begin{lemma}\label{lemma:5to4(Pi)Sing}
  Let $((\vecc{s}_I, \vecc{z}_J), \vecc p, \vecc{Q}_N) \in \mathcal{F}(\mathcal{E}_5^r)$ with $r \in \{\mathrm{buy}, \, \mathrm{sell}\}$, then there exists a payoff-equivalent assignment $((\vecc{\tilde S}_I, \vecc{\tilde Z}_J), \vecc{\tilde p}, \vecc{\tilde Q}_N) \in \mathcal{F}(\mathcal{E}_4^r)$.\footnote{It follows that every assignment $((\vecc{s}_I, \vecc{z}_J), \vecc p, \vecc{Q}_N) \in \mathcal{F}(\mathcal{E}_5^{\mathrm{no}})$ is payoff equivalent to some $((\vecc{\tilde S}_I, \vecc{\tilde Z}_J), \vecc{\tilde p}, \vecc{\tilde Q}_N) \in \mathcal{F}(\mathcal{E}_4^r)$ with $r \in \{\mathrm{buy}, \, \mathrm{sell}\}$.}
\end{lemma}

\begin{proof}
  Let $((\vecc{s}_I, \vecc{z}_J), \vecc p, \vecc{Q}_N)$ be a fifth-order feasible assignment with single side payments. Consider a fourth-order admissible allocation $(\vecc{\tilde S}_I, \vecc{\tilde Z}_J)$ that is item equivalent to $(\vecc{s}_I, \vecc{z}_J)$ (from Lemma~\ref{lemma:12345to4(I)Allocation}). We want to find a corresponding pricing system $\vecc{\tilde p}$ and side payment scheme $\vecc{\tilde Q}_N$ such that $\tilde{p}_i + \tilde{q}_i = p_i + q_i \in \R$ for every $i \in I$ and $\tilde{p}_j -\tilde{q}_j = p_j - q_j \geq 0$ for every $j \in J$. These conditions ensure that every agent receives or pays the same amount of money as in the initial assignment, so that $((\vecc{\tilde S}_I, \vecc{\tilde Z}_J), \vecc{\tilde p}, \vecc{\tilde Q}_N)$ is payoff equivalent to $((\vecc{s}_I, \vecc{z}_J), \vecc p, \vecc{Q}_N)$.

  To determine the individual prices and side payments, we can apply the following reasoning. Suppose, for example, that side payments are only allowed among buyers. In this setting, every buyer can side-pay each other, so if any $i \in I$ needs to receive money (i.e., $\tilde{p}_i + \tilde{q}_i < 0$), there are always agents that can side-pay them.
  On the other side, the sellers $j \in J$ that need to receive money (i.e., $\tilde{q}_j - \tilde{p}_j < 0$) are only the ones that in the starting assignment $((\vecc{s}_I, \vecc{z}_J), \vecc p, \vecc{Q}_I)$ had a non-empty cumulative bundle $\vecc z_j$ (since sellers can't receive any side payments in this setting). Hence, there exists at least one buyer who trades a non-empty bundle with such $j$ in the new assignment as well. By setting the prices $\tilde p_{ij}(\vecc{\tilde z}_{ij})$ so that $\sum_{i : \vecc{\tilde z}_{ij} \neq \vecc{0}} \tilde p_{ij}(\vecc{\tilde z}_{ij}) = p_j - q_j$, $j$ can also maintain their payoff (the prices of the non-traded bundles can be set arbitrarily, for example, $\tilde p_{ij}(\vecc{\omega}) = \alpha \geq 0$ for every $\vecc{\omega} \neq \vecc{\tilde s}_{ij} = \vecc{\tilde z}_{ij}$, $i \in I$ and $j \in J$).
  Observing that
  $$\sum_{i \in I} (\tilde{p}_i + \tilde{q}_i) = \sum_{i \in I} (p_i + q_i) = \sum_{j \in J} p_j = \sum_{j \in J} \tilde p_j$$
  and each buyer can receive a side payment from any other buyer, we can conclude that all the sellers are able to receive the right amount of money.
\end{proof}

Price equivalence does not hold in this case due to the fully centralized nature of the fifth-order setting, which may allow part of a buyer's cumulative price to be effectively directed to a seller who doesn't offer any item the buyer requests. This indirect exchange of money is not possible in the fourth-order setting, since the price of the empty bundle, which those agents trade with each other, is always normalized to zero (see Example~\ref{example:5noto4(IPi)No}).

If instead of maintaining the same single side-payment restrictions we restrict or change the side of the market in which side payments are allowed, i.e., going from the fifth-order setting with at least buyers' or sellers' single side payments to the fourth-order setting with sellers' or buyers' single side payments respectively, then payoff equivalence does not hold for the same reasons illustrated in Examples~\ref{example:1BuySepFullnotto...} and \ref{example:1SellSepFullnotto...}. Similarly, going from the fifth-order setting with any type of side payments to the fourth-order setting with no side payments does not admit payoff equivalence (see Example~\ref{example:5noto4(IPi)No}).

The second result highlights an interesting property of the third-order setting with sellers' single side payments: Any feasible assignment with at most sellers' single side payments is payoff equivalent to a third-order admissible assignment with sellers' single side payments. This is due to the fact that price personalization only applies to buyers, but not for sellers in the third-order setting. Hence, it is sufficient to adjust the prices so that buyers can pay the exact amount they have to maintain payoff equivalence, and then readjust the sellers' payoffs through side payments.

\begin{lemma}\label{lemma:245to3(Pi)Sel}
  Let $((\vecc{S}_I, \vecc{Z}_J), \vecc p, \vecc{Q}_J) \in \mathcal{F}(\mathcal{E}_h^r)$ with $h \in \{2, 4, 5\}$ and $r \in \{\mathrm{sell}, \, \mathrm{no}\}$, then there exists a payoff-equivalent assignment $((\vecc{\tilde S}_I, \vecc{\tilde Z}_J), \vecc{\tilde p}, \vecc{\tilde Q}_J) \in \mathcal{F}(\mathcal{E}_3^{\mathrm{sell}})$.
\end{lemma}

\begin{proof}
  We only need to prove the statement for $h = 4$, since we have already shown in Lemmas~\ref{lemma:2to4(P)} and \ref{lemma:5to4(Pi)Sing} that every feasible second-order and fifth-order admissible assignments with at most sellers' single side payments is payoff equivalent to some assignment in the fourth-order setting with sellers' single side payments.

  Let then $((\vecc{S}_I, \vecc{Z}_J), \vecc p, \vecc{Q}_J) \in \mathcal{F}(\mathcal{E}_4^{\mathrm{sell}})$. To find a payoff-equivalent third-order admissible assignment $((\vecc{\tilde S}_I, \vecc{\tilde Z}_J), \vecc{\tilde p}, \vecc{\tilde Q}_J)$, we can keep the same allocation $(\vecc{\tilde S}_I, \vecc{\tilde Z}_J) = (\vecc{S}_I, \vecc{Z}_J)$ and adjust the prices so that if a buyer trades the same bundle with multiple sellers, the price they pay to each of them is the same.

  More formally, let $i \in I$ be a buyer who requests bundle $\vecc{\omega} \leq \vecc{\Omega}$ from $m \geq 1$ sellers $j_1, \dots, j_m \in J$ at the price of $p_{ij_1}(\vecc{\omega}), \dots, p_{ij_m}(\vecc{\omega})$ respectively. Define the price of $\vecc{\omega}$ for $i$ in the new allocation as the mean of these prices:
  $$\tilde p_i(\vecc{\omega}) = \frac{\sum_{l = 1}^m p_{ij_l}(\vecc{\omega})}{m}$$
  (as usual, the prices of the non-traded bundles can be set arbitrarily). If we repeat this procedure for all buyers and their corresponding bundles, the total price that each buyer $i \in I$ pays remains the same, as does their payoff. It follows that also the total amount of money that the buyers pay to the sellers stays the same, i.e.,
  $$\sum_{j \in J} \tilde p_j = \sum_{i \in I} \tilde p_i = \sum_{i \in I} p_i = \sum_{j \in J} p_j.$$

  For each seller $j \in J$, the new cumulative price is $\tilde p_j = \sum_{i \in I} \tilde p_i(\vecc{\tilde z}_{ij})$. Therefore, to maintain the same payoff as in the initial assignment, they should pay or receive a cumulative side payment of $\tilde q_j = \tilde p_j - p_j + q_j \in \R$. Since all the sellers can side-pay each other and
  $$\sum_{j \in J} \tilde q_j = \sum_{j \in J} (\tilde p_j - p_j + q_j) = \sum_{j \in J} (\tilde p_j - p_j) + \sum_{j \in J} q_j = 0,$$
  the individual side payments $\vecc{\tilde Q}_J$ can be easily determined.
\end{proof}

Further negative results are discussed in this context are discussed in Examples~\ref{example:45noto3(IPPiM)BuyNo} and \ref{example:1BuySepFullnotto...}.

\section{Stable Solutions}\label{sec:stable_solutions_properties}
Now that the previous section has established useful relations between the different settings, we continue to investigate their stability properties.
Most importantly, the maximum social welfare achievable by an admissible assignment can differ across different settings or stability notions.
However, if two markets are payoff equivalent, their $\mathcal{T}$-core coincides for any partition $\mathcal{T}$ of the set of agents, as both the assignments that realize a given payoff vector for the core, just as the possible deviating coalitions, correspond to each other.
Therefore, from Theorem~\ref{thm:setting-equivalences}, we can already deduce the following.

\begin{corollary}
  Let $\mathcal{E}$ be an economy and $\mathcal{T}$ an arbitrary partition of its agents.
  \begin{itemize}
    \item The $\mathcal{T}$-cores of all the order settings with full or separate side payments coincide, i.e., $\mathcal{C}^{\mathcal{T}}(\mathcal{E}_h^{r}) = \mathcal{C}^{\mathcal{T}}(\mathcal{E}_{\tilde h}^{\tilde r})$ for every $h, \tilde{h} \in \{1, \dots, 5\}$ and $r, \tilde r \in \{\mathrm{full}, \, \mathrm{sep}\}$.
    \item $\mathcal{C}^{\mathcal{T}}(\mathcal{E}_4^{\mathrm{buy}}) = \mathcal{C}^{\mathcal{T}}(\mathcal{E}_{5}^{\mathrm{buy}})$.
    \item $\mathcal{C}^{\mathcal{T}}(\mathcal{E}_3^{\mathrm{sell}}) = \mathcal{C}^{\mathcal{T}}(\mathcal{E}_4^{\mathrm{sell}}) = \mathcal{C}^{\mathcal{T}}(\mathcal{E}_{5}^{\mathrm{sell}})$.
  \end{itemize}
\end{corollary}
From Lemma~\ref{lem:TRTUwithr_is_SRTU}, we know that $\mathcal{C}^{\mathcal{T}}(\mathcal{E}_h^{\mathrm{full}}) = \mathcal{C}^{\mathcal{T}^{\mathrm{no}}}(\mathcal{E}_h^{\mathrm{full}})$.
Therefore, we can establish an even stronger result: For the full and separate side-payment settings, regardless of the order setting, all notions of stability coincide.
The same structure of solutions is also achieved when the stability notion sufficiently complements the side-payment setting.

\begin{corollary}\label{cor:FullSepCoreCoincide}
  $\mathcal{C}^{\mathcal{T}}(\mathcal{E}_h^{r}) $ is the same for every $h\in \{1, \dots, 5\}$, $r \in \{\mathrm{full}, \, \mathrm{sep}, \, \mathrm{buy}, \, \mathrm{sell}, \, \mathrm{no}\}$ and arbitrary partition $\mathcal{T}$ such that $\mathcal{T} \vee \mathcal{T}^r \succeq \mathcal{T}^{\mathrm{sep}}$.
\end{corollary}

\subsection{Separations}
The existence of stable vectors is not guaranteed in general market settings (see Example~\ref{example:emptyCores}). However, if the TU-core is not empty (and consequently, also the $\mathcal{T}$-core for every partition $\mathcal{T} \succeq \mathcal{T}^{\mathrm{sep}}$), all of its vectors achieve the maximum social welfare, independently of the market order setting and side-payment restrictions. However, this is generally not the case for finer partitions and lower order settings, where not only bRTU-core, sRTU-core, and NTU-core stable vectors might not achieve the maximum social welfare (see Example~\ref{example:RTUUniquePayoff}), but there also might be stable vectors whose total utility differs from each other (see Examples~\ref{example:NTUDifferentPayoffs} and \ref{example:sRTUCorePriceability}).

On the other hand, in some economies in which the TU-core is empty, limiting side payments to one side of the market or completely forbidding them can yield stable outcomes. Moreover, bRTU- and sRTU-stable vectors might achieve a social welfare that is arbitrarily larger than the one of any NTU-stable outcome (see Examples~\ref{example:sidePaymentsAreCool} and \ref{example:sidePaymentsAreCool2}).
Thus, the assumptions on how agents transfer utility outside the market have a critical influence on the achievable total welfare in what can be considered a stable assignment, and the new stability notions meaningfully introduce new solution concepts to these settings.

\subsection{Stability-Equivalent Settings}\label{sec:StabilityEquivalent}
The differences between the stability notions disappear when the resale restrictions in the market allow for more personalized prices.
By using the structural properties of stable outcomes, we are able to establish further equivalence results, that do not hold for general assignments.
Specifically, in this section, we discuss how this is the case for the fourth- and fifth-order settings, where side-payments can be effectively entirely disregarded.
Note that, by Theorem~\ref{thm:allisNTU}, $\mathcal{T}$-stability in a market $\mathcal{E}_h^r$ only depends on the join of the partitions $\mathcal{S} = \mathcal{T} \vee \mathcal{T}^r$, or in other words, $\mathcal{C}^{\mathcal{T}}(\mathcal{E}_h^r) = \mathcal{C}^{\mathcal{T}^\mathrm{no}}(\mathcal{E}_h^{s})$ for the side-payment restrictions corresponding to $\mathcal{S}$. Therefore, we will discuss the results by assuming that all the transfers of utility are made inside the market and $\mathcal{T} = \mathcal{T}^{\mathrm{no}}$.

Let's start with the fifth-order setting. We first show that for every stable assignment in the fifth-order setting, we can always find a payoff-equivalent one in $\mathcal{E}_5^{\mathrm{no}}$ by simply including the side payments in the cumulative prices that every agent sends or receives from the clearinghouse.

\begin{lemma}\label{lemma:NTUStableAss5}
  Let $((\vecc{s}_I, \vecc{z}_J), \vecc p, \vecc{Q}_N) \in \mathcal{F}(\mathcal{E}_5^r)$ be an NTU-stable assignment in $\mathcal{E}_5^r$, then there exists a payoff-equivalent NTU-stable assignment $((\vecc{\tilde s}_I, \vecc{\tilde z}_J), \vecc{\tilde p}) \in \mathcal{F}(\mathcal{E}_5^{\mathrm{no}})$.
\end{lemma}

\begin{proof}
  In the fifth-order setting, all the price payments are received or sent by the clearinghouse as a cumulative price.
  Since the only restrictions on prices are non-negativity and the normalization for the empty bundle, we show that agents can instead exchange these payments through the clearinghouse by adjusting their cumulative prices.
  For an NTU-stable assignment $((\vecc{s}_I, \vecc{z}_J), \vecc p, \vecc{Q}_N) \in \mathcal{F}(\mathcal{E}_5^r)$, we can define a payoff-equivalent assignment $((\vecc{s}_I, \vecc{z}_J), \vecc{\tilde p})$ with the same allocation and by setting cumulative prices $\tilde p_i = p_i + q_i$ for every $i \in I$ and $\tilde p_j = p_j - q_j$ for every $j \in J$.

  To check the admissibility of this assignment in $\mathcal{E}_5^\mathrm{no}$, we need to establish that all the prices are non-negative and that no positive price is paid for an empty bundle:

  \begin{itemize}
    \item If $\tilde p_i < 0$ for some $i \in I$, this means that in addition to the bundle, buyer $i$ has a net gain in the money they exchange.
      The agents in the coalition $C = N \setminus \{i\}$ can improve their cumulative payoff in $\mathcal{E}_5^r$ by not offering the items in $\vecc{s}_i$ to $i$. The remaining buyers can slightly decrease their prices, while the sellers slightly increase them until budget balance is restored, thereby increasing each agent's payoff in contradiction to the stability of the original assignment.
    \item If $\tilde p_j < 0$ for some $j \in J$, then seller $j$ has a negative payoff $\pi_j = \tilde \pi_j = \tilde p_j - v_j(\vecc{z}_j) < 0$, contradicting their individual rationality.
    \item If some agent $k \in N$ trades the empty bundle and $\tilde p_k > 0$ (due to side payments), then:
      \begin{itemize}
        \item if $k = i$ is a buyer, their payoff is negative since $\pi_i = \tilde\pi_i = v_i(\vecc{0}) - \tilde p_i < 0$, whereas
        \item if $k = j$ is a seller, excluding them from the trade wouldn't affect the allocation of the other agents in $C = N \setminus \{j\}$, and in fact they could increase their payoff in $\mathcal{E}_5^r$ by redistributing $\tilde p_j$ among themselves.
      \end{itemize}
  \end{itemize}

  All the previous conditions contradict the feasibility and stability of $((\vecc{s}_I, \vecc{z}_J), \vecc p, \vecc{Q}_N)$, so we can deduce that $((\vecc{s}_I, \vecc{z}_J), \vecc{\tilde p}) \in \mathcal{F}(\mathcal{E}_5^{\mathrm{no}})$.
  Finally, its NTU-stability follows directly from payoff equivalence and the stability of the initial assignment.%
  \footnote{Note that restricting the side payments also restricts the number of possible deviating assignments, so if a payoff vector is stable in $\mathcal{C}^{\mathcal{T}^{\mathrm{no}}}(\mathcal{E}_h^r)$ and can be realized by an admissible assignment in $\mathcal{E}_h^{\tilde r}$ with $\tilde r \preceq r$, then it is also stable in $\mathcal{C}^{\mathcal{T}^{\mathrm{no}}}(\mathcal{E}_h^{\tilde r})$.}
\end{proof}

On the other hand, every assignment that is NTU-stable in the market in $\mathcal{E}_5^{\mathrm{no}}$ is also stable if we allow more side payments. If this were not the case, we could still find a deviating assignment with no side payments in the more general market.

\begin{lemma}\label{lem:NTU_stable_noSP5}
  Let $((\vecc{s}_I, \vecc{z}_J), \vecc p) \in \mathcal{F}(\mathcal{E}_5^\mathrm{no})$ be an NTU-stable assignment in $\mathcal{E}_5^\mathrm{no}$, then it is also NTU-stable in $\mathcal{E}_5^r$.
\end{lemma}

\begin{proof}
  We only need to prove the result for $r = \mathrm{full}$, since for the other side-payment settings it then follows immediately. Suppose by contradiction that the assignment $((\vecc{s}_I, \vecc{z}_J), \vecc p) \in \mathcal{F}(\mathcal{E}_5^{\mathrm{no}})$ is NTU-stable in $\mathcal{E}_5^{\mathrm{no}}$ but not in $\mathcal{E}_5^{\mathrm{full}}$. This means that there exists a coalition $C \subseteq N$ and an assignment $((\vecc{\tilde s}_I, \vecc{\tilde z}_J), \vecc{\tilde p}, \vecc{\tilde Q}_N) \in \mathcal{F}(\mathcal{E}_5^{\mathrm{full}}[C])$ with non-zero side payments (otherwise the initial assignment wouldn't be stable in $\mathcal{E}_5^{\mathrm{no}}$) such that $\tilde \pi_k > \pi_k$ for every $k \in C$. Suppose that $C$ is a minimal deviating coalition, so that there is no coalition $C' \subseteq C$ where each agent can strictly improve their payoff over $\vecc{\tilde\pi}_{C'}$. Then, by Lemma~\ref{lemma:NTUStableAss5}, there exists an assignment $((\vecc{\tilde s}'_I, \vecc{\tilde z}'_J), \vecc{\tilde p'})\in \mathcal{F}(\mathcal{E}_5^{\mathrm{no}}[C])$ that is payoff equivalent to $((\vecc{\tilde s}_I, \vecc{\tilde z}_J), \vecc{\tilde p}, \vecc{\tilde Q}_N)$ and hence also a deviating assignment in $\mathcal{E}_5^{\mathrm{no}}$, which contradicts the initial stability assumption.
\end{proof}

Then, by Lemma~\ref{lem:TRTUwithr_is_SRTU}, we deduce that $\mathcal{C}^{\mathcal{T}}(\mathcal{E}_5^{r}) = \mathcal{C}^{\mathcal{S}}(\mathcal{E}_{5}^{\tilde r})$ for every $r, \tilde r \in \{\mathrm{full}, \, \mathrm{sep}, \, \mathrm{buy}, \, \mathrm{sell}, \, \mathrm{no}\}$ and partitions $\mathcal{T}, \mathcal{S}$.
Note that the previous two lemmas can be applied to any type of side-payment restrictions, not only the ones we originally defined.

We now discuss the analogous results for the fourth-order setting. In this case, the proofs are more involved since the role of the clearinghouse is not as powerful as before, and therefore, we need to find appropriate paths in the pricing graph, through which we are able to send the side payments instead.

\begin{lemma}\label{lemma:fourthSettingStableAssignments}
  Let $((\vecc{S}_I, \vecc{Z}_J), \vecc p, \vecc{Q}_N) \in \mathcal{F}(\mathcal{E}_4^r)$ be an NTU-stable assignment in $\mathcal{E}_4^r$, then there exists a payoff-equivalent NTU-stable assignment $((\vecc{\tilde S}_I, \vecc{\tilde Z}_J), \vecc{\tilde p}) \in \mathcal{F}(\mathcal{E}_4^{\mathrm{no}})$.
\end{lemma}

\begin{proof}
  Similar to the previous proof, we are going to transfer all the side payments into the payments of the market. However, in this case, the structure of the payments is more intricate than before.
  Starting from any NTU-stable assignment $((\vecc{S}_I, \vecc{Z}_J), \vecc p, \vecc{Q}_N) \in \mathcal{F}(\mathcal{E}_4^r)$, we want to construct a payoff-equivalent fourth-order admissible assignment $((\vecc{S}_I, \vecc{Z}_J), \vecc{\tilde p})$ with no side payments.
  For the proof, we use the payment digraph $G = (N, E) = (N, E_p \cup E_q)$, where $E_p$ is the set of price arcs and $E_q$ is the set of side-payment arcs. On this graph, we can sequentially eliminate each arc in $E_q$, until no more side payments are left (see Algorithm \ref{alg:4Stability}). The remainder of the proof establishes the correctness of the algorithm.

  Note that, in the fourth-order setting, the role of the clearinghouse solely consists of checking the assignment admissibility, and both the package exchanges as well as the price transactions correspond to direct trades. To simplify the notation, we can then assume that buyers and sellers trade directly with each other. Hence,
  \begin{itemize}
    \item a price arc $(i,j)$ exists in $E_p$ if and only if the two agents trade a non-empty bundle, i.e., $\vecc{s}_{ij} = \vecc{z}_{ij} \neq \vecc{0}$, and
    \item a side-payment arc $(k,l)$ exists in $E_q$ if and only if $q_{kl}>0$.
  \end{itemize}
  For the first reduction step of this graph, we want to make sure that no parallel edges exist, so that the graph is simple. This could be the case when a side payment arc exists in parallel to a payment arc.
  To avoid ambiguity, if a buyer $i$ both pays $p_{ij}$ and side-pays $q_{ij}>0$ to a seller $j$, we eliminate the side-payment arc and keep the price arc by increasing the corresponding price to $p_{ij} + q_{ij} > 0$ (this doesn't affect admissibility since in the fourth-order setting the prices are fully personalized).

  In order to find a way to reroute the side payments through the payment graph, we define the residual graph $\bar{G} = (N, E \cup E^{-1})$, where $E^{-1} = \{e^{-1}=(l,k) : e = (k, l) \in E\}$ is the set of reversed arcs. For a side-payment edge $(s, t)$ that we want to eliminate, we can look for a path in the residual graph to send the money through instead. The edges on this path can be forwards (side-)payment edges, which we can arbitrarily increase, or a backward edge which can be used up to the current amount of money.
  The capacity of each arc is thus given by:

  \begin{equation}\label{eq:capacity}
    c_e =
    \begin{cases}
      p_{ij} ,  &\text{$e \in E^{-1}_{p}.$}\\
      +\infty, &\text{otherwise.}
    \end{cases}
  \end{equation}

  Consider a pair of agents $(s, t)$ with positive side payment $q_{st} > 0$, and let $\bar{G}_{st} $ be the graph where we remove both arcs corresponding to the side payment between the two agents.
  We can now compute the maximum $s-t$ flow $\vecc f$ in $\bar{G}_{st}$. If the maximum flow value is at least $q_{st}$, we use a flow of value $q_{st}$ to route the side payment through the payment graph in some way and set $q_{st} = 0$. Thus, the payments remain admissible, all the payoffs remain the same, but we have eliminated the side-payment edge $(s, t)$. Apply this approach iteratively to all side-payment edges. Because in every iteration, at least one side payment arc is removed from the graph, the algorithm terminates in at most $N(N-1)$ steps and achieves a payoff-equivalent assignment that does not use any side-payment edges.

  If, however, the maximum flow value is less than $q_{st}$, we do not succeed in routing the full side-payment amount through the graph. In this case, apply the flow to the payment graph and reduce the side payment accordingly. In the new residual graph, no $s-t$-path with positive capacity exists. Consider the strongly connected component $C$ of $s$ in the residual graph spanned by all edges with positive capacity (i.e., we might remove the backward payment edges that effectively transfer no money). Observe that by definition, we can transfer at least some positive amount $\epsilon$ of money from $s$ to all agents in the strongly connected component in the payment graph. By doing so, we strictly increase the payoff of each agent in $C$. In order to establish this as a deviating coalition, we now only need to make sure that the allocation is admissible. Because all package exchanges correspond to payment arcs, the only exchanges that need to be taken care of are trades with agents that have no path of positive capacity from $C$. Therefore, the only edges connecting them to $C$ are backwards payment edges with value 0. Such an agent is a buyer who doesn't pay for any of the bundles they trade with $C$, and we simply set $z_{ij} = 0$ for all such edges, as the corresponding sellers cannot lose payoff by not offering this bundle. This coalition and constructed assignment contradict that the assignment $((\vecc{S}_I, \vecc{Z}_J), \vecc p, \vecc{Q}_N) \in \mathcal{F}(\mathcal{E}_4^r)$ is stable.

  The NTU-stability of the resulting assignment $((\vecc{S}_I, \vecc{Z}_J), \vecc{\tilde p})$ follows directly from payoff equivalence and the stability of the initial assignment.
\end{proof}

\begin{algorithm}
  \KwIn{Stable assignment $((\vecc{S}_I, \vecc{Z}_J), \vecc p, \vecc{Q}_N) \in \mathcal{F}(\mathcal{E}_4^r)$ with corresponding digraph $G = (N, E = E_p \cup E_q)$.}
  \KwOut{Stable assignment $((\vecc{S}_I, \vecc{Z}_J), \vecc{\tilde p}) \in \mathcal{F}(\mathcal{E}_4^{\mathrm{no}})$.}

  $\vecc{\tilde p} \gets \vecc{p}$\;
  \For{$(k_s, k_t) \in N \times N$ such that $q_{k_sk_t} > 0$}{
    $E_q \gets E_q \setminus \{(k_s,k_t)\}$\;
    $\bar{G} \gets (N, E \cup E^{-1})$\;
    $\vecc{c}$ as in \eqref{eq:capacity}\;
    Compute the maximum $k_s -k_t$ flow $\vecc{f}$ in $\bar G$ with net flow $q_{k_sk_t}$\;
    \For{$(i,j) \in E_p$}{
      $\tilde p_{ij} \gets \tilde p_{ij} + f_{ij} - f_{ji}$\;
    }
    \For{$(k,l) \in E_q$}{
      $q_{kl} \gets q_{kl} + f_{kl} - f_{lk}$\;
    }
  }
  \caption{Eliminating side payments for the fourth-order setting.}\label{alg:4Stability}
\end{algorithm}

As before, any NTU-stable assignment in $\mathcal{E}_4^{\mathrm{no}}$ is also stable with any type of side payments. With a similar construction as in the proof of the previous lemma, we can show that we can otherwise find a deviating coalition $C$ and an assignment with no side payments, where every agent strictly improves their payoff.

\begin{lemma}\label{lem:NTUstable4no_is_stable_r}
  Let $((\vecc{\tilde S}_I, \vecc{\tilde Z}_J), \vecc{\tilde p}) \in \mathcal{F}(\mathcal{E}_4^{\mathrm{no}})$ be an NTU-stable assignment in $\mathcal{E}_4^{\mathrm{no}}$, then it is also NTU-stable in $\mathcal{E}_4^r$.
\end{lemma}

\begin{proof}
  Consider an NTU-stable assignment $((\vecc{S}_I, \vecc{Z}_J), \vecc{p}) \in \mathcal{F}(\mathcal{E}_4^{\mathrm{no}})$. Suppose by contradiction that it is not NTU-stable in $\mathcal{E}_4^r$, and in particular that there exists a deviating coalition $C$ with an assignment $((\vecc{\tilde S}_I, \vecc{\tilde Z}_J), \vecc{\tilde p}) \in \mathcal{F}(\mathcal{E}_4^{r})$ s.t.\ every agent $k$ improves their payoff. Analogously to the proof of Lemma~\ref{lemma:fourthSettingStableAssignments}, we can try to eliminate all side-payment edges by routing the payments through the graph. If all side-payment edges are deleted, we have found an assignment contradicting the stability of $((\vecc{S}_I, \vecc{Z}_J), \vecc{p})$.

  However, in this case, the deviation need not be a stable assignment, and we therefore cannot immediately conclude the construction when we do not find a path to eliminate the edge.
  Nevertheless, a similar argument still holds: Consider the same strongly connected component in the residual graph induced by some agent $s$ with $q_{st} > 0$. In case the side-payment edge cannot be eliminated by routing its value through the remaining graph, the same conclusion holds: the only connections to agents outside of the coalition are by using only payment edges with value zero.
  We can now remove the traded bundles to all such agents without affecting the payoff of the agents in the coalition. This may include the agent $t$, for whom we additionally also remove the remaining side-payment value $q_{st}$, which also can only increase the payoff of agent $s$.

  After iterating this argument over the remaining coalition $C$, thereby maintaining non-decreasing payoff for each agent in $C$ while removing side-payment edges one-by-one, we end up with a coalition that is feasible in $\mathcal{E}_4^{\mathrm{no}}$ and still improves the payoff to provide a deviation in contradiction with the stability of $((\vecc{S}_I, \vecc{Z}_J), \vecc{p})$.
\end{proof}

Since the previous two lemmas are still valid for any type of side payments, by Lemma~\ref{lem:TRTUwithr_is_SRTU}, we can then deduce that $\mathcal{C}^{\mathcal{T}}(\mathcal{E}_4^{r}) = \mathcal{C}^{\mathcal{S}}(\mathcal{E}_{4}^{\tilde r})$ for every $r, \tilde r \in \{\mathrm{full}, \, \mathrm{sep}, \, \mathrm{buy}, \, \mathrm{sell}, \, \mathrm{no}\}$ and partitions $\mathcal{T}, \mathcal{S}$.

Summarizing the results in this section, we can deduce that all the $\mathcal{T}$-cores for the fourth- and fifth-order settings coincide.

\begin{corollary}
  $\mathcal{C}^{\mathcal{T}}(\mathcal{E}_h^r)$ is the same for every $h \in \{4, 5\}$, $r \in \{\mathrm{full}, \, \mathrm{sep}, \, \mathrm{buy}, \, \mathrm{sell}, \, \mathrm{no}\}$ and partition $\mathcal{T}$.
\end{corollary}

Combining the observations from above provides the summary result of Theorem~\ref{thm:45stability_all_NTU}, stating that if the join partition is coarse enough, or the prices are personalized enough, all the stability notions in an economy $\mathcal{E}$ coincide.

\section{Concluding Remarks}
We expand well-established combinatorial market models by adding explicit side payments, allowing for different restrictions on utility transfers among agents. To compare the relative strengths of the resulting markets, we introduce a hierarchy of equivalence notions, ranging from item equivalence to market equivalence. This comparison yields near-monotonicity: the set of assignments that can be realized generally increases with more personalized prices or with the allowance of side payments across more groups. However, in the second-order setting, payoff equivalence fails when compared to higher-order settings under restricted utility transfer. This irregularity disappears once sufficiently strong side payments are permitted. We also generalize the classical notions of TU- and NTU-stability to model limits on financial transfers between different groups of agents. The new concept aligns naturally with our models of side payments and also sheds new light on the original notions, especially NTU-stability. In particular, allowing restricted forms of collusion expands the set of stable outcomes, leading to more nuanced solution concepts that reflect different assumptions about agents’ ability to coordinate.

\subsection*{Acknowledgments}
The authors thank their advisor Andreas S.\ Schulz for his valuable input. They also thank the anonymous reviewers for their constructive feedback.

\bibliographystyle{splncs04}
\bibliography{bib}

\newpage
\appendix

\section{Formal Definition of the Market Model}
\label{secapp:formal_defs}

We denote by $\N$ the set of positive integers, by $\R$ the real numbers, and by $\N_{\geq 0}$ and $\R_{\geq 0}$ the sets of non-negative integer and real numbers, respectively.
We set vectors like $\vecc x$ in bold font to distinguish them from scalars and denote their $i$-th component by $x_i$.
The symbol $\vecc 0$ indicates the zero vector of the appropriate dimension, depending on the context.
Finally, if $\vecc{\alpha}, \vecc{\beta} \in \R^n$, we will write $\vecc{\alpha} < \vecc{\beta}$ to indicate that $\alpha_i < \beta_i$ for every $i \in \{1, \dots, n\}$.

\subsection{Combinatorial Market Economies}
The following model for combinatorial markets based on a package assignment model was introduced in \cite{BikhchandaniOstroy:ThePackageAssignmentModel}.
The set of agents $N$ in this \emph{combinatorial market} is split into two disjoint groups: sellers and buyers.
The sellers offer (bundles of) indivisible items to the buyers, who can trade an arbitrarily divisible amount of money for the bundles they request.
We typically denote the sellers by $J = \{j_1,\dots, j_{\card{J}}\}$ and the buyers by $I = \{i_1,\dots, i_{\card{I}}\}$.
Following the usual convention, we refer to a subset of agents $C \subseteq N$ as a \emph{coalition}. We will indicate with $C_I = C \cap I$ the subset of buyers in $C$ and with $C_J = C \cap J$ the subset of sellers.

For the items $A  = \{a_1, \dots, a_{\card{A}}\}$ that are traded in the market, multiple identical copies may exist.
Individually, every seller $j \in J$ has an endowment of items $\vecc{\Omega}_j = (\Omega_j^1, \dots, \Omega_j^{\card{A}})$, i.e., $\Omega_j^\ell \in \N_{\geq 0}$ is the number of units of item $a_\ell \in A$ that $j$ initially owns.
The \emph{total endowment} of the market is denoted by $\vecc \Omega \in \N^{\card{A}}$ and corresponds to the sum of all the individual endowments.
We do not consider budget restrictions for the buyers (or, equivalently, assume sufficient money to be available to them to afford any bundles they request).

While money holds an equivalent unit value for every agent, agents individually assign value to bundles, measured in terms of monetary units.
We define the \emph{valuation function} $v_i : 2^{\vecc{\Omega}} \rightarrow \R_{\geq 0}$ of buyer $i \in I$ as the function that specifies the agent's valuation of any bundle $\vecc{\omega} \leq \vecc{\Omega}$.\footnote{With an abuse of notation, we indicate with $2^{\vecc{\Omega}}$ the power set of the set of items corresponding to $\vecc\Omega$.}
Analogously, we define the valuation function $v_j : 2^{\vecc{\Omega}_j} \rightarrow \R_{\geq 0}$ of seller $j \in J$.
We assume that the valuation of each agent $k \in N$ is normalized for the empty bundle and that free disposal of items applies, i.e., $v_k(\vecc 0) = 0$ and $v_k(\vecc{\omega}_1) \leq v_k(\vecc{\omega}_2)$ for every $\vecc{\omega}_1 \leq \vecc{\omega}_2 \leq \vecc{\Omega}$.
While we simply restrict the domain of the total items offered by each seller $j$ to their respective endowment $\vecc{\Omega}_j$, this can alternatively be accomplished (as, e.g., in \ \cite{BikhchandaniOstroy:ThePackageAssignmentModel}) by setting the valuation to infinity for all bundles that exceed their individual endowment or (in order to remain within $\R$) to exceed any buyer's budget.

The crucial feature of this model of combinatorial markets is that trades are not directly executed between seller-buyer pairs, but instead through a \emph{clearinghouse} that receives bundles from sellers and, under certain restrictions (discussed later), potentially repackages them into new bundles before shipping them to buyers. It also collects all payments from buyers and distributes them to sellers in accordance with the market pricing system (as defined below).

An \emph{economy} in our model is then uniquely determined by the set of agents, the initial endowments of the sellers, and the agents' valuation functions, and we will denote it by $\mathcal{E} = (N, (\vecc{\Omega}_j)_{j \in J}, (v_k)_{k \in N})$. For a coalition $C \subseteq N$, we use the notation $\mathcal{E}[C] = (C, (\vecc{\Omega}_j)_{j \in C_J}, (v_k)_{k \in C})$ to indicate the economy where the agents in $C$ only trade among each other, without considering the endowments and the valuations of the agents outside of the coalition. By a \emph{market} $\mathcal{E}_m$, we denote the underlying economy $\mathcal{E}$ together with a specific setting of market restrictions as outlined in Appendix~\ref{secapp:formal_market-settings}.

\subsubsection*{Allocations}
Next, we introduce the notation to describe the exchange of bundles between buyers and sellers.
Every seller $j \in J$ can offer one bundle of items to each potential buyer, which we denote $\vecc{z}_{ij} = (z_{ij}^1, \dots, z_{ij}^{\card{A}}) \leq \vecc{\Omega}_j$ for each buyer $i \in I$.
The indices $i$ and $j$ in $\vecc{z}_{ij}$ correspond, respectively, to the buyer's and the seller's \emph{labels} on the bundle, i.e., the designated recipient and the sender of the bundle.
If a seller does not offer any bundle to a specific buyer, we use the empty bundle $z_{ij} = \vecc 0$.
Then, the clearinghouse receives all these bundles and can potentially repackage or relabel them, depending on the setting.
Similarly, every buyer $i \in I$ requests one bundle $\vecc{s}_{ij} = (s_{ij}^1, \dots, s_{ij}^{\card{A}}) \leq \vecc{\Omega}$ from each seller $j \in J$.
It is important to highlight that, due to the repackaging, $\vecc s_{ij}$ and $\vecc{z}_{ij}$ do not necessarily correspond to each other, but the market must remain feasible in total.
We establish this further below with different repackaging restrictions.

We denote by $\vecc{S}_i = (\vecc{s}_{ij})_{j \in J}$ and $\vecc{Z}_j = (\vecc{z}_{ij})_{i \in I}$ the vectors of individual bundles that buyer $i$ requests and seller $j$ offers, and the respective cumulative bundles of each agent by $\vecc{s}_i = \sum_{j \in J} \vecc{s}_{ij} = (s_i^1, \dots, s_i^{\card{A}}) \leq \vecc{\Omega}$ and $\vecc{z}_j = \sum_{i \in I} \vecc{z}_{ij} = (z_j^1, \dots, z_j^{\card{A}}) \leq \vecc{\Omega}_j$.
Note that the last inequality establishes that the bundles are feasible with respect to the individual endowment of each seller.
We refer to the tuple $(\vecc{S}_I, \vecc{Z}_J)$ of all the individual bundles that the agents trade as an \emph{allocation} in the economy.
The \emph{social welfare} of an allocation is then defined as the net gain in valuations from the trade across all agents, i.e.,
$$\sum_{i \in I} v_i(\vecc{s}_i) - \sum_{j \in J} v_j(\vecc{z}_j).$$

\subsubsection*{Prices}
We define the \emph{pricing system} $\vecc p$ in the economy $\mathcal{E}$ as functions $p_{ij} : 2^{\vecc{\Omega}_j} \rightarrow \R_{\geq 0}$ that represent the price offers for every bundle $\vecc{\omega} \leq \vecc{\Omega}_j$ being offered from seller $j \in J$ to buyer $i \in I$.
Unlike the allocations, this pricing system defines prices for \emph{both} buyers and sellers simultaneously.
We normalize the prices of the empty bundle to zero, i.e., $p_{ij}(\vecc{0}) = 0$ for every $i \in I$ and $j \in J$.

Given an allocation $(\vecc{S}_I, \vecc{Z}_J)$ and a pricing system $\vecc p$, each buyer $i \in I$ sends the appropriate amount of money $p_{ij}(\vecc{s}_{ij})$ to the clearinghouse for each bundle $\vecc{s}_{ij}$ they request from sellers $j \in J$.
Then, the clearinghouse sends payments of $p_{ij}(\vecc{z}_{ij})$ to the sellers $j\in J$ for every bundle $z_{ij}$ they offer to buyers $i\in I$.
With a slight abuse of notation, we define the cumulative prices associated with the allocation $(\vecc{S}_I, \vecc{Z}_J)$ as $p_i = p_i(\vecc{S}_i) = \sum_{j \in J} p_{ij}(\vecc{s}_{ij})$ and $p_j = p_j(\vecc{Z}_j) = \sum_{i \in I} p_{ij}(\vecc{z}_{ij})$.
An essential assumption of our model is that the \emph{budget balance} condition holds, which asserts that the clearinghouse is financially neutral and does not take or inject any money into the market, i.e.,
$ \sum_{i \in I} p_i = \sum_{j \in J} p_j$.

\subsection{Market Settings}\label{secapp:formal_market-settings}

Now, we are ready to formally introduce the market settings that are analyzed in this work.
These settings differ in the possible alterations the clearinghouse can make to the packages it receives and in the restrictions on the pricing system.
The first four settings were first defined in \cite{BikhchandaniOstroy:ThePackageAssignmentModel}, whereas the fifth setting was introduced in \cite{BichlerWaldherr:CorePricingCombExchange}.
While the former states all settings in terms of inequalities (for the purpose of LP duality), we directly state the equality versions of the admissibility conditions.

\subsubsection*{First-Order Setting}
This setting gives the most power to the clearinghouse, which can both repackage and reassign the bundles, while any package labels are completely disregarded.
The admissibility of an allocation $(\vecc{S}_I, \vecc{Z}_J)$ therefore reduces to every individual item being offered and requested at the same count, i.e.,
\begin{equation*}
  \sum_{\vecc{\omega} \leq \vecc{\Omega}} \left(\sum_{\{(i,j) \, : \, \vecc{s}_{ij} = \vecc{\omega}\}} \vecc{s}_{ij}\right) = \sum_{\vecc{\omega} \leq \vecc{\Omega}} \left(\sum_{\{(i,j) \, : \, \vecc{z}_{ij} = \vecc{\omega}\}} \vecc{z}_{ij}\right).
\end{equation*}
Note that this is a \emph{system} of equalities with a scalar equality for every distinct item, as the bundles are given as vectors of item counts. It can alternatively be stated in terms of cumulative bundles as $$\sum_{i \in I} \vecc{s}_{i} = \sum_{j \in J} \vecc{z}_{j}.$$

A first-order pricing system $\vecc{p}$ is based on common item prices, independently of how they are bundled.
The price for every bundle is determined as a linear combination of the bundle with a shared price vector $\vecc{p} \in \R_{\geq 0}^{\card{A}}$, where $p_a$ is the price of one unit of item $a$, i.e., $$p_{ij}(\vecc{\omega}) = \vecc{p} \cdot \vecc{\omega} = \sum_{a\in A} p_a \omega_a$$ for every $\vecc{\omega} \leq \vecc{\Omega}$, $i \in I$ and $j \in J$.

\subsubsection*{Second-Order Setting}
Here, the clearinghouse can reassign entire bundles by disregarding the package labels of both sides, but not open any package to create new bundles.
An allocation $(\vecc{S}_I, \vecc{Z}_J)$ is admissible in this setting if and only if all the bundles requested by the buyers are offered by the sellers, i.e.,
\begin{equation*}
  \sum_{\{(i,j) \, : \, \vecc{s}_{ij} = \vecc{\omega}\}} \vecc{s}_{ij} = \sum_{\{(i,j) \, : \, \vecc{z}_{ij} = \vecc{\omega}\}} \vecc{z}_{ij}, \quad \forall \vecc{\omega} \leq \vecc{\Omega}.
\end{equation*}

The second-order pricing system is bundle-specific, i.e., $p_{ij}(\vecc{\omega}) = p(\vecc{\omega})$ for every $\vecc{\omega} \leq \vecc{\Omega}$, $i \in I$ and $j \in J$, and all agents share the same price for offering or requesting any distinct bundle.

\subsubsection*{Third-Order Setting}
In this setting, the labels on the packages are taken into consideration for only one side of the market.
There are, therefore, two variants of this setting, but, like \cite{BikhchandaniOstroy:ThePackageAssignmentModel}, we restrict ourselves to considering the \emph{buyers'} labels, i.e., every bundle has a buyer's label on it and can only be sent to that specific buyer.
The clearinghouse can, however, put a different seller's label on a package. If a buyer requests a specific bundle from seller $j$, but it is only offered from seller $j'$, the clearinghouse can adapt the label on this package to $j$ and thus fulfill the buyer's request.

An allocation $(\vecc{S}_I, \vecc{Z}_J)$ is admissible in this setting if all the bundles requested from buyers are offered \emph{to them} by some seller, i.e.,
\begin{equation*}
  \sum_{\{j \, : \, \vecc{s}_{ij} = \vecc{\omega}\}} \vecc{s}_{ij} = \sum_{\{j \, : \, \vecc{z}_{ij} = \vecc{\omega}\}} \vecc{z}_{ij}, \quad \forall \vecc{\omega} \leq \vecc{\Omega}, \, i \in I.
\end{equation*}

The corresponding third-order pricing system is bundle- and buyer-specific, i.e., $p_{ij}(\vecc{\omega}) = p_i(\vecc{\omega})$ for every $\vecc{\omega} \leq \vecc{\Omega}$, $i \in I$ and $j \in J$.
Therefore, the same bundle can only sell for a different price if it is directed to another buyer.

\subsubsection*{Fourth-Order Setting}
In the fourth and last setting from \cite{BikhchandaniOstroy:ThePackageAssignmentModel}, both the sellers' and buyers' labels apply, and the clearinghouse does not repackage or relabel any bundles.
An allocation $(\vecc{S}_I, \vecc{Z}_J)$ is admissible if and only if both sides of the allocation correspond directly, i.e., the sellers offer the bundles specifically to those buyers who also request those specific bundles from them:
\begin{equation*}
  \vecc{s}_{ij} = \vecc{z}_{ij}, \quad \forall i \in I, \, j \in J.
\end{equation*}
The fourth-order pricing system is fully personalized with bundle-, buyer- and seller-specific prices $p_{ij}(\vecc{\omega})$ for every $\vecc{\omega} \leq \vecc{\Omega}$, $i \in I$ and $j \in J$.

\subsubsection*{Fifth-Order Setting}
This setting was introduced in \cite{BichlerWaldherr:CorePricingCombExchange} and is different from the previous four in that the agents only trade their cumulative bundle with the clearinghouse ($\vecc{s}_i$ for the buyers and $\vecc{z}_j$ for the sellers), without any indication of labels.
The clearinghouse can repackage the submitted bundles at the item level, similar to the first-order setting.
Therefore, an allocation $(\vecc{S}_I, \vecc{Z}_J) = (\vecc{s}_I, \vecc{z}_J)$ is admissible if and only if all the items requested by the buyers are provided by the sellers, i.e.,
\begin{equation*}
  \sum_{i \in I} \vecc{s}_{i} = \sum_{j \in J} \vecc{z}_{j}.
\end{equation*}
The fifth-order pricing system, however, is given as fully personalized bundle- and agent-specific prices, i.e., $p_k(\vecc{\omega})$ for every $\vecc{\omega} \leq \vecc{\Omega}$ and $k \in N$.
Every agent only sends or receives a single payment to the clearinghouse, so the abuse of notation with the cumulative prices $p_k$ from above is conveniently consistent with this structurally different setting.

\section{Summary of Equivalence Results}\label{secapp:tables}
The equivalence results across the different settings are summarized in Tables~\ref{table:to1} to \ref{table:to5}. The symbol $\Pi$ indicates that every assignment in the market setting corresponding to a given row admits a payoff-equivalent assignment in the market setting corresponding to the respective column. Similarly, the letters P and M denote price equivalence and market equivalence, respectively.

\begin{table}[ht]
  \centering
  \small
  \begin{tabular}{|cc|ccccc|}
    \hline
    & To
    & \multicolumn{5}{c|}{{${1^\mathrm{st}}$}} \\
    \cline{3-7}

    From
    &
    & \multicolumn{1}{c|}{{Full}}
    & \multicolumn{1}{c|}{{Separate}}
    & \multicolumn{1}{c|}{{Buyers}}
    & \multicolumn{1}{c|}{{Sellers}}
    & \multicolumn{1}{c|}{{No}} \\
    \hline

    \multirow{5}{*}{{${1^\mathrm{st}}$ }}
    & \multicolumn{1}{|c|}{{Full}}
    & \multicolumn{1}{c|}{}
    & \multicolumn{1}{c|}{$\Pi$}
    &
    & \multirow{3}{*}{}
    & \\
    \cline{2-2}\cline{4-4}
    & \multicolumn{1}{|c|}{{Separate}}
    & \multicolumn{2}{c|}{\multirow{4}{*}{M}}
    &
    &
    & \\
    \cline{2-2}\cline{5-5}
    & \multicolumn{1}{|c|}{{Buyers}}
    &
    &
    & \multicolumn{1}{c|}{}
    &
    & \\
    \cline{2-2}\cline{5-6}
    & \multicolumn{1}{|c|}{{Sellers}}
    &
    & \multicolumn{1}{c|}{}
    & \multicolumn{1}{c|}{}
    & \multicolumn{1}{c|}{}
    & \\
    \cline{2-2}\cline{5-5}\cline{7-7}
    & \multicolumn{1}{|c|}{{No}}
    &
    &
    &
    &
    & \\
    \hline

    {\makecell{${2^\mathrm{nd}}$  \\ {${3^\mathrm{rd}}$ } \\ {${4^\mathrm{th}}$ } \\ {${5^\mathrm{th}}$ }}}
    & \multicolumn{1}{|c|}{{All Side Payments}}
    & \multicolumn{2}{c|}{$\Pi$}
    &
    &
    & \\
    \hline
  \end{tabular}

  \caption{Equivalence to first-order admissible assignments under different side-payment restrictions.}
  \label{table:to1}
\end{table}

\begin{table}[ht]
  \centering
  \small
  \begin{tabular}{|cc|ccccc|}
    \hline
    & To
    & \multicolumn{5}{c|}{{${2^\mathrm{nd}}$ }} \\
    \cline{3-7}

    From
    &
    & \multicolumn{1}{c|}{{Full}}
    & \multicolumn{1}{c|}{{Separate}}
    & \multicolumn{1}{c|}{{Buyers}}
    & \multicolumn{1}{c|}{{Sellers}}
    & \multicolumn{1}{c|}{{No}} \\
    \hline

    \multirow{5}{*}{\makecell{{${1^\mathrm{st}}$ } \\ {${2^\mathrm{nd}}$ }}}
    & \multicolumn{1}{|c|}{{Full}}
    & \multicolumn{1}{c|}{}
    & \multicolumn{1}{c|}{$\Pi$}
    &
    & \multirow{3}{*}{}
    & \\
    \cline{2-2}\cline{4-4}
    & \multicolumn{1}{|c|}{{Separate}}
    & \multicolumn{2}{c|}{\multirow{4}{*}{M}}
    &
    &
    & \\
    \cline{2-2}\cline{5-5}
    & \multicolumn{1}{|c|}{{Buyers}}
    &
    &
    & \multicolumn{1}{c|}{}
    &
    & \\
    \cline{2-2}\cline{5-6}
    & \multicolumn{1}{|c|}{{Sellers}}
    &
    & \multicolumn{1}{c|}{}
    & \multicolumn{1}{c|}{}
    & \multicolumn{1}{c|}{}
    & \\
    \cline{2-2}\cline{5-5}\cline{7-7}
    & \multicolumn{1}{|c|}{{No}}
    &
    &
    &
    &
    & \\
    \hline

    {\makecell{{${3^\mathrm{rd}}$ } \\ {${4^\mathrm{th}}$ } \\ {${5^\mathrm{th}}$ }}}
    & \multicolumn{1}{|c|}{{All Side Payments}}
    & \multicolumn{2}{c|}{$\Pi$}
    &
    &
    & \\
    \hline
  \end{tabular}

  \caption{Equivalence to second-order admissible assignments under different side-payment restrictions.}
  \label{table:to2}
\end{table}

\begin{table}[H]
  \centering
  \small
  \begin{tabular}{|cc|ccccc|}
    \hline
    & To
    & \multicolumn{5}{c|}{{${3^\mathrm{rd}}$ }} \\
    \cline{3-7}

    From
    &
    & \multicolumn{1}{c|}{{Full}}
    & \multicolumn{1}{c|}{{Separate}}
    & \multicolumn{1}{c|}{{Buyers}}
    & \multicolumn{1}{c|}{{Sellers}}
    & \multicolumn{1}{c|}{{No}} \\
    \hline

    \multirow{5}{*}{\makecell{{${1^\mathrm{st}}$ } \\ {${3^\mathrm{rd}}$ }}}
    & \multicolumn{1}{|c|}{{Full}}
    & \multicolumn{1}{c|}{}
    & \multicolumn{1}{c|}{$\Pi$}
    &
    & \multirow{3}{*}{}
    & \\
    \cline{2-2}\cline{4-4}
    & \multicolumn{1}{|c|}{{Separate}}
    & \multicolumn{2}{c|}{\multirow{4}{*}{M}}
    &
    &
    & \\
    \cline{2-2}\cline{5-5}
    & \multicolumn{1}{|c|}{{Buyers}}
    &
    &
    & \multicolumn{1}{c|}{}
    &
    & \\
    \cline{2-2}\cline{5-6}
    & \multicolumn{1}{|c|}{{Sellers}}
    &
    & \multicolumn{1}{c|}{}
    & \multicolumn{1}{c|}{}
    & \multicolumn{1}{c|}{}
    & \\
    \cline{2-2}\cline{5-5}\cline{7-7}
    & \multicolumn{1}{|c|}{{No}}
    &
    &
    &
    &
    & \\
    \hline

    \multirow{5}{*}{\makecell{{${2^\mathrm{nd}}$ } \\ {${4^\mathrm{th}}$ } \\ {${5^\mathrm{th}}$ }}}
    & \multicolumn{1}{|c|}{{Full}}
    & \multicolumn{2}{c|}{\multirow{5}{*}{$\Pi$}}
    &
    &
    & \\
    \cline{2-2}
    & \multicolumn{1}{|c|}{{Separate}}
    &
    & \multicolumn{1}{c|}{}
    &
    &
    & \\
    \cline{2-2}
    & \multicolumn{1}{|c|}{{Buyers}}
    &
    & \multicolumn{1}{c|}{}
    &
    &
    & \\
    \cline{2-2}\cline{6-6}
    & \multicolumn{1}{|c|}{{Sellers}}
    &
    & \multicolumn{1}{c|}{}
    & \multicolumn{1}{c|}{}
    & \multicolumn{1}{c|}{\multirow{2}{*}{$\Pi$}}
    & \\
    \cline{2-2}
    & \multicolumn{1}{|c|}{{No}}
    & \multicolumn{1}{c}{}
    & \multicolumn{1}{c|}{}
    & \multicolumn{1}{c|}{}
    & \multicolumn{1}{c|}{}
    & \\
    \hline
  \end{tabular}

  \caption{Equivalence to third-order admissible assignments under different side-payment restrictions.}
  \label{table:to3}
\end{table}

\begin{table}[H]
  \centering
  \small
  \begin{tabular}{|cc|ccccc|}
    \hline
    & To
    & \multicolumn{5}{c|}{{${4^\mathrm{th}}$ }} \\
    \cline{3-7}

    From
    &
    & \multicolumn{1}{c|}{{Full}}
    & \multicolumn{1}{c|}{{Separate}}
    & \multicolumn{1}{c|}{{Buyers}}
    & \multicolumn{1}{c|}{{Sellers}}
    & \multicolumn{1}{c|}{{No}} \\
    \hline

    \multirow{5}{*}{\makecell{{${1^\mathrm{st}}$ } \\ {${3^\mathrm{rd}}$ } \\ {${4^\mathrm{th}}$ }}}
    & \multicolumn{1}{|c|}{{Full}}
    & \multicolumn{1}{c|}{}
    & \multicolumn{1}{c|}{$\Pi$}
    &
    & \multirow{3}{*}{}
    & \\
    \cline{2-2}\cline{4-4}
    & \multicolumn{1}{|c|}{{Separate}}
    & \multicolumn{2}{c|}{\multirow{4}{*}{M}}
    &
    &
    & \\
    \cline{2-2}\cline{5-5}
    & \multicolumn{1}{|c|}{{Buyers}}
    &
    &
    & \multicolumn{1}{c|}{}
    &
    & \\
    \cline{2-2}\cline{5-6}
    & \multicolumn{1}{|c|}{{Sellers}}
    &
    & \multicolumn{1}{c|}{}
    & \multicolumn{1}{c|}{}
    & \multicolumn{1}{c|}{}
    & \\
    \cline{2-2}\cline{5-5}\cline{7-7}
    & \multicolumn{1}{|c|}{{No}}
    &
    &
    &
    &
    & \\
    \hline

    \multirow{5}{*}{{${2^\mathrm{nd}}$ }}
    & \multicolumn{1}{|c|}{{Full}}
    & \multicolumn{1}{c|}{}
    & \multicolumn{1}{c|}{$\Pi$}
    &
    & \multirow{3}{*}{}
    & \\
    \cline{2-2}\cline{4-4}
    & \multicolumn{1}{|c|}{{Separate}}
    & \multicolumn{2}{c|}{\multirow{4}{*}{P}}
    &
    &
    & \\
    \cline{2-2}\cline{5-5}
    & \multicolumn{1}{|c|}{{Buyers}}
    &
    &
    & \multicolumn{1}{c|}{}
    &
    & \\
    \cline{2-2}\cline{5-6}
    & \multicolumn{1}{|c|}{{Sellers}}
    &
    & \multicolumn{1}{c|}{}
    & \multicolumn{1}{c|}{}
    & \multicolumn{1}{c|}{}
    & \\
    \cline{2-2}\cline{5-5}\cline{7-7}
    & \multicolumn{1}{|c|}{{No}}
    &
    &
    &
    &
    & \\
    \hline

    \multirow{5}{*}{{${5^\mathrm{th}}$ }}
    & \multicolumn{1}{|c|}{{Full}}
    & \multicolumn{2}{c|}{\multirow{5}{*}{$\Pi$}}
    &
    &
    & \\
    \cline{2-2}
    & \multicolumn{1}{|c|}{{Separate}}
    &
    & \multicolumn{1}{c|}{}
    &
    &
    & \\
    \cline{2-2}\cline{5-5}
    & \multicolumn{1}{|c|}{{Buyers}}
    &
    &
    & \multicolumn{1}{c|}{}
    &
    & \\
    \cline{2-2}\cline{5-6}
    & \multicolumn{1}{|c|}{{Sellers}}
    &
    & \multicolumn{1}{c|}{}
    & \multicolumn{1}{c|}{}
    & \multicolumn{1}{c|}{}
    & \\
    \cline{2-2}\cline{5-5}
    & \multicolumn{1}{|c|}{{No}}
    &
    &
    &
    & \multicolumn{1}{c|}{}
    & \\
    \hline
  \end{tabular}

  \caption{Equivalence to fourth-order admissible assignments under different side-payment restrictions.}
  \label{table:to4}
\end{table}

\begin{table}[H]
  \centering
  \begin{tabular}{|cc|ccccc|}
    \hline
    & To
    & \multicolumn{5}{c|}{{${5^\mathrm{th}}$ }} \\
    \cline{3-7}

    From
    &
    & \multicolumn{1}{c|}{{Full}}
    & \multicolumn{1}{c|}{{Separate}}
    & \multicolumn{1}{c|}{{Buyers}}
    & \multicolumn{1}{c|}{{Sellers}}
    & \multicolumn{1}{c|}{{No}} \\
    \hline

    \multirow{5}{*}{\makecell{{${1^\mathrm{st}}$ } \\ {${2^\mathrm{nd}}$ } \\ {${3^\mathrm{rd}}$ } \\ {${4^\mathrm{th}}$ }}}
    & \multicolumn{1}{|c|}{{Full}}
    & \multicolumn{1}{c|}{}
    & \multicolumn{1}{c|}{$\Pi$}
    &
    & \multirow{3}{*}{}
    & \\
    \cline{2-2}\cline{4-4}
    & \multicolumn{1}{|c|}{{Separate}}
    & \multicolumn{2}{c|}{\multirow{4}{*}{P}}
    &
    &
    & \\
    \cline{2-2}\cline{5-5}
    & \multicolumn{1}{|c|}{{Buyers}}
    &
    &
    & \multicolumn{1}{c|}{}
    &
    & \\
    \cline{2-2}\cline{5-6}
    & \multicolumn{1}{|c|}{{Sellers}}
    &
    & \multicolumn{1}{c|}{}
    & \multicolumn{1}{c|}{}
    & \multicolumn{1}{c|}{}
    & \\
    \cline{2-2}\cline{5-5}\cline{7-7}
    & \multicolumn{1}{|c|}{{No}}
    &
    &
    &
    &
    & \\
    \hline

    \multirow{5}{*}{{${5^\mathrm{th}}$ }}
    & \multicolumn{1}{|c|}{{Full}}
    & \multicolumn{1}{c|}{}
    & \multicolumn{1}{c|}{$\Pi$}
    &
    & \multirow{3}{*}{}
    & \\
    \cline{2-2}\cline{4-4}
    & \multicolumn{1}{|c|}{{Separate}}
    & \multicolumn{2}{c|}{\multirow{4}{*}{M}}
    &
    &
    & \\
    \cline{2-2}\cline{5-5}
    & \multicolumn{1}{|c|}{{Buyers}}
    &
    &
    & \multicolumn{1}{c|}{}
    &
    & \\
    \cline{2-2}\cline{5-6}
    & \multicolumn{1}{|c|}{{Sellers}}
    &
    & \multicolumn{1}{c|}{}
    & \multicolumn{1}{c|}{}
    & \multicolumn{1}{c|}{}
    & \\
    \cline{2-2}\cline{5-5}\cline{7-7}
    & \multicolumn{1}{|c|}{{No}}
    &
    &
    &
    &
    & \\
    \hline
  \end{tabular}

  \caption{Equivalence to fifth-order admissible assignments under different side-payment restrictions.}
  \label{table:to5}
\end{table}

\section{Examples}

In this section, we provide various examples of markets that serve as counterexamples and justification for the stated non-equivalences between different settings.

\subsection{Equivalence Counterexamples}
The following is an outline of the counterexamples illustrated in this section, which cover possible equivalence implications that \emph{fail} to hold across the different market settings. We proceed by considering the settings in order of increasing levels of resale restrictions, and conclude with examples that account for the remaining cases.
\begin{enumerate}
  \item Example~\ref{example:2345noto1(I)SingNo} shows that, in some economies $\mathcal{E}$, there exist feasible $h^{\mathrm{th}}$-order admissible assignments for $h \in \{2, \dots, 5\}$ and any type of side payments such that
    \begin{enumerate}
      \item they are not price equivalent to any feasible first-order admissible assignment with a full or separate side-payment scheme, and
      \item they are not payoff equivalent to any feasible first-order admissible assignment with a (buyers' or sellers') single or no side-payment scheme.
    \end{enumerate}

  \item Example~\ref{example:345noto2(IPPiM)SingNo} shows that, in some economies $\mathcal{E}$, there exist feasible $h^{\mathrm{th}}$-order admissible assignments for $h \in \{3, 4, 5\}$ and any type of side payments such that
    \begin{enumerate}
      \item they are not price equivalent to any feasible second-order admissible assignment with a full or separate side-payment scheme, and
      \item they are not payoff equivalent to any feasible second-order admissible assignment with a (buyers' or sellers') single or no side-payment scheme.
    \end{enumerate}

  \item Example~\ref{example:34noto2(M)} shows that, in some economies $\mathcal{E}$, there exist feasible second-order admissible assignments with any type of side payments such that
    \begin{enumerate}
      \item they are not price equivalent to any feasible third-order admissible assignment with a full, separate, or sellers' single side payment scheme,
      \item they are not payoff equivalent to any feasible third-order admissible assignment with a buyers' single or no side-payment scheme, and
      \item they are not market equivalent to any feasible fourth-order admissible assignment with any type of side-payment scheme.
    \end{enumerate}

  \item Example~\ref{example:45noto3(IPPiM)BuyNo} shows that, in some economies $\mathcal{E}$, there exist feasible $h^{\mathrm{th}}$-order admissible assignments for $h \in \{4, 5\}$ and any type of side payments such that
    \begin{enumerate}
      \item they are not price equivalent to any feasible third-order admissible assignment with a full, separate, or sellers' single side payment scheme,
      \item they are not payoff equivalent to any feasible third-order admissible assignment with a buyers' single or no side-payment scheme.
    \end{enumerate}

  \item Example~\ref{example:5noto4(IPi)No} shows that, in some economies $\mathcal{E}$, there exist feasible fifth-order admissible assignments with any type of side payments such that
    \begin{enumerate}
      \item they are not price equivalent to any feasible fourth-order admissible assignment with a full, separate, or (buyers' or sellers') single side payment scheme, and
      \item they are not payoff equivalent to any feasible fourth-order admissible assignment with a no-side payment scheme.
    \end{enumerate}

  \item Example~\ref{example:FullnottoSep(P)} shows that, in some economies $\mathcal{E}$, there exist feasible $h^{\mathrm{th}}$-order admissible assignments with full side payments that are not price equivalent to any feasible $\tilde h^{\mathrm{th}}$-order admissible assignment with (at most) separate side payments, for every $h, \tilde h \in \{1, \dots, 5\}$.

  \item Example~\ref{example:1BuySepFullnotto...} shows that, in some economies $\mathcal{E}$, there exists no feasible assignment that is admissible in any of the following settings:
    \begin{enumerate}
      \item first- and second-order settings with sellers' single or no side payments,
      \item third- and fourth-order settings with no side payments,
      \item third- and fourth-order settings with sellers' single side payments, and
      \item fifth-order settings with sellers' single or no side payments,
    \end{enumerate}
    and that is payoff equivalent to any feasible $h^{\mathrm{th}}$-order admissible assignment with $h \in \{1, \dots, 5\}$ and a full, separate, or buyers' single side-payment scheme.

  \item Example~\ref{example:1SellSepFullnotto...} shows that, in some economies $\mathcal{E}$, there exists no feasible assignment that is admissible in any of the following settings:
    \begin{enumerate}
      \item first-, second-, and third-order settings with buyers' single or no side payments,
      \item fourth-order setting with no side payments,
      \item fourth-order setting with buyers' single side payments, and
      \item fifth-order settings with buyers' single or no side payments,
    \end{enumerate}
    and that is payoff equivalent to any feasible $h^{\mathrm{th}}$-order admissible assignment with $h \in \{1, \dots, 5\}$ and a full, separate, or sellers' single side-payment scheme.
\end{enumerate}

\subsubsection{To the First-Order Settings}
\begin{example}\label{example:2345noto1(I)SingNo}
  Consider an economy $\mathcal{E}$ with two buyers $I = \{i_1, i_2\}$ and two sellers $J = \{j_1, j_2\}$. The sellers' initial endowments are $\vecc\Omega_{j_1} \equiv \{a\}$ and $\vecc\Omega_{j_2} \equiv \{a, a\}$, so that the total endowment of the economy is $\vecc\Omega \equiv \{a, a, a\}$. The agents' valuations are reported in Table~\ref{table:2345noto1(I)SingNo_valuations}.

  \begin{table}[ht]
    \centering
    \begin{tabular}{|c|c|c|c|}
      \hline
      $\vecc{\omega}$ & $\{a\}$ & $\{a, a\}$ & $\{a, a, a\}$ \\
      \hline
      $v_{i_1}(\vecc{\omega})$ & $6$ & $6$ & $6$ \\
      \hline
      $v_{i_2}(\vecc{\omega})$ & $0$ & $17$ & $17$ \\
      \hline
      $v_{j_1}(\vecc{\omega})$ & $5$ & $-$ & $-$ \\
      \hline
      $v_{j_2}(\vecc{\omega})$ & $16$ & $16$ & $-$ \\
      \hline
    \end{tabular}

    \caption{Valuations $v_k(\vecc{\omega})$ for each agent $k \in N$ and each bundle $\vecc{\omega} \leq \vecc{\Omega}$ in Example~\ref{example:2345noto1(I)SingNo}.}
    \label{table:2345noto1(I)SingNo_valuations}
  \end{table}

  We first define a feasible assignment that is admissible in every market setting of $\mathcal{E}$ with at least second-order resale restrictions. Let $(\vecc{S}_I, \vecc{Z}_J)$ be the allocation in which agents $i_1$ and $j_1$ trade the bundle $\vecc{s}_{i_1j_1} = \vecc{z}_{i_1j_1} \equiv \{a\}$, while agents $i_2$ and $j_2$ trade the bundle $\vecc{s}_{i_2j_2} = \vecc{z}_{i_2j_2} \equiv \{a, a\}$. The corresponding pricing system $\vecc{p}$ is reported in Table~\ref{table:2345noto1(I)SingNo_pricing}, where $M \gg 0$.

  \begin{table}[ht]
    \centering
    \begin{tabular}{|c|c|c|c|}
      \hline
      $\vecc{\omega}$ & $\{a\}$ & $\{a, a\}$ & $\{a, a, a\}$ \\
      \hline
      $p(\vecc{\omega})$ & $5.5$ & $16.5$ & $M$ \\
      \hline
    \end{tabular}

    \caption{Second-order admissible pricing system $\vecc p$ in Example~\ref{example:2345noto1(I)SingNo}.}
    \label{table:2345noto1(I)SingNo_pricing}
  \end{table}

  Since the allocation is fourth-order admissible and the pricing system is second-order admissible, the assignment $((\vecc{S}_I, \vecc{Z}_J), \vecc p)$ is second-, third-, and fourth-order admissible (and fifth-order admissible if we consider the cumulative equivalent assignment $((\vecc{s}_I, \vecc{z}_J), \vecc p)$).
  It is also admissible in all side-payment settings, since the corresponding $\vecc{Q}_N = \vecc 0$.
  The payoff is non-negative for each agent; in particular, $\pi_k = 0.5$ for every $k \in N$. Therefore, the assignment is feasible in all of the settings mentioned above. A representation of the payment graph is provided in Figure~\ref{figure:2345noto1(I)SingNo_pricing}.

  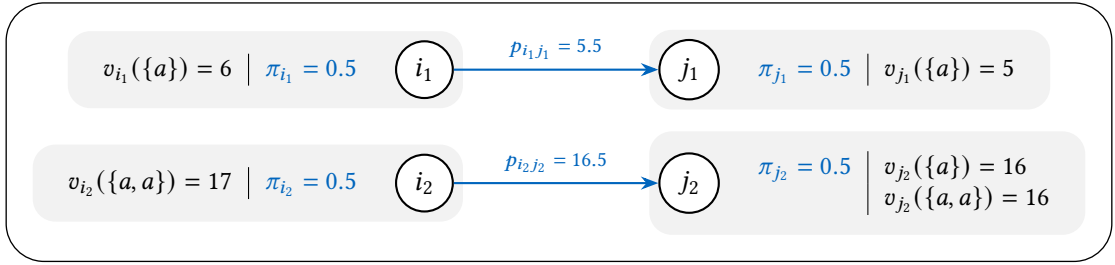
\begin{figure}[ht]
    \centering
    \begin{tikzpicture}[>=Stealth, thick, node/.style={draw, circle, fill=white, minimum size=0.7cm}]
      \begin{scope}[local bounding box=diagram]
        \node[node] (i1) at (0,1.5) {$i_1$};
        \node[left=0.15cm of i1, font=\small] (text-i1) {
          \begin{tabular}{r|l} $v_{i_1}(\{a\}) = 6$ & \textcolor{TUMBlue}{$\pi_{i_1} = 0.5$}
        \end{tabular}};

        \node[node] (i2) at (0,0) {$i_2$};
        \node[left=0.15cm of i2, font=\small] (text-i2) {
          \begin{tabular}{r|l} $v_{i_2}(\{a, a\}) = 17$ & \textcolor{TUMBlue}{$\pi_{i_2} = 0.5$}
        \end{tabular}};

        \node[node] (j1) at (3.5,1.5) {$j_1$};
        \node[right=0.15cm of j1, font=\small] (text-j1) {
          \begin{tabular}{r|l} \textcolor{TUMBlue}{$\pi_{j_1} = 0.5$} & $v_{j_1}(\{a\}) = 5$
        \end{tabular}};

        \node[node] (j2) at (3.5,0) {$j_2$};
        \node[right=0.15cm of j2, font=\small] (text-j2) {
          \begin{tabular}{r|l} \textcolor{TUMBlue}{$\pi_{j_2} = 0.5$} & $v_{j_2}(\{a\}) = 16$ \\ & $v_{j_2}(\{a, a\}) = 16$
        \end{tabular}};

        \begin{scope}[on background layer]
          \foreach \n/\t in {i1/text-i1, i2/text-i2, j1/text-j1, j2/text-j2}
          \node[fill=gray!10, rounded corners=3mm, inner sep=1.2mm, fit=(\n) (\t)] {};
        \end{scope}

        \draw[->, TUMBlue] (i1) -- node[midway, above, font=\scriptsize, text=TUMBlue] {$p_{i_1j_1} = 5.5$} (j1);
        \draw[->, TUMBlue] (i2) -- node[midway, above, font=\scriptsize, text=TUMBlue] {$p_{i_2j_2} = 16.5$} (j2);

        \begin{scope}[on background layer]
          \node[draw=black, line width=0.5pt, rounded corners=5mm, inner sep=10pt, fit=(diagram)] {};
        \end{scope}
      \end{scope}
    \end{tikzpicture}

    \caption{Graph representation of the assignment $((\vecc{S}_I, \vecc{Z}_J), \vecc{p})$ in Example~\ref{example:2345noto1(I)SingNo}.}
    \label{figure:2345noto1(I)SingNo_pricing}
  \end{figure}

  We consider next a generic first-order admissible assignment $((\vecc{\tilde S}_I, \vecc{\tilde Z}_J), \vecc{\tilde p}, \vecc{\tilde Q}_N) \in \mathcal{F}(\mathcal{E}_1^r)$ with $r \in \{\mathrm{full}, \, \mathrm{sep}, \, \mathrm{buy}, \, \mathrm{sell}, \, \mathrm{no}\}$.
  The first observation is that, for this assignment to be feasible, the agents' cumulative bundles $(\vecc{\tilde s}_I, \vecc{\tilde z}_J)$ must either coincide with the corresponding $\vecc{s}_i$ or $\vecc{z}_j$, or be empty. To see this, suppose by contradiction that buyer $i_2$ requests only the bundle $\{a\}$ from one of the sellers (a similar argument applies if seller $j_2$ offers only a single unit of item $a$). We distinguish two cases:
  \begin{enumerate}
    \item If $i_1$ requests the empty bundle, then admissibility implies that only one of the sellers needs to offer the bundle $\{a\}$. In this case, the total valuation of the buyers is $v_{i_1}(\emptyset) + v_{i_2}(\{a\}) = 0$, whereas the total valuation of the sellers is at least $\min\{v_{j_1}(\{a\}), v_{j_2}(\{a\})\} = 5$ (which corresponds to seller $j_1$ offering bundle $\{a\}$).
    \item If $i_1$ requests the bundle $\{a\}$, then the sellers must offer two units of item $a$. The total valuation of the buyers is now $v_{i_1}(\{a\}) + v_{i_2}(\{a\}) = 6$, while the total valuation of the sellers is at least $\min\{v_{j_1}(\{a\}) + v_{j_2}(\{a\}), v_{j_2}(\{a, a\})\} = 16$ (which corresponds to seller $j_2$ offering two copies of bundle $\{a\}$).
  \end{enumerate}
  In both cases, the social welfare of the allocation is negative, meaning that at least one agent gets a negative payoff. Therefore, in any feasible assignment, $i_2$'s cumulative bundle must be either the empty bundle or $\{a, a\}$ (and analogously for seller $j_2$).
  Following this observation, and since the price of the empty bundle is normalized to zero, a necessary condition for this assignment to yield the same cumulative price or payoff for each agent as $((\vecc{S}_I, \vecc{Z}_J), \vecc p)$ is that $(\vecc{\tilde s}_I, \vecc{\tilde z}_J) = (\vecc{s}_I, \vecc{z}_J)$, which we will always assume in the following.
  Regarding the admissibility of the pricing system, $\vecc{\tilde p}$ must assign a unique price to each item. Hence, if an agent $k \in N$ trades $m$ units of item $a$ in total, their cumulative price is $\tilde{p}_{k} = m\tilde{p}_a$.

  We now analyze the different side-payment restrictions for $((\vecc{\tilde S}_I, \vecc{\tilde Z}_J), \vecc{\tilde p}, \vecc{\tilde Q}_N)$ and establish for which of them there exists no such assignment that is price or payoff equivalent to $((\vecc{S}_I, \vecc{Z}_J), \vecc p)$.

  \paragraph{Full and separate side-payment settings} We have already proved in Lemma~\ref{lemma:12345to1(Pi)FullSep} that, for any assignment in any market setting, there exists a payoff-equivalent first-order admissible assignment with separate (and hence also full) side payments. For price equivalence to hold between $((\vecc{S}_I, \vecc{Z}_J), \vecc p)$ and an assignment $((\vecc{\tilde S}_I, \vecc{\tilde Z}_J), \vecc{\tilde p}, \vecc{\tilde Q}_N) \in \mathcal{F}(\mathcal{E}_1^r)$ with $r \in \{\mathrm{full}, \, \mathrm{sep}\}$, we must have, on the one hand,
  $$\tilde p_a = \tilde p_{i_1} = \tilde p_{j_1} = p_{i_1} = p_{j_1} = 5.5,$$
  and, on the other hand,
  $$2\tilde p_a = \tilde p_{i_2} = \tilde p_{j_2} = p_{i_2} = p_{j_2} = 16.5.$$
  However, there is no price $\tilde p_a \geq 0$ that satisfies both conditions. Therefore, no feasible first-order admissible assignment with full or separate side payments is price-equivalent to $((\vecc{S}_I, \vecc{Z}_J), \vecc p)$.

  \paragraph{Buyers' single side-payment setting} In the buyers' single side-payment setting, the price of item $a$ must be sufficient to pay every seller at least the valuation of their cumulative bundle, since they cannot receive money through side payments. In our case, this implies
  $$\tilde p_a \geq \max\left\{v_{j_1}(\vecc{\tilde z}_{j_1}), \frac{v_{j_2}(\vecc{\tilde z}_{j_2})}{2}\right\} = 8,$$
  which in turn implies that the buyers must pay in total
  $$\tilde p_{i_1} + \tilde p_{i_2} = 3\tilde p_a \geq 24.$$
  However, the total valuation of the buyers is only
  $$v_{i_1}(\vecc{\tilde s}_{i_1}) + v_{i_2}(\vecc{\tilde s}_{i_2}) = 23 < 3\tilde p_a,$$
  so there exists no first-order pricing system that makes $((\vecc{\tilde S}_I, \vecc{\tilde Z}_J), \vecc{\tilde p}, \vecc{\tilde Q}_N)$ a feasible assignment in $\mathcal{E}_1^{\mathrm{buy}}$. Therefore there is no feasible assignment in $\mathcal{E}_1^{\mathrm{buy}}$ that is even item equivalent to $((\vecc{S}_I, \vecc{Z}_J), \vecc p)$.\footnote{In every feasible assignment, if side payments among sellers are allowed, the minimum amount of money they must receive is their total valuation $\sum_{j \in J} v_j(\{\vecc{z}_j\})$, since they can then redistribute the money so that each seller gets a non-negative payoff. If side payments are not allowed, each seller must receive a cumulative price of at least $v_j(\{\vecc{z}_j\})$ to ensure feasibility. The same reasoning applies symmetrically to the buyers.}

  \paragraph{Sellers' single side-payment setting} In the sellers' single side-payment setting, since sellers can redistribute payments among themselves, the total amount of money they receive must be at least
  $$\tilde p_{j_1} + \tilde p_{j_2} = 3\tilde{p}_a \geq v_{j_1}(\vecc{\tilde z}_{j_1}) + v_{j_2}(\vecc{\tilde z}_{j_2}) = 21,$$
  which implies that $\tilde p_a \geq 7$. However, buyer $i_1$ can only pay at most $$\tilde{p}_a \leq v_{i_1}(\vecc{\tilde s}_{i_1}) = 6$$ for one unit of item $a$. Therefore, also in this setting, no first-order admissible pricing that makes $((\vecc{\tilde S}_I, \vecc{\tilde Z}_J), \vecc{\tilde p}, \vecc{\tilde Q}_N)$ a feasible assignment in $\mathcal{E}_1^{\mathrm{sell}}$ exists. Hence, there is no feasible assignment in $\mathcal{E}_1^{\mathrm{sell}}$ that is even item equivalent to $((\vecc{S}_I, \vecc{Z}_J), \vecc p)$.

  \paragraph{No side-payment setting} By combining the price conditions from the previous two settings, the same conclusions apply in the no side-payment setting.
\end{example}

\subsubsection{To the Second-Order Settings}
\begin{example}\label{example:345noto2(IPPiM)SingNo}
  Consider an economy $\mathcal{E}$ with two buyers $I = \{i_1, i_2\}$ and two sellers $J = \{j_1, j_2\}$. The sellers' initial endowments are $\vecc\Omega_{j_1} = \vecc\Omega_{j_2} \equiv \{a\}$, so that the total endowment of the economy is $\vecc\Omega \equiv \{a, a\}$. The agents' valuations are reported in Table~\ref{table:345noto2(IPPiM)SingNo_valuations}.

  \begin{table}[ht]
    \centering
    \begin{tabular}{|c|c|c|}
      \hline
      $\vecc{\omega}$ & $\{a\}$ & $\{a, a\}$ \\
      \hline
      $v_{i_1}(\vecc{\omega})$ & $6$ & $6$ \\
      \hline
      $v_{i_2}(\vecc{\omega})$ & $11$ & $11$ \\
      \hline
      $v_{j_1}(\vecc{\omega})$ & $5$ & $-$\\
      \hline
      $v_{j_2}(\vecc{\omega})$ & $10$ & $-$ \\
      \hline
    \end{tabular}

    \caption{Valuations $v_k(\vecc{\omega})$ for each agent $k \in N$ and each bundle $\vecc{\omega} \leq \vecc{\Omega}$ in Example~\ref{example:345noto2(IPPiM)SingNo}.}
    \label{table:345noto2(IPPiM)SingNo_valuations}
  \end{table}

  We first define a feasible assignment that is admissible in every market setting of $\mathcal{E}$ with at least third-order resale restrictions. Let $(\vecc{S}_I, \vecc{Z}_J)$ be the allocation in which agents $i_1$ and $j_1$ trade the bundle $\vecc{s}_{i_1j_1} = \vecc{z}_{i_1j_1} \equiv \{a\}$, while agents $i_2$ and $j_2$ trade the bundle $\vecc{s}_{i_2j_2} = \vecc{z}_{i_2j_2} \equiv \{a\}$. The corresponding pricing system $\vecc{p}$ is reported in Table~\ref{table:345noto2(IPPiM)SingNo_pricing}, where $M \gg 0$.

  \begin{table}[ht]
    \centering
    \begin{tabular}{|c|c|c|c|}
      \hline
      $\vecc{\omega}$ & $\{a\}$ & $\{a, a\}$ \\
      \hline
      $p_{i_1}(\vecc{\omega})$ & $5.5$ & $M$\\
      \hline
      $p_{i_2}(\vecc{\omega})$ & $10.5$ & $M$\\
      \hline
    \end{tabular}

    \caption{Third-order admissible pricing system $\vecc p$ in Example~\ref{example:345noto2(IPPiM)SingNo}.}
    \label{table:345noto2(IPPiM)SingNo_pricing}
  \end{table}

  Since the allocation is fourth-order admissible and the pricing system is third-order admissible, the assignment $((\vecc{S}_I, \vecc{Z}_J), \vecc p)$ is third- and fourth-order admissible (and fifth-order admissible if we consider the cumulative equivalent $((\vecc{s}_I, \vecc{z}_J), \vecc p)$). It is also admissible in all side-payment settings, since the corresponding $\vecc{Q}_N = \vecc 0$.
  The payoff is non-negative for each agent; in particular, $\pi_k = 0.5$ for every $k \in N$. Therefore, the assignment is feasible in all of the settings mentioned above. A representation of the payment graph is provided in Figure~\ref{figure:345noto2(IPPiM)SingNo_pricing}.

  \begin{figure}[ht]
    \centering
    \begin{tikzpicture}[>=Stealth, thick, node/.style={draw, circle, fill=white, minimum size=0.7cm}]
      \begin{scope}[local bounding box=diagram]
        \node[node] (i1) at (0,1.5) {$i_1$};
        \node[left=0.15cm of i1, font=\small] (text-i1) {
          \begin{tabular}{r|l} $v_{i_1}(\{a\}) = 6$ & \textcolor{TUMBlue}{$\pi_{i_1} = 0.5$}
        \end{tabular}};

        \node[node] (i2) at (0,0) {$i_2$};
        \node[left=0.15cm of i2, font=\small] (text-i2) {
          \begin{tabular}{r|l} $v_{i_2}(\{a\}) = 11$ & \textcolor{TUMBlue}{$\pi_{i_2} = 0.5$}
        \end{tabular}};

        \node[node] (j1) at (3.5,1.5) {$j_1$};
        \node[right=0.15cm of j1, font=\small] (text-j1) {
          \begin{tabular}{r|l} \textcolor{TUMBlue}{$\pi_{j_1} = 0.5$} & $v_{j_1}(\{a\}) = 5$
        \end{tabular}};

        \node[node] (j2) at (3.5,0) {$j_2$};
        \node[right=0.15cm of j2, font=\small] (text-j2) {
          \begin{tabular}{r|l} \textcolor{TUMBlue}{$\pi_{j_2} = 0.5$} & $v_{j_2}(\{a\}) = 10$
        \end{tabular}};

        \begin{scope}[on background layer]
          \foreach \n/\t in {i1/text-i1, i2/text-i2, j1/text-j1, j2/text-j2}
          \node[fill=gray!10, rounded corners=3mm, inner sep=1.2mm, fit=(\n) (\t)] {};
        \end{scope}

        \draw[->, TUMBlue] (i1) -- node[midway, above, font=\scriptsize, text=TUMBlue] {$p_{i_1j_1} = 5.5$} (j1);
        \draw[->, TUMBlue] (i2) -- node[midway, above, font=\scriptsize, text=TUMBlue] {$p_{i_2j_2} = 10.5$} (j2);

        \begin{scope}[on background layer]
          \node[draw=black, line width=0.5pt, rounded corners=5mm, inner sep=10pt, fit=(diagram)] {};
        \end{scope}
      \end{scope}
    \end{tikzpicture}

    \caption{Graph representation of the assignment $((\vecc{S}_I, \vecc{Z}_J), \vecc p)$ in Example~\ref{example:345noto2(IPPiM)SingNo}.}
    \label{figure:345noto2(IPPiM)SingNo_pricing}
  \end{figure}
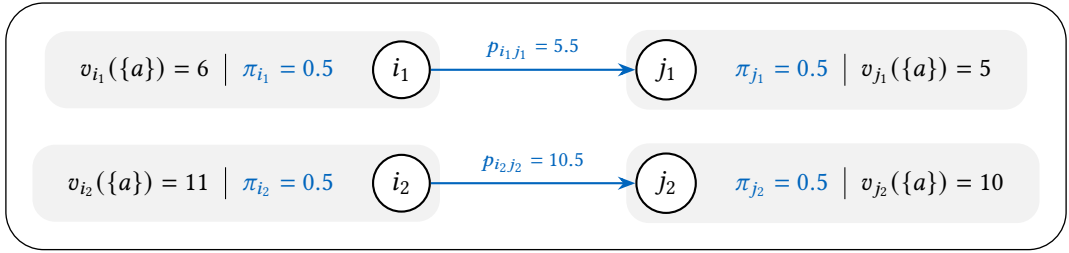

  We consider next a generic second-order admissible assignment $((\vecc{\tilde S}_I, \vecc{\tilde Z}_J), \vecc{\tilde p}, \vecc{\tilde Q}_N) \in \mathcal{F}(\mathcal{E}_2^r)$ with $r \in \{\mathrm{full}, \, \mathrm{sep}, \, \mathrm{buy}, \, \mathrm{sell}, \, \mathrm{no}\}$.
  The only second-order admissible allocation $(\vecc{\tilde S}_I, \vecc{\tilde Z}_J)$ of interest is $(\vecc{S}_I, \vecc{Z}_J)$ itself; otherwise, at least one agent would receive an empty cumulative bundle, and their valuation and cumulative price would be zero.
  For the corresponding pricing system $\vecc{\tilde p}$ to be second-order admissible, all agents $k \in N$ must pay or receive the same cumulative price $\tilde p_k = \tilde p(\{a\})$, since they all trade the same bundle.

  We now analyze the different side-payment restrictions for $((\vecc{S}_I, \vecc{Z}_J), \vecc{\tilde p}, \vecc{\tilde Q}_N)$ and establish for which of them there exists no such assignment that is price or payoff equivalent to $((\vecc{S}_I, \vecc{Z}_J), \vecc p)$.

  \paragraph{Full and separate side-payment settings} We have already proved in Lemma~\ref{lemma:12345to1(Pi)FullSep} that, for any assignment in any market setting, there exists a payoff-equivalent second-order admissible assignment with separate (and hence also full) side payments. For price equivalence to hold between $((\vecc{S}_I, \vecc{Z}_J), \vecc p)$ and an assignment $((\vecc{S}_I, \vecc{Z}_J), \vecc{\tilde p}, \vecc{\tilde Q}_N) \in \mathcal{F}(\mathcal{E}_2^r)$ with $r \in \{\mathrm{full}, \, \mathrm{sep}\}$, we must have, on the one hand,
  $$\tilde p(\{a\}) = \tilde p_{i_1} = \tilde p_{j_1} = p_{i_1} = p_{j_1} = 5.5$$
  and, on the other hand,
  $$\tilde p(\{a\}) = \tilde p_{i_2} = \tilde p_{j_2} = p_{i_2} = p_{j_2} = 10.5.$$
  However, there is no price $\tilde p(\{a\}) \geq 0$ that satisfies both conditions. Therefore, no feasible second-order admissible assignment with full or separate side payments is price-equivalent to $((\vecc{S}_I, \vecc{Z}_J), \vecc p)$.

  \paragraph{Buyers' single side-payment setting} In the buyers' single side-payment setting, both sellers need to receive the same amount of money since they are offering the same bundle. For the allocation $(\vecc{\tilde S}_I, \vecc{\tilde Z}_J)$ to be priceable, the cumulative price of each seller must then be at least
  $$\tilde p_{j_1} = \tilde p_{j_2} = \tilde p(\{a\}) \geq \max\{v_{j_1}(\vecc{\tilde z}_{j_1}), v_{j_2}(\vecc{\tilde z}_{j_2})\} = 10.$$
  But the total valuation of the buyers is only
  $$v_{i_1}(\vecc{\tilde s}_{i_1}) + v_{i_2}(\vecc{\tilde s}_{i_2}) = 17 < 2\tilde p(\{a\}) = 20,$$
  so there is no feasible assignment in $\mathcal{E}_2^{\mathrm{buy}}$ that is even item equivalent to $((\vecc{S}_I, \vecc{Z}_J), \vecc p)$.

  \paragraph{Sellers' single side-payment setting} If we allow for side payments only on the sellers' side, the minimum total price the buyers must pay is given by the sum of the sellers' valuations
  $$\tilde p_{i_1} + \tilde p_{i_2} = \tilde p_{j_1} + \tilde p_{j_2} \geq v_{j_1}(\vecc{\tilde z}_{j_1}) + v_{j_2}(\vecc{\tilde z}_{j_2}) = 15,$$
  which means that
  $$\tilde p_{i_1} = \tilde p_{i_2} = \tilde p(\{a\}) \geq 7.5.$$
  However, buyer $i_1$'s valuation for bundle $\{a\}$ is only $v_{i_1}(\vecc{\tilde s}_{i_1}) = 6 < 7.5$. Hence, there is no feasible assignment in $\mathcal{E}_2^{\mathrm{sell}}$ that is item equivalent to $((\vecc{S}_I, \vecc{Z}_J), \vecc p)$.

  \paragraph{No side-payment setting} Finally, with no side payments allowed, combining the constraints on $\tilde p(\{a\})$ of the two previous settings, we can reach the same conclusions.
\end{example}

\subsubsection{From the Second to the Third- and Fourth-Order Settings}
\begin{example}\label{example:34noto2(M)}
  Consider an economy $\mathcal{E}$ with two buyers $I = \{i_1, i_2\}$ and four sellers $J = \{j_1, \dots, j_4\}$, where $i_2$, $j_3$, and $j_4$ are agents that only request or offer empty bundles to all the other agents.\footnote{This is not a structural feature that is used to establish this example, but just a simplification to emphasize the essential structure of this market.} The initial endowments of the sellers are $\vecc\Omega_{j_1} = \vecc\Omega_{j_2} \equiv \{a, b, c\}$ and $\vecc\Omega_{j_3} = \vecc\Omega_{j_4} \equiv \emptyset$, so that the total endowment of the economy is $\vecc\Omega \equiv \{a, a, b, b, c, c\}$. The agents' valuations are reported in Table~\ref{table:2noto3(IPPi)BuyNo_valuations}.

  \begin{table}[ht]
    \centering
    \begin{tabular}{|c|c|c|c|c|c|c|c|c|c|}
      \hline
      $\vecc{\omega}$ & $\{a\}$ & $\{b\}$ & $\{c\}$ & $\{a, b\}$ & $\{a, c\}$ & $\{b, c\}$ & $\{a, b, c\}$ & Others & $\vecc{\Omega}$ \\
      \hline
      $v_{i_1}(\vecc{\omega})$ & $0$ & $0$ & $0$ & $0$ & $0$ & $0$ & $0$ & $0$ & $27$ \\
      \hline
      $v_{i_2}(\vecc{\omega})$ & $0$ & $0$ & $0$ & $0$ & $0$ & $0$ & $0$ & $0$ & $0$ \\
      \hline
      $v_{j_1}(\vecc{\omega})$ & $5$ & $5$ & $5$ & $10$ & $10$ & $5$ & $10$ & $-$ & $-$ \\
      \hline
      $v_{j_2}(\vecc{\omega})$ & $5$ & $5$ & $5$ & $6$ & $10$ & $10$ & $14$ & $-$ & $-$ \\
      \hline
      $v_{j_3}(\vecc{\omega})$ & $-$ & $-$ & $-$ & $-$ & $-$ & $-$ & $-$ & $-$ & $-$ \\
      \hline
      $v_{j_4}(\vecc{\omega})$ & $-$ & $-$ & $-$ & $-$ & $-$ & $-$ & $-$ & $-$ & $-$ \\
      \hline
    \end{tabular}

    \caption{Valuations $v_k(\vecc{\omega})$ for each agent $k \in N$ and each bundle $\vecc{\omega} \leq \vecc{\Omega}$ in Example~\ref{example:34noto2(M)}.}
    \label{table:2noto3(IPPi)BuyNo_valuations}
  \end{table}

  Consider the allocation $(\vecc{S}_I, \vecc{Z}_J)$ such that:
  \begin{enumerate}
    \item buyer $i_1$ requests bundles $\vecc{s}_{i_1j_1} \equiv \{a\}$, $\vecc{s}_{i_1j_2} \equiv \{a, b\}$, $\vecc{s}_{i_1j_3} \equiv \{b, c\}$ and $\vecc{s}_{i_1j_4} \equiv \{c\}$,
    \item seller $j_1$ offers bundles $\vecc{z}_{i_1j_1} \equiv \{a\}$ and $\vecc{z}_{i_2j_1} \equiv \{b, c\}$, and
    \item seller $j_2$ offers bundles $\vecc{z}_{i_1j_2} \equiv \{a, b\}$ and $\vecc{z}_{i_2j_2} \equiv \{c\}$.
  \end{enumerate}
  After the relabeling of the clearinghouse, bundle $\vecc{z}_{i_2j_1}$ becomes $\vecc{s}_{i_1j_3}$ and bundle $\vecc{z}_{i_2j_2}$ becomes $\vecc{s}_{i_1j_4}$, so that $i_1$ is effectively trading two distinct individual bundles with $j_1$ and two distinct individual bundles with $j_2$. The corresponding pricing system $\vecc{p}$ is reported in Table~\ref{table:2noto3(IPPi)BuyNo_pricing}, where $M \gg 0$.

  \begin{table}[ht]
    \centering
    \begin{tabular}{|c|c|c|c|c|c|c|c|c|c|}
      \hline
      $\vecc{\omega}$ & $\{a\}$ & $\{b\}$ & $\{c\}$ & $\{a, b\}$ & $\{a, c\}$ & $\{b, c\}$ & $\{a, b, c\}$ & Others \\
      \hline
      $p(\vecc{\omega})$ & $5$ & $M$ & $7$ & $8$ & $M$ & $6$ & $M$ & $M$ \\
      \hline
    \end{tabular}

    \caption{Second-order admissible pricing system $\vecc p$ in Example~\ref{example:34noto2(M)}.}
    \label{table:2noto3(IPPi)BuyNo_pricing}
  \end{table}

  The assignment $((\vecc{S}_I, \vecc{Z}_J), \vecc p)$ is second-order admissible, and it is admissible in all the side-payment settings since $\vecc{Q}_N = \vecc 0$. The payoff is $\pi_k = 1$ for $k \in \{i_1, j_1, j_2\}$ and $\pi_k = 0$ for $k \in \{i_2, j_3, j_4\}$, so that the assignment is feasible in any of the settings considered above. We provide a representation of the payment graph in Figure~\ref{figure:2noto3(IPPi)BuyNo}.

  \begin{figure}[ht]
    \centering
    \begin{tikzpicture}[>=Stealth, thick, node/.style={draw, circle, fill=white, minimum size=0.7cm}]
      \begin{scope}[local bounding box=diagram]
        \node[node] (i1) at (0,1.5) {$i_1$};
        \node[left=0.15cm of i1, font=\small] (text-i1) {
          \begin{tabular}{r|l} $v_{i_1}(\vecc{\Omega}) = 27$ & \textcolor{TUMBlue}{$\pi_{i_1} = 1$}
        \end{tabular}};

        \node[node] (i2) at (0,-1.35) {$i_2$};
        \node[left=0.15cm of i2, font=\small] (text-i2) {
          \begin{tabular}{l} \textcolor{TUMBlue}{$\pi_{i_2} = 0$}
        \end{tabular}};

        \node[node] (j1) at (5,2.75) {$j_1$};
        \node[right=0.15cm of j1, font=\small] (text-j1) {
          \begin{tabular}{r|l} \textcolor{TUMBlue}{$\pi_{j_1} = 1$} & $v_{j_1}(\{a\}) = 5$ \\ & $v_{j_1}(\{b, c\}) = 5$ \\ & $v_{j_1}(\{a, b, c\}) = 10$
        \end{tabular}};

        \node[node] (j2) at (5,0.5) {$j_2$};
        \node[right=0.15cm of j2, font=\small] (text-j2) {
          \begin{tabular}{r|l} \textcolor{TUMBlue}{$\pi_{j_2} = 1$} & $v_{j_2}(\{c\}) = 5$ \\ & $v_{j_2}(\{a, b\}) = 6$ \\ & $v_{j_2}(\{a, b, c\}) = 14$
        \end{tabular}};

        \node[node] (j3) at (5,-1.35) {$j_3$};
        \node[right=0.15cm of j3, font=\small] (text-j3) {
          \begin{tabular}{l} \textcolor{TUMBlue}{$\pi_{j_3} = 0$}
        \end{tabular}};

        \node[node] (j4) at (5,-2.85) {$j_4$};
        \node[right=0.15cm of j4, font=\small] (text-j4) {
          \begin{tabular}{l} \textcolor{TUMBlue}{$\pi_{j_4} = 0$}
        \end{tabular}};

        \begin{scope}[on background layer]
          \foreach \n/\t in {i1/text-i1, i2/text-i2, j1/text-j1, j2/text-j2, j3/text-j3, j4/text-j4}
          \node[fill=gray!10, rounded corners=3mm, inner sep=1.2mm, fit=(\n) (\t)] {};
        \end{scope}

        \draw[->, TUMBlue] (i1) -- node[sloped, pos=0.5, above, font=\scriptsize, text=TUMBlue] {$p_{i_1j_1} = 5$} (j1);
        \draw[->, TUMBlue] (i1) -- node[sloped, pos=0.4, above, font=\scriptsize, text=TUMBlue] {$p_{i_1j_2} = 8$} (j2);
        \draw[->, draw=gray!70] (i1) -- node[sloped, pos=0.8, below, font=\scriptsize, text=gray] {$p_{i_1j_3} = 6$} (j3);
        \draw[->, draw=gray!70] (i1) -- node[sloped, pos=0.75, below, font=\scriptsize, text=gray] {$p_{i_1j_4} = 7$} (j4);

        \draw[->, draw=gray!70] (i2) -- node[sloped, pos=0.17, above, font=\scriptsize, text=gray] {$p_{i_2j_1} = 6$} (j1);
        \draw[->, draw=gray!70] (i2) -- node[sloped, pos=0.25, below, font=\scriptsize, text=gray] {$p_{i_2j_2} = 7$} (j2);

        \begin{scope}[on background layer]
          \node[draw=black, line width=0.5pt, rounded corners=5mm, inner sep=10pt, fit=(diagram)] {};
        \end{scope}
      \end{scope}
    \end{tikzpicture}

    \caption{Graph representation of the assignment $((\vecc{S}_I, \vecc{Z}_J), \vecc p)$ in Example~\ref{example:34noto2(M)}. The gray arrows indicate the payment relations from the agents' perspective for bundles relabeled by the clearinghouse, but that do not correspond to a real exchange of goods.}
    \label{figure:2noto3(IPPi)BuyNo}
  \end{figure}
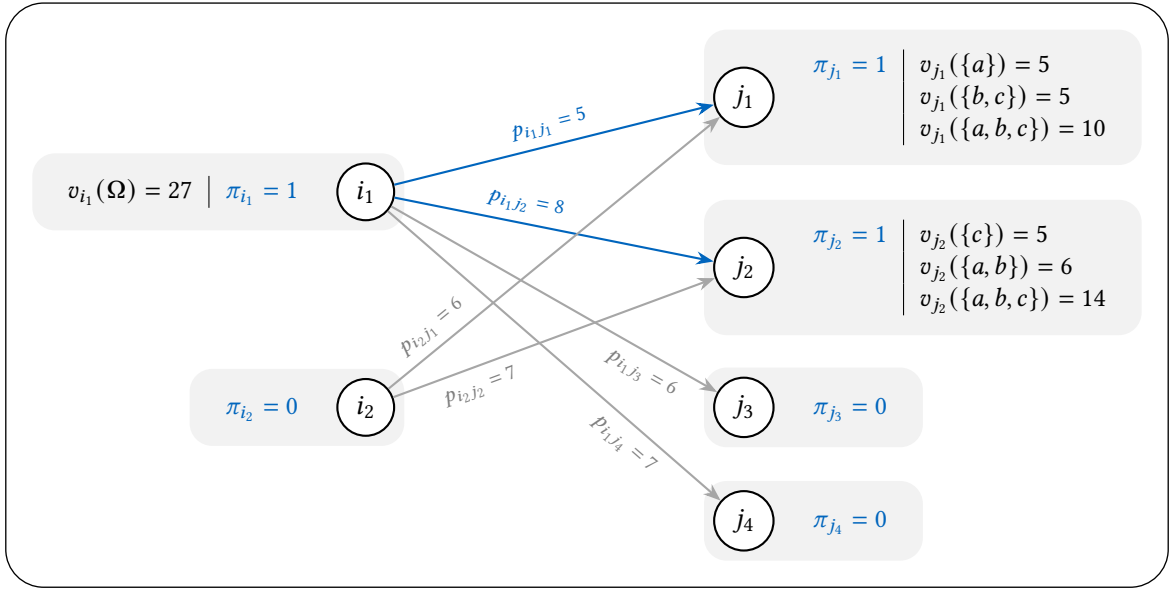

  In the third- and fourth-order settings, every seller must offer each individual bundle to the buyer who requests it. This means that, to define a third- or fourth-order admissible allocation $(\vecc{\tilde S}_I, \vecc{\tilde Z}_J)$ that is item equivalent to $(\vecc{S}_I, \vecc{Z}_J)$, the first two sellers must offer all their items to buyer $i_1$ as a single bundle, i.e., $\vecc{\tilde z}_{i_1j_1} \equiv \{a, b, c\}$ and $\vecc{\tilde z}_{i_1j_2} \equiv \{a, b, c\}$, whereas $\vecc{\tilde z}_{i_2j_1} = \vecc{\tilde z}_{i_2j_2} \equiv \emptyset$. Since no repackaging is allowed, buyer $i_1$ will then have to request two copies of bundle $\{a, b, c\}$ and two copies of the empty bundle.

  As the price of the empty bundle is always normalized to zero, in order to maintain the same cumulative price $p_k$ for each agent $k \in N$, the first two sellers need to receive respectively $\tilde p_{j_1} = \tilde p_{i_1j_1}(\{a, b, c\}) = 11$ and $\tilde p_{j_2} = \tilde p_{i_1j_2}(\{a, b, c\}) = 15$ from buyer $i_1$.
  However, we can already observe that there exists at least one case for which the prices assigned by $\vecc{p}$ and $\vecc{\tilde p}$ don't coincide, for example, $\tilde p_{i_1j_1}(\{a, b, c\}) \neq p_{i_1j_1}(\{a, b, c\}) = p(\{a, b, c\}) = M$. Therefore, we can rule out the existence of any feasible third- or fourth-order admissible assignment that is market equivalent to $((\vecc{S}_I, \vecc{Z}_J), \vecc p)$.

  Let's now analyze the different side-payment settings when $((\vecc{\tilde S}_I, \vecc{\tilde Z}_J), \vecc{\tilde p}, \vecc{\tilde Q}_N)$ is a third-order admissible assignment.

  \paragraph{Full, separate, and sellers' single side-payment settings} We have already proved in Lemma~\ref{lemma:12345to1(Pi)FullSep} and \ref{lemma:245to3(Pi)Sel} that payoff equivalence holds for these settings. However, the cumulative prices $\tilde p_{j_1}$ and $\tilde p_{j_2}$ defined above are not compatible with third-order pricing systems, since buyer $i_1$ would have to pay two different prices for the same bundle, i.e., $\tilde p_{i_1j_1}(\{a, b, c\}) \neq \tilde p_{i_1j_2}(\{a, b, c\})$. Hence, we can't find any third-order admissible assignment that is price equivalent to $((\vecc{S}_I, \vecc{Z}_J), \vecc p)$.

  \paragraph{Buyers' single and no side-payment settings} In the buyers' single and no side-payment settings, to fulfill each seller's valuation, the pricing system has to satisfy the condition
  $$\tilde p_{i_1}(\{a, b, c\}) \geq \max\{v_{j_1}(\vecc{\tilde z}_{j_1}), v_{j_2}(\vecc{\tilde z}_{j_2})\} = 14.$$
  However, buyer $i_1$ is not able to pay this amount to two sellers since
  $$ \tilde p_{i_1} = 2 \tilde p_{i_1}(\{a, b, c\}) \geq 28 > v_{i_1}(\vecc{\tilde s}_{i_1}) = 27.$$
  Therefore, we conclude that there is no feasible assignment in $\mathcal{E}_3^r$ with $r \in \{\mathrm{buy}, \, \mathrm{no}\}$ that is even item equivalent to $((\vecc{S}_I, \vecc{Z}_J), \vecc p)$.
\end{example}

\subsubsection{To the Third-Order Settings}
\begin{example}\label{example:45noto3(IPPiM)BuyNo}
  Consider an economy $\mathcal{E}$ with two buyers $I = \{i_1, i_2\}$ and three sellers $J = \{j_1, j_2, j_3\}$. The initial endowments of the sellers are respectively $\vecc\Omega_{j_1} \equiv \{a\}$, $\vecc\Omega_{j_2} \equiv \{a, b\}$ and $\vecc\Omega_{j_3} \equiv \{b\}$, so that the total endowment of the economy is $\vecc\Omega \equiv \{a, a, b, b\}$. The agents' valuations are reported in Table~\ref{table:45noto3(IPPiM)BuyNo_valuations}.

  \begin{table}[ht]
    \centering
    \begin{tabular}{|c|c|c|c|c|c|c|c|c|}
      \hline
      $\vecc{\omega}$ & $\{a\}$ & $\{b\}$ & $\{a, a\}$ & $\{a, b\}$ & $\{b, b\}$ & $\{a, a, b\}$ & $\{a, b, b\}$ & $\{a, a, b, b\}$ \\
      \hline
      $v_{i_1}(\vecc{\omega})$ & $0$ & $0$ & $33$ & $0$ & $0$ & $33$ & $0$ & $33$ \\
      \hline
      $v_{i_2}(\vecc{\omega})$ & $0$ & $0$ & $0$ & $0$ & $18$ & $0$ & $18$ & $18$ \\
      \hline
      $v_{j_1}(\vecc{\omega})$ & $10$ & $-$ & $-$ & $-$ & $-$ & $-$ & $-$ & $-$ \\
      \hline
      $v_{j_2}(\vecc{\omega})$ & $20$ & $10$ & $-$ & $31$ & $-$ & $-$ & $-$ & $-$ \\
      \hline
      $v_{j_3}(\vecc{\omega})$ & $-$ & $5$ & $-$ & $-$ & $-$ & $-$ & $-$ & $-$ \\
      \hline
    \end{tabular}

    \caption{Valuations $v_k(\vecc{\omega})$ for each agent $k \in N$ and each bundle $\vecc{\omega} \leq \vecc{\Omega}$ in Example~\ref{example:45noto3(IPPiM)BuyNo}.}
    \label{table:45noto3(IPPiM)BuyNo_valuations}
  \end{table}

  Consider the allocation $(\vecc{S}_I, \vecc{Z}_J)$ where buyer $i_1$ gets two copies of bundle $\{a\} \equiv \vecc{s}_{i_1j_1} = \vecc{z}_{i_1j_1} = \vecc{s}_{i_1j_2} = \vecc{z}_{i_1j_2}$ from sellers $j_1$ and $j_2$, and buyer $i_2$ gets two copies of bundle $\{b\} \equiv \vecc{s}_{i_2j_2} = \vecc{z}_{i_2j_2} = \vecc{s}_{i_2j_3} = \vecc{z}_{i_2j_3}$ from sellers $j_2$ and $j_3$. The corresponding pricing system $\vecc{p}$ is reported in Table~\ref{table:45noto3(IPPiM)BuyNo_pricing}, where $M \gg 0$.

  \begin{table}[ht]
    \centering
    \begin{tabular}{|c|c|c|c|}
      \hline
      $\vecc{\omega}$ & $\{a\}$ & $\{b\}$ & Others \\
      \hline
      $p_{i_1j_1}(\vecc{\omega})$ & $11$ & $M$ & $M$ \\
      \hline
      $p_{i_1j_2}(\vecc{\omega})$ & $21$ & $M$ & $M$ \\
      \hline
      $p_{i_2j_1}(\vecc{\omega})$ & $M$ & $11$ & $M$ \\
      \hline
      $p_{i_2j_3}(\vecc{\omega})$ & $M$ & $6$ & $M$ \\
      \hline
      Others & $M$ & $M$ & $M$ \\
      \hline
    \end{tabular}

    \caption{Fourth-order admissible pricing system $\vecc p$ in Example~\ref{example:45noto3(IPPiM)BuyNo}.}
    \label{table:45noto3(IPPiM)BuyNo_pricing}
  \end{table}

Since the allocation and the pricing system are fourth-order admissible, the assignment $((\vecc{S}_I, \vecc{Z}_J), \vecc p)$ is fourth-order admissible, and the respective cumulative assignment $((\vecc{s}_I, \vecc{z}_J), \vecc p)$) is fifth-order admissible. It is also admissible in all the side-payment settings since the corresponding side payments are $\vecc{Q}_N = \vecc 0$.
The payoff is non-negative for each agent, in particular, $\pi_k = 1$ for every $k \in N$, so the assignment is feasible in all the settings above. We provide a representation of the payment graph in Figure~\ref{figure:45noto3(IPPiM)BuyNo}.

\begin{figure}[ht]
  \centering
  \begin{tikzpicture}[>=Stealth, thick, node/.style={draw, circle, fill=white, minimum size=0.7cm}]
    \begin{scope}[local bounding box=diagram]
      \node[node] (i1) at (0,1) {$i_1$};
      \node[left=0.15cm of i1, font=\small] (text-i1) {
        \begin{tabular}{r|l} $v_{i_1}(\{a, a\}) = 33$ & \textcolor{TUMBlue}{$\pi_{i_1} = 1$}
      \end{tabular}};

      \node[node] (i2) at (0,-1) {$i_2$};
      \node[left=0.15cm of i2, font=\small] (text-i2) {
        \begin{tabular}{r|l} $v_{i_2}(\{b, b\}) = 18$ & \textcolor{TUMBlue}{$\pi_{i_2} = 1$}
      \end{tabular}};

      \node[node] (j1) at (3.5,2) {$j_1$};
      \node[right=0.15cm of j1, font=\small] (text-j1) {
        \begin{tabular}{r|l} \textcolor{TUMBlue}{$\pi_{j_1} = 1$} & $v_{j_1}(\{a\}) = 10$
      \end{tabular}};

      \node[node] (j2) at (3.5,0) {$j_2$};
      \node[right=0.15cm of j2, font=\small] (text-j2) {
        \begin{tabular}{r|l} \textcolor{TUMBlue}{$\pi_{j_2} = 1$} & $v_{j_2}(\{a\}) = 20$ \\ & $v_{j_2}(\{b\}) = 10$ \\ & $v_{j_2}(\{a, b\}) = 31$
      \end{tabular}};

      \node[node] (j3) at (3.5,-2) {$j_3$};
      \node[right=0.15cm of j3, font=\small] (text-j3) {
        \begin{tabular}{r|l} \textcolor{TUMBlue}{$\pi_{j_3} = 1$} & $v_{j_3}(\{b\}) = 5$
      \end{tabular}};

      \begin{scope}[on background layer]
        \foreach \n/\t in {i1/text-i1, i2/text-i2, j1/text-j1, j2/text-j2, j3/text-j3}
        \node[fill=gray!10, rounded corners=3mm, inner sep=1.2mm, fit=(\n) (\t)] {};
      \end{scope}

      \draw[->, TUMBlue] (i1) -- node[pos=0.5, sloped, above, font=\scriptsize, text=TUMBlue] {$p_{i_1j_1} = 11$} (j1);
      \draw[->, TUMBlue] (i1) -- node[pos=0.5, sloped, above, font=\scriptsize, text=TUMBlue] {$p_{i_1j_2} = 21$} (j2);
      \draw[->, TUMBlue] (i2) -- node[pos=0.5, sloped, above, font=\scriptsize, text=TUMBlue] {$p_{i_2j_2} = 11$} (j2);
      \draw[->, TUMBlue] (i2) -- node[pos=0.5, sloped, above, font=\scriptsize, text=TUMBlue] {$p_{i_2j_3} = 6$} (j3);

      \begin{scope}[on background layer]
        \node[draw=black, line width=0.5pt, rounded corners=5mm, inner sep=10pt, fit=(diagram)] {};
      \end{scope}
    \end{scope}
  \end{tikzpicture}

  \caption{Graph representation of the assignment $((\vecc{S}_I, \vecc{Z}_J), \vecc p)$ in Example~\ref{example:45noto3(IPPiM)BuyNo}.}
  \label{figure:45noto3(IPPiM)BuyNo}
\end{figure}
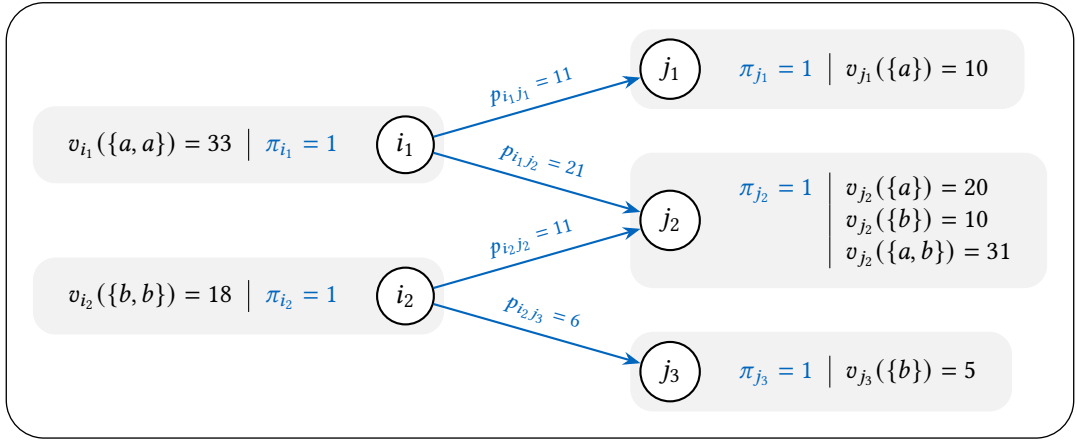

Let $((\vecc{\tilde S}_I, \vecc{\tilde Z}_J), \vecc{\tilde p}, \vecc{\tilde Q}_N) \in \mathcal{F}(\mathcal{E}_3^r)$ with $r \in \{\mathrm{full}, \, \mathrm{sep}, \, \mathrm{buy}, \, \mathrm{sell}, \, \mathrm{no}\}$.
For the assignment to be feasible, the cumulative bundles $(\vecc{\tilde s}_I, \vecc{\tilde z}_J)$ must either coincide with the respective $\vecc{s}_i$ and $\vecc{z}_j$, or be empty. To prove this, suppose by contradiction that $\vecc{\tilde s}_{i_1} \equiv \{a\}$ (a similar argument can be made if we assume that $\vecc{\tilde s}_{i_2} \equiv \{b\}$, $\vecc{\tilde z}_{j_2} \equiv \{a\}$, or $\vecc{\tilde z}_{j_2} \equiv \{b\}$). Since $v_{i_1}(\{a\}) = 0$, we can distinguish two cases:
\begin{enumerate}
  \item If buyer $i_2$ requests either the empty bundle or bundle $\{b\}$, the total valuation of the buyers is $v_{i_1}(\{a\}) + v_{i_2}(\emptyset) = v_{i_1}(\{a\}) + v_{i_2}(\{b\}) = 0$, and the total valuation of the sellers is $\sum_{j \in J} v_j(\vecc {\tilde z}_j) > 0$.
  \item If buyer $i_2$ requests bundle $\{b,b\}$, the total valuation of the buyers is $v_{i_1}(\{a\}) + v_{i_2}(\{b, b\}) = 18$, and the minimum possible total valuation of the sellers is $\min\{v_{j_1}(\{a\}) + v_{j_2}(\{b\}), v_{j_2}(\{a, b\})\} + v_{j_3}(\{b\}) = 25$.
\end{enumerate}
In both cases, the social welfare of the admissible allocations is negative, so at least one agent would get a negative payoff. Hence, in any feasible assignment in $\mathcal{E}_3^r$, buyer $i_1$'s cumulative bundle must either be the empty bundle or bundle $\{a, a\}$ (and, similarly, buyer $i_2$'s cumulative bundle must either be the empty bundle or bundle $\{b,b\}$, and seller $j_2$'s cumulative bundle must either be the empty bundle or bundle $\{a, b\}$).

The only allocation that we can consider for the equivalence properties is then the one where no agent has a cumulative empty bundle, i.e., such that $(\vecc{\tilde s}_I, \vecc{\tilde z}_J) = (\vecc{s}_I, \vecc{z}_J)$. Furthermore, the corresponding pricing system $\vecc{\tilde p}$ has to assign the same price to a bundle that is sold to a buyer multiple times, so in our case, $\tilde p_{i_1j_1}(\{a\}) = \tilde p_{i_1j_2}(\{a\}) = \tilde p_{i_1}(\{a\})$ and $\tilde p_{i_2j_2}(\{b\}) = \tilde p_{i_2j_3}(\{b\}) = \tilde p_{i_2}(\{b\})$.

Let's now analyze which equivalence properties do not hold in the different side-payment settings for the assignment $((\vecc{\tilde S}_I, \vecc{\tilde Z}_J), \vecc{\tilde p}, \vecc{\tilde Q}_N)$.

\paragraph{Full, separate, and sellers' single side-payment settings} We have already proved in Lemmas~\ref{lemma:12345to1(Pi)FullSep} and \ref{lemma:245to3(Pi)Sel} that payoff equivalence holds for these settings. To show that price equivalence does not hold, suppose for example that $\tilde{p}_{j_1} = p_{j_1} = 11$ so that the price $\tilde{p}_{i_1}(\{a\}) = 11$ is fixed. This means that buyer $i_1$, whose cumulative bundle is $\vecc{\tilde{s}}_{i_1} \equiv \{a, a\}$, would have to pay a cumulative price of
$$\tilde p_{i_1} = 2 \tilde p_{i_1}(\{a\}) = 22 \neq p_{i_1} = 32.$$
It follows that we can't find any feasible assignment in $\mathcal{E}_3^r$ with $r \in \{\mathrm{full}, \, \mathrm{sep}, \, \mathrm{sell}\}$ that is price equivalent to $((\vecc{S}_I, \vecc{Z}_J), \vecc p)$.

\paragraph{Buyers' single and no side-payment settings} If side payments among the sellers are not allowed, seller $j_2$ must receive a cumulative price of at least
$$\tilde{p}_{j_2} = \tilde p_{i_1}(\{a\}) + \tilde p_{i_2}(\{b\}) \geq v_{j_2}(\vecc{\tilde z}_{j_2}) = 31.$$
However, the buyers can pay a total price of at most
$$\tilde p_{i_1} + \tilde p_{i_2} = 2 \tilde p_{i_1}(\{a\}) + 2 \tilde p_{i_2}(\{b\}) \leq v_{i_1}(\vecc{\tilde s}_{i_1}) + v_{i_2}(\vecc{\tilde s}_{i_2}) = 51,$$
which means that
$$\tilde p_{i_1}(\{a\}) + \tilde p_{i_2}(\{b\}) \leq 25.5.$$
Since the two conditions on $\tilde p_{i_1}(\{a\}) + \tilde p_{i_2}(\{b\})$ contradict each other, there is no third-order admissible pricing system in these two side-payment settings that make the allocation $(\vecc{\tilde S}_I, \vecc{\tilde Z}_J)$ priceable. Hence, there is also no feasible assignment in $\mathcal{E}_3^r$ with $r \in \{\mathrm{buy}, \, \mathrm{no}\}$ that is even item equivalent to $((\vecc{S}_I, \vecc{Z}_J), \vecc{p})$.
\end{example}

\subsubsection{To the Fourth-Order Settings}
\begin{example}\label{example:5noto4(IPi)No}
Consider an economy $\mathcal{E}$ with two buyers $I = \{i_1, i_2\}$ and two sellers $J = \{j_1, j_2\}$. The initial endowments of the sellers are respectively $\vecc\Omega_{j_1} \equiv \{a\}$ and $\vecc\Omega_{j_2} \equiv \{b\}$, so that the total endowment of the economy is $\vecc\Omega \equiv \{a, b\}$. The agents' valuations are reported in Table~\ref{table:5noto4(IPi)No_valuations}.

\begin{table}[ht]
  \centering
  \begin{tabular}{|c|c|c|c|}
    \hline
    $\vecc{\omega}$ & $\{a\}$ & $\{b\}$ & $\{a, b\}$ \\
    \hline
    $v_{i_1}(\vecc{\omega})$ & $16$ & $0$ & $16$ \\
    \hline
    $v_{i_2}(\vecc{\omega})$ & $0$ & $6$ & $6$ \\
    \hline
    $v_{j_1}(\vecc{\omega})$ & $9$ & $-$ & $-$ \\
    \hline
    $v_{j_2}(\vecc{\omega})$ & $-$ & $9$ & $-$ \\
    \hline
  \end{tabular}

  \caption{Valuations $v_k(\vecc{\omega})$ for each agent $k \in N$ and each bundle $\vecc{\omega} \leq \vecc{\Omega}$ in Example~\ref{example:5noto4(IPi)No}.}
  \label{table:5noto4(IPi)No_valuations}
\end{table}

Consider the fifth-order admissible allocation $(\vecc{s}_I, \vecc{z}_J)$ where agents $i_1$ and $j_1$ are assigned bundle $\vecc{s}_{i_1} = \vecc{z}_{j_1} \equiv \{a\}$ and agents $i_2$ and $j_2$ are assigned bundle $\vecc{s}_{i_2} = \vecc{z}_{j_2} \equiv \{b\}$. The corresponding pricing system $\vecc{p}$ is reported in Table~\ref{table:5noto4(IPi)No_pricing}, where $M \gg 0$.

\begin{table}[ht]
  \centering
  \begin{tabular}{|c|c|c|c|}
    \hline
    $\vecc{\omega}$ & $\{a\}$ & $\{b\}$ & $\{a, b\}$\\
    \hline
    $p_{i_1}(\vecc{\omega})$ & $15$ & $M$ & $M$ \\
    \hline
    $p_{i_2}(\vecc{\omega})$ & $M$ & $5$ & $M$ \\
    \hline
    $p_{j_1}(\vecc{\omega})$ & $10$ & $M$ & $M$ \\
    \hline
    $p_{j_2}(\vecc{\omega})$ & $M$ & $10$ & $M$ \\
    \hline
  \end{tabular}

  \caption{Fifth-order admissible pricing system $\vecc p$ in Example~\ref{example:5noto4(IPi)No}.}
  \label{table:5noto4(IPi)No_pricing}
\end{table}

Since the allocation and the pricing system are fifth-order admissible, the assignment $((\vecc{s}_I, \vecc{z}_J), \vecc p)$ is fifth-order admissible. It is also admissible in all side-payment settings since the corresponding side payments are $\vecc{Q}_N = \vecc 0$.
The payoff is non-negative for each agent, in particular, $\pi_k = 1$ for every $k \in N$, so the assignment is feasible in all the settings above. We provide a representation of the payment graph in Figure~\ref{figure:5noto4(IPi)No}.

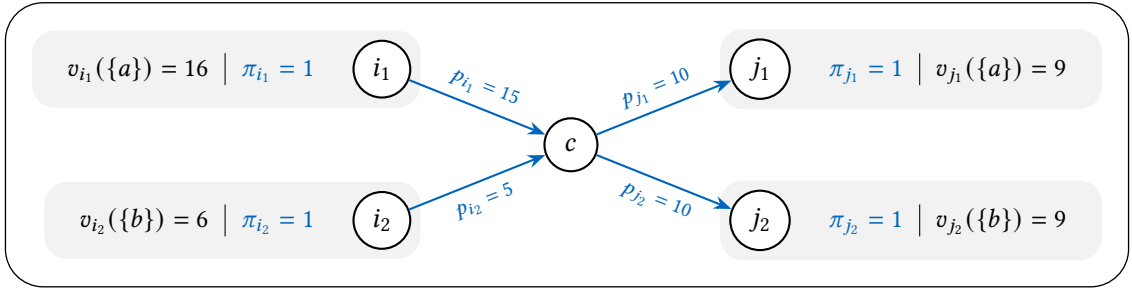
\begin{figure}[ht]
  \centering
  \begin{tikzpicture}[>=Stealth, thick, node/.style={draw, circle, fill=white, minimum size=0.7cm}]
    \begin{scope}[local bounding box=diagram]
      \node[node] (i1) at (0,1) {$i_1$};
      \node[left=0.15cm of i1, font=\small] (text-i1) {
        \begin{tabular}{r|l} $v_{i_1}(\{a\}) = 16$ & \textcolor{TUMBlue}{$\pi_{i_1} = 1$}
      \end{tabular}};

      \node[node] (i2) at (0,-1) {$i_2$};
      \node[left=0.15cm of i2, font=\small] (text-i2) {
        \begin{tabular}{r|l} $v_{i_2}(\{b\}) = 6$ & \textcolor{TUMBlue}{$\pi_{i_2} = 1$}
      \end{tabular}};

      \node[node] (c) at (2.5,0) {$c$};

      \node[node] (j1) at (5,1) {$j_1$};
      \node[right=0.15cm of j1, font=\small] (text-j1) {
        \begin{tabular}{r|l} \textcolor{TUMBlue}{$\pi_{j_1} = 1$} & $v_{j_1}(\{a\}) = 9$
      \end{tabular}};

      \node[node] (j2) at (5,-1) {$j_2$};
      \node[right=0.15cm of j2, font=\small] (text-j2) {
        \begin{tabular}{r|l} \textcolor{TUMBlue}{$\pi_{j_2} = 1$} & $v_{j_2}(\{b\}) = 9$
      \end{tabular}};

      \begin{scope}[on background layer]
        \foreach \n/\t in {i1/text-i1, i2/text-i2, j1/text-j1, j2/text-j2}
        \node[fill=gray!10, rounded corners=3mm, inner sep=1.2mm, fit=(\n) (\t)] {};
      \end{scope}

      \draw[->, TUMBlue] (i1) -- node[pos=0.5, sloped, above, font=\scriptsize, text=TUMBlue] {$p_{i_1} = 15$} (c);
      \draw[->, TUMBlue] (i2) -- node[pos=0.5, sloped, below, font=\scriptsize, text=TUMBlue] {$p_{i_2} = 5$} (c);
      \draw[->, TUMBlue] (c)  -- node[pos=0.5, sloped, above, font=\scriptsize, text=TUMBlue] {$p_{j_1} = 10$} (j1);
      \draw[->, TUMBlue] (c)  -- node[pos=0.5, sloped, below, font=\scriptsize, text=TUMBlue] {$p_{j_2} = 10$} (j2);

      \begin{scope}[on background layer]
        \node[draw=black, line width=0.5pt, rounded corners=5mm, inner sep=10pt, fit=(diagram)] {};
      \end{scope}
    \end{scope}
  \end{tikzpicture}

  \caption{Graph representation of the assignment $((\vecc{s}_I, \vecc{z}_J), \vecc p)$ in Example~\ref{example:5noto4(IPi)No}.}
  \label{figure:5noto4(IPi)No}
\end{figure}

Let $((\vecc{\tilde S}_I, \vecc{\tilde Z}_J), \vecc{\tilde p}, \vecc{\tilde Q}_N) \in \mathcal{F}(\mathcal{E}_4^r)$ with $r \in \{\mathrm{full}, \, \mathrm{sep}, \, \mathrm{buy}, \, \mathrm{sell}, \, \mathrm{no}\}$.
The only fourth-order admissible allocation $(\vecc{\tilde S}_I, \vecc{\tilde Z}_J)$ where no agent has a cumulative empty bundle is the one where $i_1$ and $j_1$ trade bundle $\vecc{\tilde s}_{i_1j_1} = \vecc{\tilde z}_{i_1j_1} \equiv \{a\}$, and $i_2$ and $j_2$ trade bundle $\vecc{\tilde s}_{i_2j_2} = \vecc{\tilde z}_{i_2j_2} \equiv \{b\}$.
Any fourth-order admissible pricing system $\vecc{\tilde p}$ has then to set $\tilde p_{i_1j_2}(\vecc{\tilde s}_{i_1j_2}) = \tilde p_{i_2j_1}(\vecc{\tilde s}_{i_2j_1}) = 0$, since the two buyer-seller pairs $(i_1, j_2)$ and $(i_2, j_1)$ only trade individual empty bundles.

Let's now analyze which equivalence properties do not hold in the different side-payment settings for the assignment $((\vecc{\tilde S}_I, \vecc{\tilde Z}_J), \vecc{\tilde p}, \vecc{\tilde Q}_N)$.

\paragraph{Full, separate, and single side-payment settings} We have already proved in Lemmas~\ref{lemma:12345to1(Pi)FullSep} and \ref{lemma:5to4(Pi)Sing} that payoff equivalence holds for these settings. For price equivalence to hold, one of the conditions that must be satisfied is $\tilde p_{i_2} = p_{i_2} = 5$. However, since seller $j_2$ can only receive the price payment $\tilde p_{i_2j_2}(\{b\}) = \tilde p_{i_2}$ from buyer $i_2$, we would have that $\tilde p_{j_2} = 5 \neq p_{j_2} = 10$.
It follows that we can't find any feasible assignment in $\mathcal{E}_4^r$ with $r \in \{\mathrm{full}, \, \mathrm{sep}, \, \mathrm{buy}, \, \mathrm{sell}\}$ that is price equivalent to $((\vecc{s}_I, \vecc{z}_J), \vecc p)$.

\paragraph{No side-payment setting} Without side payments, the price that seller $j_2$ must receive from buyer $i_2$ is at least $\tilde p_{j_2} = \tilde{p}_{i_2j_2}(\{b\}) \geq v_{j_2}(\{b\}) = 9$. However, the valuation of $i_2$ for the same bundle is only $v_{i_2}(\{b\}) = 6 < 9$. This implies that the allocation $(\vecc{\tilde S}_I, \vecc{\tilde Z}_J)$ is not priceable in this setting and, therefore, there is no feasible assignment in $\mathcal{E}_4^{\mathrm{no}}$ that is even item equivalent to $((\vecc{s}_I, \vecc{z}_J), \vecc{p})$.
\end{example}

\subsubsection{From the full to the separate side-payment settings}
\begin{example}\label{example:FullnottoSep(P)}
Consider an economy $\mathcal{E}$ with two buyers $I = \{i_1, i_2\}$ and two sellers $J = \{j_1, j_2\}$. The initial endowments of the sellers are respectively $\vecc\Omega_{j_1} \equiv \{a\}$ and $\vecc\Omega_{j_2} \equiv \{b\}$, so that the total endowment of the economy is $\vecc\Omega \equiv \{a, b\}$. The agents' valuations are reported in Table~\ref{table:FullnottoSep(P)_valuations}.

\begin{table}[ht]
  \centering
  \begin{tabular}{|c|c|c|c|}
    \hline
    $\vecc{\omega}$ & $\{a\}$ & $\{b\}$ & $\{a, b\}$ \\
    \hline
    $v_{i_1}(\vecc{\omega})$ & $16$ & $0$ & $16$ \\
    \hline
    $v_{i_2}(\vecc{\omega})$ & $0$ & $6$ & $6$ \\
    \hline
    $v_{j_1}(\vecc{\omega})$ & $9$ & $-$ & $-$ \\
    \hline
    $v_{j_2}(\vecc{\omega})$ & $-$ & $9$ & $-$ \\
    \hline
  \end{tabular}

  \caption{Valuations $v_k(\vecc{\omega})$ for each agent $k \in N$ and each bundle $\vecc{\omega} \leq \vecc{\Omega}$ in Example~\ref{example:FullnottoSep(P)}.}
  \label{table:FullnottoSep(P)_valuations}
\end{table}

Consider the allocation $(\vecc{S}_I, \vecc{Z}_J)$ where agents $i_1$ and $j_1$ trade bundle $\vecc{s}_{i_1j_1} = \vecc{z}_{i_1j_1} \equiv \{a\}$, and agents $i_2$ and $j_2$ trade bundle $\vecc{s}_{i_2j_2} = \vecc{z}_{i_2j_2} \equiv \{b\}$. The corresponding pricing system $\vecc{p}$ is such that $p_a = 10$ and $p_b = 5$. Moreover, buyer $i_1$ side-pays seller $j_2$ an amount $q_{i_1j_2} = -q_{j_2i_1} = 5$.

Since the allocation is fourth-order admissible and the pricing system is first-order admissible, the assignment $((\vecc{S}_I, \vecc{Z}_J), \vecc p, \vecc Q_N)$ is admissible in all the order settings (including the fifth-order one if we consider the cumulative equivalent $((\vecc{s}_I, \vecc{z}_J), \vecc p, \vecc{Q}_N)$). However, it is only admissible in the corresponding full side-payment settings since $q_{ij} \neq 0$ for some $i \in I$ and $j \in J$.
The payoff is non-negative for each agent, in particular, $\pi_k = 1$ for every $k \in N$, so the assignment is also feasible. We provide a representation of the payment graph in Figure~\ref{figure:FullnottoSep(P)}.

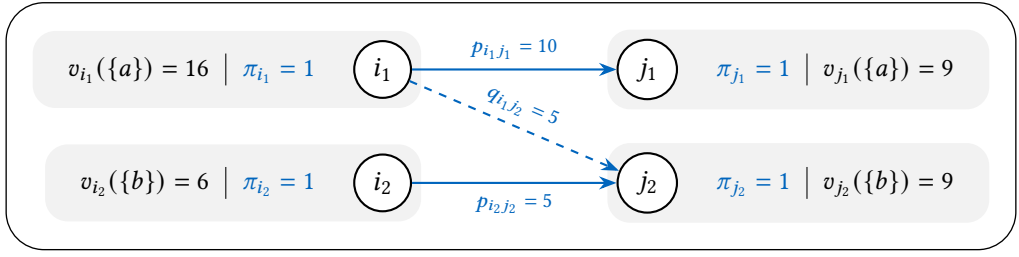
\begin{figure}[ht]
  \centering
  \begin{tikzpicture}[>=Stealth, thick, node/.style={draw, circle, fill=white, minimum size=0.7cm}]
    \begin{scope}[local bounding box=diagram]
      \node[node] (i1) at (0,1) {$i_1$};
      \node[left=0.15cm of i1, font=\small] (text-i1) {
        \begin{tabular}{r|l} $v_{i_1}(\{a\}) = 16$ & \textcolor{TUMBlue}{$\pi_{i_1} = 1$}
      \end{tabular}};

      \node[node] (i2) at (0,-0.5) {$i_2$};
      \node[left=0.15cm of i2, font=\small] (text-i2) {
        \begin{tabular}{r|l} $v_{i_2}(\{b\}) = 6$ & \textcolor{TUMBlue}{$\pi_{i_2} = 1$}
      \end{tabular}};

      \node[node] (j1) at (3.5,1) {$j_1$};
      \node[right=0.15cm of j1, font=\small] (text-j1) {
        \begin{tabular}{r|l} \textcolor{TUMBlue}{$\pi_{j_1} = 1$} & $v_{j_1}(\{a\}) = 9$
      \end{tabular}};

      \node[node] (j2) at (3.5,-0.5) {$j_2$};
      \node[right=0.15cm of j2, font=\small] (text-j2) {
        \begin{tabular}{r|l} \textcolor{TUMBlue}{$\pi_{j_2} = 1$} & $v_{j_2}(\{b\}) = 9$
      \end{tabular}};

      \begin{scope}[on background layer]
        \foreach \n/\t in {i1/text-i1, i2/text-i2, j1/text-j1, j2/text-j2}
        \node[fill=gray!10, rounded corners=3mm, inner sep=1.2mm, fit=(\n) (\t)] {};
      \end{scope}

      \draw[->, TUMBlue] (i1) -- node[midway, above, font=\scriptsize, text=TUMBlue] {$p_{i_1j_1} = 10$} (j1);
      \draw[->, TUMBlue] (i2) -- node[midway, below, font=\scriptsize, text=TUMBlue] {$p_{i_2j_2} = 5$} (j2);
      \draw[->, dashed, TUMBlue] (i1) -- node[midway, above, sloped, font=\scriptsize, text=TUMBlue] {$q_{i_1j_2} = 5$} (j2);

      \begin{scope}[on background layer]
        \node[draw=black, line width=0.5pt, rounded corners=5mm, inner sep=10pt, fit=(diagram)] {};
      \end{scope}
    \end{scope}
  \end{tikzpicture}

  \caption{Graph representation of the assignment $((\vecc{S}_I, \vecc{Z}_J), \vecc p, \vecc{Q}_N)$ in Example~\ref{example:FullnottoSep(P)}.}
  \label{figure:FullnottoSep(P)}
\end{figure}

Let $((\vecc{\tilde S}_I, \vecc{\tilde Z}_J), \vecc{\tilde p}, \vecc{\tilde Q}_N) \in \mathcal{F}(\mathcal{E}_{\tilde h}^{\tilde r})$ with $\tilde h \in \{1, \dots, 5\}$ and $\tilde r \in \{\mathrm{sep}, \, \mathrm{buy}, \, \mathrm{sell}, \, \mathrm{no}\}$. In order for every agent $k \in N$ to maintain the same cumulative price $\tilde p_k = p_k$, their cumulative bundle must not be empty (otherwise, $\tilde p_k = 0 \neq p_k$). Therefore, the only admissible allocation we can consider is $(\vecc{\tilde S}_I, \vecc{\tilde Z}_J) = (\vecc{S}_I, \vecc{Z}_J)$.
However, since side payments from buyers to sellers are not allowed, the minimum total price that the sellers need to receive is greater than the one of the initial assignment
$$\tilde p_{j_1} + \tilde p_{j_2} \geq v_{j_1}(\{a\}) + v_{j_2}(\{b\}) = 18 > 15 = p_{j_1} + p_{j_2},$$
so at least $\tilde p_{j_1} \neq p_{j_1}$ or $\tilde p_{j_2} \neq p_{j_2}$. We can then conclude that there is no feasible assignment with at most separate side payments that is price equivalent to $((\vecc{S}_I, \vecc{Z}_J), \vecc p, \vecc Q_N)$.
\end{example}

\subsubsection{To the Buyers' Single Side-Payment Settings}
\begin{example}\label{example:1BuySepFullnotto...}
Consider an economy $\mathcal{E}$ with three buyers $I = \{i_1, i_2, i_3\}$ and two sellers $J = \{j_1, j_2\}$, where $i_3$ is a buyer whose valuation is zero for every bundle. The initial endowments of the sellers are $\vecc\Omega_{j_1} = \vecc\Omega_{j_2} \equiv \{a\}$, so that the total endowment of the economy is $\vecc\Omega \equiv \{a, a\}$. The agents' valuations are reported in Table~\ref{table:1BuySepFullnotto..._valuations}.

\begin{table}[ht]
  \centering
  \begin{tabular}{|c|c|c|c|}
    \hline
    $\vecc{\omega}$ & $\{a\}$ & $\{a, a\}$ \\
    \hline
    $v_{i_1}(\vecc{\omega})$ & $6$ & $6$ \\
    \hline
    $v_{i_2}(\vecc{\omega})$ & $17$ & $17$ \\
    \hline
    $v_{i_3}(\vecc{\omega})$ & $0$ & $0$ \\
    \hline
    $v_{j_1}(\vecc{\omega})$ & $9$ & $-$ \\
    \hline
    $v_{j_2}(\vecc{\omega})$ & $9$ & $-$ \\
    \hline
  \end{tabular}

  \caption{Valuations $v_k(\vecc{\omega})$ for each agent $k \in N$ and each bundle $\vecc{\omega} \leq \vecc{\Omega}$ in Example~\ref{example:1BuySepFullnotto...}.}
  \label{table:1BuySepFullnotto..._valuations}
\end{table}

Consider the allocation $(\vecc{S}_I, \vecc{Z}_J)$ where agents $i_1$ and $j_1$ trade bundle $\vecc{s}_{i_1j_1} = \vecc{z}_{i_1j_1} \equiv \{a\}$, and agents $i_2$ and $j_2$ trade bundle $\vecc{s}_{i_2j_2} = \vecc{z}_{i_2j_2} \equiv \{a\}$. The corresponding pricing system $\vecc{p}$ is such that $p_a = 10$. Moreover, buyer $i_2$ side pays buyer $i_1$ an amount $q_{i_2i_1} = - q_{i_1i_2} = 5$, and buyer $i_3$ an amount $q_{i_2i_3} = -q_{i_3i_2} = 1$.

Since the allocation is fourth-order admissible and the pricing system is first-order admissible, the assignment $((\vecc{S}_I, \vecc{Z}_J), \vecc p, \vecc{Q}_I)$ is admissible in all the order settings (including the fifth-order one if we consider the cumulative equivalent $((\vecc{s}_I, \vecc{z}_J), \vecc p, \vecc{Q}_I)$). However, it is only admissible in the full, separate, and buyers' single side-payment settings.
The payoff is non-negative for each agent, in particular, $\pi_k = 1$ for every $k \in N$, so that the assignment is feasible in all the settings above. We provide a representation of the payment graph in Figure~\ref{figure:1BuySepFullnotto...}.

\begin{figure}[ht]
  \centering
  \begin{tikzpicture}[>=Stealth, thick, node/.style={draw, circle, fill=white, minimum size=0.7cm}]
    \begin{scope}[local bounding box=diagram]
      \node[node] (i1) at (0,2) {$i_1$};
      \node[left=0.15cm of i1, font=\small] (text-i1) {
        \begin{tabular}{r|l} $v_{i_1}(\{a\}) = 6$ & \textcolor{TUMBlue}{$\pi_{i_1} = 1$}
      \end{tabular}};

      \node[node] (i2) at (0,0) {$i_2$};
      \node[left=0.15cm of i2, font=\small] (text-i2) {
        \begin{tabular}{r|l} $v_{i_2}(\{a\}) = 17$ & \textcolor{TUMBlue}{$\pi_{i_2} = 1$}
      \end{tabular}};

      \node[node] (i3) at (0,-2) {$i_3$};
      \node[left=0.15cm of i3, font=\small] (text-i3) {
        \begin{tabular}{l} \textcolor{TUMBlue}{$\pi_{i_3} = 1$}
      \end{tabular}};

      \node[node] (j1) at (3.5,2) {$j_1$};
      \node[right=0.15cm of j1, font=\small] (text-j1) {
        \begin{tabular}{r|l} \textcolor{TUMBlue}{$\pi_{j_1} = 1$} & $v_{j_1}(\{a\}) = 9$
      \end{tabular}};

      \node[node] (j2) at (3.5,0) {$j_2$};
      \node[right=0.15cm of j2, font=\small] (text-j2) {
        \begin{tabular}{r|l} \textcolor{TUMBlue}{$\pi_{j_2} = 1$} & $v_{j_2}(\{a\}) = 9$
      \end{tabular}};

      \begin{scope}[on background layer]
        \foreach \n/\t in {i1/text-i1, i2/text-i2, i3/text-i3, j1/text-j1, j2/text-j2}
        \node[fill=gray!10, rounded corners=3mm, inner sep=1.2mm, fit=(\n) (\t)] {};
      \end{scope}

      \draw[->, TUMBlue] (i1) -- node[midway, above, font=\scriptsize, text=TUMBlue] {$p_{i_1j_1} = 10$} (j1);
      \draw[->, TUMBlue] (i2) -- node[midway, above, font=\scriptsize, text=TUMBlue] {$p_{i_2j_2} = 10$} (j2);
      \draw[->, dashed, TUMBlue] (i2) -- node[midway, left, font=\scriptsize, text=TUMBlue] {$q_{i_2i_1} = 5$} (i1);
      \draw[->, dashed, TUMBlue] (i2) -- node[midway, left, font=\scriptsize, text=TUMBlue] {$q_{i_2i_3} = 1$} (i3);

      \begin{scope}[on background layer]
        \node[draw=black, line width=0.5pt, rounded corners=5mm, inner sep=10pt, fit=(diagram)] {};
      \end{scope}
    \end{scope}
  \end{tikzpicture}

  \caption{Graph representation of the assignment $((\vecc{S}_I, \vecc{Z}_J), \vecc p, \vecc{Q}_I)$ in Example~\ref{example:1BuySepFullnotto...}.}
  \label{figure:1BuySepFullnotto...}
\end{figure}
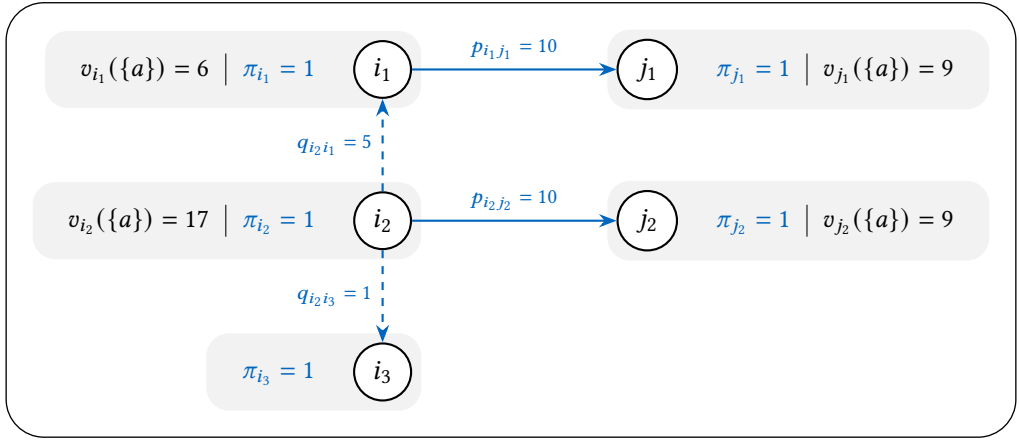

Let's now analyze some of the settings in which we can't find any feasible assignment that is price or payoff equivalent to $((\vecc{S}_I, \vecc{Z}_J), \vecc p, \vecc{Q}_I)$.
For simplicity, we will always refer to the allocation in the new settings as $((\vecc{\tilde S}_I, \vecc{\tilde Z}_J), \vecc{\tilde p}, \vecc{\tilde Q}_N)$.

We can already observe that, if at least one of the cumulative bundles is empty (except for $\vecc{s}_{i_3}$), not even item equivalence is possible. We will then always assume in the following that $(\vecc{\tilde s}_I, \vecc{\tilde z}_J) = (\vecc{s}_I, \vecc{z}_J)$.

\paragraph{First- and second-order settings with sellers' single or no side payments} If side payments among buyers are not allowed, buyer $i_1$ can only afford to pay for bundle $\{a\}$ a maximum price of
$$\tilde p_{i_1} = \tilde p(\{a\}) = \tilde p_a \leq v_{i_1}(\{a\}) = 6$$
to one of the sellers. However, the minimum price that each seller $j \in J$ must receive is
$$\tilde p_j = \tilde p(\{a\}) = \tilde p_a \geq \frac{v_{j_1}(\{a\}) + v_{j_2}(\{a\})}{2} = 9$$
in the seller's single side-payment setting, and
$$\tilde p_j = \tilde p(\{a\}) = \tilde p_{a} \geq \max\{v_{j_1}(\{a\}), v_{j_2}(\{a\})\} = 9$$
in the no side-payment setting. Hence, any allocation where $\vecc{\tilde{s}}_{i_1} \neq \vecc{0}$ is not priceable in these settings, and we deduce that there is no feasible assignment in the first- and second-order settings with sellers' single or no side payments that is even item equivalent to $((\vecc{S}_I, \vecc{Z}_J), \vecc p, \vecc{Q}_I)$.

\paragraph{Third- and fourth-order settings with no side payments} Since no side payments are allowed in these settings, both sellers $j \in J$ need to receive a cumulative price $\tilde p_j \geq v_{j}(\{a\}) = 9$, so the same reasoning above applies.
Therefore, there is no feasible assignment in the third- and fourth-order settings with no side payments that is even item equivalent to $((\vecc{S}_I, \vecc{Z}_J), \vecc p, \vecc{Q}_I)$.

\paragraph{Third- and fourth-order settings with sellers' single side payments} If we allow only side payments among the sellers, as we have observed above, the maximum cumulative price that buyer $i_1$ can pay is strictly less than the one they paid in the initial assignment, i.e., $\tilde p_{i_1} \leq v_{i_1}(\{a\}) = 6 < p_{i_1} = 10$.
Also, buyer $i_3$ can't receive any side payment anymore, and so their payoff must be $\tilde\pi_{i_3} = 0 < \pi_{i_3} = 1$.
Hence, there is no feasible assignment in the third- and fourth-order settings with sellers' single side payments that is payoff equivalent to $((\vecc{S}_I, \vecc{Z}_J), \vecc p, \vecc{Q}_I)$.

\paragraph{Fifth-order setting with sellers' single or no side payments} For the fifth-order setting, we can make the same observations on the prices and payoffs as in the previous point. Therefore, there is no feasible assignment in the fifth-order setting with sellers' single or no side payments that is payoff equivalent to $((\vecc{S}_I, \vecc{Z}_J), \vecc p, \vecc{Q}_I)$.
\end{example}

\subsubsection{To the Sellers' Single Side-Payment Settings}
\begin{example}\label{example:1SellSepFullnotto...}
Consider an economy $\mathcal{E}$ with two buyers $I = \{i_1, i_2\}$ and four sellers $J = \{j_1, \dots, j_4\}$, where $j_4$ is a seller who doesn't own any item. The initial endowments of the sellers are $\vecc\Omega_{j_1} = \vecc\Omega_{j_2} \equiv \{a\}$ and $\vecc\Omega_{j_3} \equiv \{b\}$, so that the total endowment of the economy is $\vecc\Omega \equiv \{a, a, b\}$. The agents' valuations are reported in Table~\ref{table:1SellSepFullnotto..._valuations}.

\begin{table}[ht]
  \centering
  \begin{tabular}{|c|c|c|c|c|c|}
    \hline
    $\vecc{\omega}$ & $\{a\}$ & $\{b\}$ & $\{a, a\}$ & $\{a, b\}$ & $\{a, a, b\}$ \\
    \hline
    $v_{i_1}(\vecc{\omega})$ & $0$ & $0$ & $21$ & $0$ & $21$ \\
    \hline
    $v_{i_2}(\vecc{\omega})$ & $0$ & $11$ & $0$ & $11$ & $11$ \\
    \hline
    $v_{j_1}(\vecc{\omega})$ & $4$ & $-$ & $-$ & $-$ & $-$ \\
    \hline
    $v_{j_2}(\vecc{\omega})$ & $20$ & $-$ & $-$ & $-$ & $-$ \\
    \hline
    $v_{j_3}(\vecc{\omega})$ & $-$ & $2$ & $-$ & $-$ & $-$ \\
    \hline
    $v_{j_4}(\vecc{\omega})$ & $-$ & $-$ & $-$ & $-$ & $-$ \\
    \hline
  \end{tabular}

  \caption{Valuations $v_k(\vecc{\omega})$ for each agent $k \in N$ and each bundle $\vecc{\omega} \leq \vecc{\Omega}$ in Example~\ref{example:1SellSepFullnotto...}.}
  \label{table:1SellSepFullnotto..._valuations}
\end{table}

Consider the allocation $(\vecc{S}_I, \vecc{Z}_J)$ where buyer $i_1$ trades with sellers $j_1$ and $j_2$ bundle $\vecc{s}_{i_1j_1} = \vecc{s}_{i_1j_2} = \vecc{z}_{i_1j_1} = \vecc{z}_{i_1j_2} \equiv \{a\}$,  and agents $i_2$ and $j_3$ trade bundle $\vecc{s}_{i_2j_3} = \vecc{z}_{i_2j_3} \equiv \{b\}$. The corresponding pricing system $\vecc{p}$ is such that $p_a = 10$ and $p_b = 10$. Moreover, seller $j_2$ receives side payments from $j_1$ and $j_3$ of respectively $q_{j_1j_2} = - q_{j_2j_1} = 5$ and $q_{j_3j_2} = - q_{j_2j_3} = 6$, and seller $j_3$ side-pays seller $j_4$ an amount $q_{j_3j_4} = -q_{j_4j_3} = 1$.

Since the allocation is fourth-order admissible and the pricing system is first-order admissible, the assignment $((\vecc{S}_I, \vecc{Z}_J), \vecc p, \vecc{Q}_J)$ is admissible in all the order settings (including the fifth-order one if we consider the cumulative equivalent $((\vecc{s}_I, \vecc{z}_J), \vecc p, \vecc{Q}_J)$). However, it is only admissible in the full, separate, and sellers' single side-payment settings.
The payoff is non-negative for each agent, in particular, $\pi_k = 1$ for every $k \in N$, so that the assignment is feasible in all the settings considered above. We provide a representation of the payment graph in Figure~\ref{figure:1SellSepFullnotto...}.

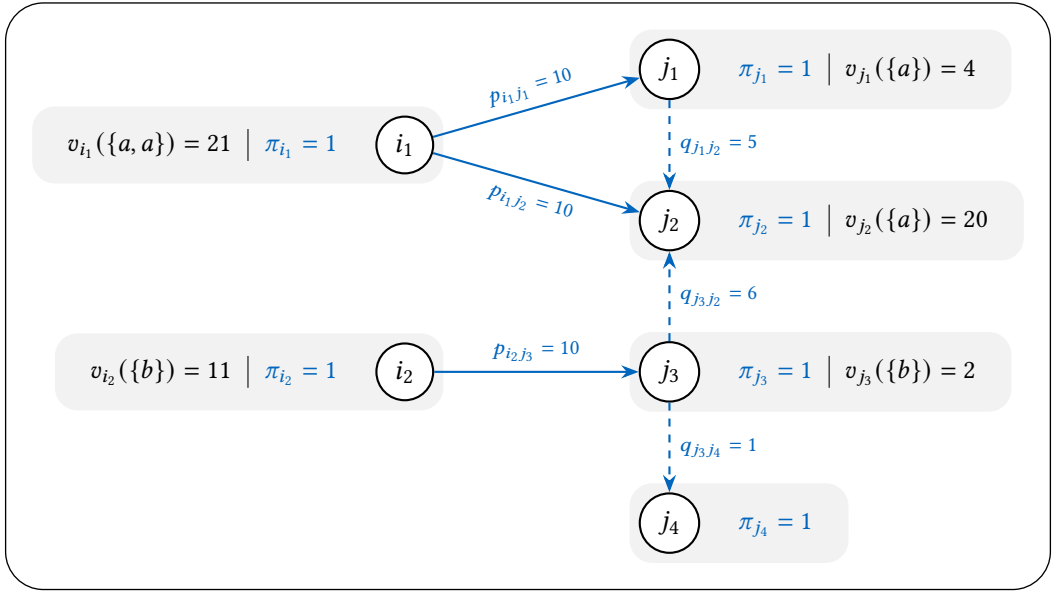
\begin{figure}[ht]
  \centering
  \begin{tikzpicture}[>=Stealth, thick, node/.style={draw, circle, fill=white, minimum size=0.7cm}]
    \begin{scope}[local bounding box=diagram]
      \node[node] (i1) at (0,2) {$i_1$};
      \node[left=0.15cm of i1, font=\small] (text-i1) {
        \begin{tabular}{r|l} $v_{i_1}(\{a, a\}) = 21$ & \textcolor{TUMBlue}{$\pi_{i_1} = 1$}
      \end{tabular}};

      \node[node] (i2) at (0,-1) {$i_2$};
      \node[left=0.15cm of i2, font=\small] (text-i2) {
        \begin{tabular}{r|l} $v_{i_2}(\{b\}) = 11$ & \textcolor{TUMBlue}{$\pi_{i_2} = 1$}
      \end{tabular}};

      \node[node] (j1) at (3.5,3) {$j_1$};
      \node[right=0.15cm of j1, font=\small] (text-j1) {
        \begin{tabular}{r|l} \textcolor{TUMBlue}{$\pi_{j_1} = 1$} & $v_{j_1}(\{a\}) = 4$
      \end{tabular}};

      \node[node] (j2) at (3.5,1) {$j_2$};
      \node[right=0.15cm of j2, font=\small] (text-j2) {
        \begin{tabular}{r|l} \textcolor{TUMBlue}{$\pi_{j_2} = 1$} & $v_{j_2}(\{a\}) = 20$
      \end{tabular}};

      \node[node] (j3) at (3.5,-1) {$j_3$};
      \node[right=0.15cm of j3, font=\small] (text-j3) {
        \begin{tabular}{r|l} \textcolor{TUMBlue}{$\pi_{j_3} = 1$} & $v_{j_3}(\{b\}) = 2$
      \end{tabular}};

      \node[node] (j4) at (3.5,-3) {$j_4$};
      \node[right=0.15cm of j4, font=\small] (text-j4) {
        \begin{tabular}{l} \textcolor{TUMBlue}{$\pi_{j_4} = 1$}
      \end{tabular}};

      \begin{scope}[on background layer]
        \foreach \n/\t in {i1/text-i1, i2/text-i2, j1/text-j1, j2/text-j2, j3/text-j3, j4/text-j4}
        \node[fill=gray!10, rounded corners=3mm, inner sep=1.2mm, fit=(\n) (\t)] {};
      \end{scope}

      \draw[->, TUMBlue] (i1) -- node[midway, above, sloped, font=\scriptsize, text=TUMBlue] {$p_{i_1j_1} = 10$} (j1);
      \draw[->, TUMBlue] (i1) -- node[midway, below, sloped, font=\scriptsize, text=TUMBlue] {$p_{i_1j_2} = 10$} (j2);
      \draw[->, TUMBlue] (i2) -- node[midway, above, sloped, font=\scriptsize, text=TUMBlue] {$p_{i_2j_3} = 10$} (j3);

      \draw[->, dashed, TUMBlue] (j1) -- node[midway, right, font=\scriptsize, text=TUMBlue] {$q_{j_1j_2} = 5$} (j2);
      \draw[->, dashed, TUMBlue] (j3) -- node[midway, right, font=\scriptsize, text=TUMBlue] {$q_{j_3j_2} = 6$} (j2);
      \draw[->, dashed, TUMBlue] (j3) -- node[midway, right, font=\scriptsize, text=TUMBlue] {$q_{j_3j_4} = 1$} (j4);

      \begin{scope}[on background layer]
        \node[draw=black, line width=0.5pt, rounded corners=5mm, inner sep=10pt, fit=(diagram)] {};
      \end{scope}
    \end{scope}
  \end{tikzpicture}

  \caption{Graph representation of the assignment $((\vecc{S}_I, \vecc{Z}_J), \vecc p, \vecc{Q}_J)$ in Example~\ref{example:1SellSepFullnotto...}.}
  \label{figure:1SellSepFullnotto...}
\end{figure}

Let's now analyze some of the settings in which we can't find any feasible assignment that is payoff equivalent to $((\vecc{S}_I, \vecc{Z}_J), \vecc p, \vecc{Q}_J)$. For simplicity, we will always indicate the new assignment with $((\vecc{\tilde S}_I, \vecc{\tilde Z}_J), \vecc{\tilde p}, \vecc{\tilde Q}_N)$.

If at least one of the cumulative bundles of the agents is empty (except for $\vecc{z}_{j_4}$) or $\vecc{\tilde s}_{i_1} = \{a\}$, it is immediate to see that the new assignment can't be even item equivalent to the initial one.
Hence, we will always assume in the following that $(\vecc{\tilde s}_I, \vecc{\tilde z}_J) = (\vecc{s}_I, \vecc{z}_J)$.

\paragraph{First-, second-, and third-order setting with buyers' single or no side payments} If side payments among sellers are not allowed, seller $j_2$ can only receive a price $\tilde p_{j_2} \geq v_{j_1}(\{a\}) = 22$ if they trade a non-empty bundle. Hence, for buyer $i_1$ to have the same cumulative bundle $\vecc{\tilde s}_{i_1} = \vecc s_{i_1}$, their price for bundle $\{a\}$ must be
$$\tilde p_{i_1}(\{a\}) \geq \max\{v_{j_1}(\{a\}), v_{j_2}(\{a\})\} = 22.$$
However, this would mean that the cumulative price would become $\tilde p_{i_1} \geq 2\tilde p_{i_1}(\{a\}) = 44$, but even the total maximum valuation for the buyers is only
$$v_{i_1}(\vecc\Omega) + v_{i_2}(\vecc\Omega) = 32 < 44 \leq \tilde p_{i_1}.$$
So there is no feasible assignment in the first-, second-, and third-order setting with buyers' or no side payments that is even item equivalent to $((\vecc{S}_I, \vecc{Z}_J), \vecc p, \vecc{Q}_J)$.

\paragraph{Fourth-order setting with no side payments} Since no side payments are allowed in this setting, buyer $i_1$ can afford to buy the cumulative bundle $\vecc{s}_{i_1} = \{a, a\}$ only if they can pay a price of
$$\tilde p_{i_1} \geq v_{j_1}(\{a\}) + v_{j_2}(\{a\}) = 26,$$
but their valuation is only $v_{i_1}(\{a, a\}) = 21 < 26$. Therefore, there is no feasible assignment in the fourth-order setting with no side payments that is even item equivalent to $((\vecc{S}_I, \vecc{Z}_J), \vecc p, \vecc{Q}_J)$.

\paragraph{Fourth-order setting with buyers' single side payments} If we allow only side payments among the buyers, as we have observed before, the cumulative price that seller $j_1$ must receive is $\tilde p_{i_1} \geq 22 > p_{j_1} = 10$. Also, seller $j_4$ can't receive any side payment anymore, and their payoff becomes $\tilde\pi_{j_4} = 0 < \pi_{j_4} = 1$.
Hence, there is no feasible fourth-order admissible assignment with buyers' side payments that is payoff equivalent to $((\vecc{S}_I, \vecc{Z}_J), \vecc p, \vecc{Q}_J)$.

\paragraph{Fifth-order setting with buyers' single or no side payments} For the fifth-order setting, we can make the same observation on the prices and payoffs as for the fourth-order setting with buyers' single side payments. Therefore, there is no feasible assignment in the fifth-order setting with buyers' or no side payments that is payoff equivalent to $((\vecc{S}_I, \vecc{Z}_J), \vecc p, \vecc{Q}_J)$.
\end{example}

\subsection{Stability Counterexamples}
In this section, we present counterexamples related to stability, focusing in particular on how the social welfare achieved by stable assignments varies across different notions of stability. The examples are the following:
\begin{enumerate}
\item Example~\ref{example:NTUDifferentPayoffs} shows that, in some economies $\mathcal{E}$ with buyers' single or no side-payment restrictions, there exist $h^{\mathrm{th}}$-order admissible assignments for $h \in \{1, 2, 3\}$ that are NTU-stable but achieve different levels of social welfare.
\item Example~\ref{example:sRTUCorePriceability} shows that, in some economies $\mathcal{E}$ with sellers' single (or no) side-payment restrictions, there exist $h^{\mathrm{th}}$-order admissible assignments for $h \in \{1, 2\}$ that are NTU-stable but achieve different levels of social welfare.
\item Example~\ref{example:RTUUniquePayoff} shows that, in some economies $\mathcal{E}$ with first- or second-order resale restrictions, the maximum attainable social welfare differs across the corresponding stability notions for TU-, bRTU-, sRTU-, and NTU-stable vectors.
\item Example~\ref{example:emptyCores} shows that, for some economies $\mathcal{E}$ with any resale $h \in \{1, \dots, 5\}$ and side-payment restrictions $r \in \{\mathrm{full}, \, \mathrm{sep}, \, \mathrm{buy}, \, \mathrm{sell}, \, \mathrm{no}\}$, the $\mathcal{T}$-core is empty for every partition $\mathcal{T}$.
\item Example~\ref{example:sidePaymentsAreCool} and Example~\ref{example:sidePaymentsAreCool2} show that there exist economies $\mathcal{E}$ with first- or second-order resale restrictions for which the TU-core is empty, whereas the NTU-core is non-empty. Moreover, allowing side payments on only one side of the market (specifically, on the buyers' side in the first example and on the sellers' side in the second) yields NTU-stable assignments whose social welfare is arbitrarily larger than that of any stable assignment in the no side-payment setting.
\end{enumerate}

\subsubsection{NTU-Stable Assignments With Different Social Welfare}
\begin{example}\label{example:NTUDifferentPayoffs}
Consider an economy $\mathcal{E}$ with two buyers $I = \{i_1, i_2\}$ and three sellers $J = \{j_1, j_2, j_3\}$. The sellers' initial endowments are $\vecc\Omega_{j_1} \equiv \{a\}$ and $\vecc\Omega_{j_2} = \vecc\Omega_{j_3} \equiv \{b\}$, so that the total endowment of the economy is $\vecc\Omega \equiv \{a, b, b\}$. The agents' valuations are reported in Table~\ref{table:NTUCorePriceability}.

\begin{table}[ht]
  \centering
  \begin{tabular}{|c|c|c|c|c|c|c|c|c|}
    \hline
    $\vecc{\omega}$ & $\{a\}$ & $\{b\}$ & $\{a, b\}$ & $\{b, b\}$ & $\{a, b, b\}$ \\
    \hline
    $v_{i_1}(\vecc{\omega})$ & $6$ & $0$ & $6$ & $0$ & $6$ \\
    \hline
    $v_{i_2}(\vecc{\omega})$ & $0$ & $0$ & $0$ & $0$ & $14$  \\
    \hline
    $v_{j_1}(\vecc{\omega})$ & $6$ & $-$ & $-$ & $-$ & $-$ \\
    \hline
    $v_{j_2}(\vecc{\omega})$ & $-$ & $4$ & $-$ & $-$ & $-$ \\
    \hline
    $v_{j_3}(\vecc{\omega})$ & $-$ & $3$ & $-$ & $-$ & $-$ \\
    \hline
  \end{tabular}

  \caption{Valuations $v_k(\vecc{\omega})$ for each agent $k \in N$ and each bundle $\vecc{\omega} \leq \vecc{\Omega}$ in Example~\ref{example:NTUDifferentPayoffs}.}
  \label{table:NTUCorePriceability}
\end{table}

\begin{figure}[ht]
  \centering
  \begin{tikzpicture}[>=Stealth, thick, node/.style={draw, circle, fill=white, minimum size=0.7cm}]
    \begin{scope}[local bounding box=diagram]
      \node[node] (i1) at (0,2) {$i_1$};
      \node[left=0.15cm of i1, font=\small] (text-i1) {
        \begin{tabular}{r|l} $v_{i_1}(\{a\}) = 6$ & \textcolor{TUMBlue}{$\pi_{i_1} = 0$} \\ & \textcolor{TUMGreen}{$\tilde\pi_{i_1} = 0$}
      \end{tabular}};

      \node[node] (i2) at (0,0) {$i_2$};
      \node[left=0.15cm of i2, font=\small] (text-i2) {
        \begin{tabular}{r|l} $v_{i_2}(\{a, b, b\}) = 14$ & \textcolor{TUMBlue}{$\pi_{i_2} = 0$} \\ & \textcolor{TUMGreen}{$\tilde\pi_{i_2} = 0$}
      \end{tabular}};

      \node[node] (j1) at (3.5,2) {$j_1$};
      \node[right=0.15cm of j1, font=\small] (text-j1) {
        \begin{tabular}{r|l} \textcolor{TUMBlue}{$\pi_{j_1} = 0$} & $v_{j_1}(\{a\}) = 6$ \\ \textcolor{TUMGreen}{$\tilde\pi_{j_1} = 0$} &
      \end{tabular}};

      \node[node] (j2) at (3.5,0) {$j_2$};
      \node[right=0.15cm of j2, font=\small] (text-j2) {
        \begin{tabular}{r|l} \textcolor{TUMBlue}{$\pi_{j_2} = 0$} & $v_{j_2}(\{b\}) = 4$ \\ \textcolor{TUMGreen}{$\tilde\pi_{j_2} = 0$} &
      \end{tabular}};

      \node[node] (j3) at (3.5,-2) {$j_3$};
      \node[right=0.15cm of j3, font=\small] (text-j3) {
        \begin{tabular}{r|l} \textcolor{TUMBlue}{$\pi_{j_3} = 1$} & $v_{j_3}(\{b\}) = 3$ \\ \textcolor{TUMGreen}{$\tilde\pi_{j_3} = 0$} &
      \end{tabular}};

      \begin{scope}[on background layer]
        \foreach \n/\t in {i1/text-i1, i2/text-i2, j1/text-j1, j2/text-j2, j3/text-j3}
        \node[fill=gray!10, rounded corners=3mm, inner sep=1.2mm, fit=(\n) (\t)] {};
      \end{scope}

      \draw[->, TUMGreen] (i1) -- node[midway, above, font=\scriptsize] {$\tilde p_{i_1j_1} = 6$} (j1);
      \draw[->, TUMBlue] (i2) -- node[midway, below, sloped, font=\scriptsize] {$p_{i_2j_1} = 6$} (j1);
      \draw[->, TUMBlue] (i2) -- node[midway, below, sloped, font=\scriptsize] {$p_{i_2j_2} = 4$} (j2);
      \draw[->, TUMBlue] (i2) -- node[midway, below, sloped, font=\scriptsize] {$p_{i_2j_3} = 4$} (j3);

      \begin{scope}[on background layer]
        \node[draw=black, line width=0.5pt, rounded corners=5mm, inner sep=10pt, fit=(diagram)] {};
      \end{scope}
    \end{scope}
  \end{tikzpicture}

  \caption{Graph representation of the stable assignments in Example~\ref{example:NTUDifferentPayoffs}, $((\vecc{S}_I, \vecc{Z}_J), \vecc p)$ with blue arrows and $((\vecc{\tilde S}_I, \vecc{\tilde Z}_J), \vecc{p})$ with green arrows.}
  \label{figure:NTUCorePriceability}
\end{figure}
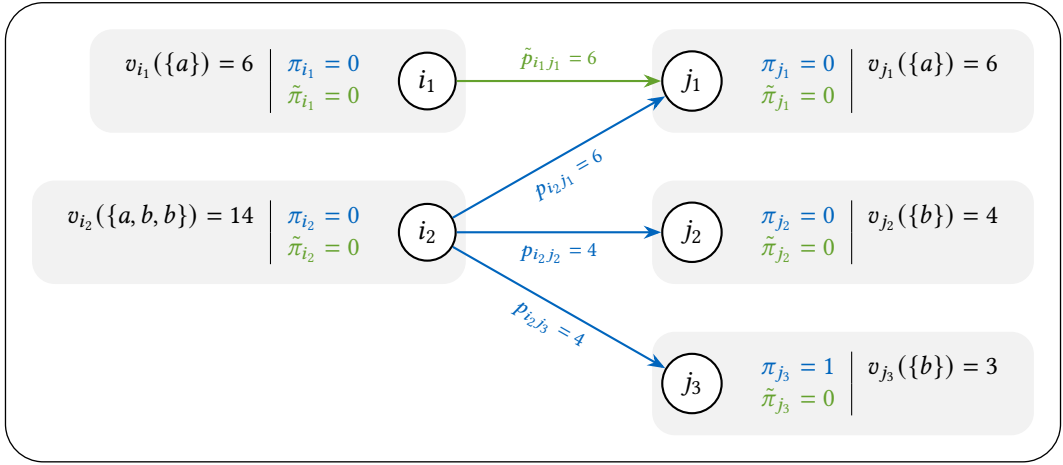

First, consider an allocation $(\vecc{S}_I, \vecc{Z}_J)$ in which buyer $i_2$ requests all the items in the economy. In particular, let $\vecc{s}_{i_2j_1} = \vecc{z}_{i_2j_1} \equiv \{a\}$ and $\vecc{s}_{i_2j_2} = \vecc{z}_{i_2j_2} = \vecc{s}_{i_2j_3} = \vecc{z}_{i_2j_3} \equiv \{b\}$, while all the remaining individual bundles are empty. Define the pricing system $\vecc{p}$ by $p_a = 6$ and $p_b = 4$. Since the allocation is third-order admissible and the pricing system is first-order admissible, the assignment $((\vecc{S}_I, \vecc{Z}_J), \vecc p)$ is first-, second-, and third-order admissible. The resulting payoffs are $\pi_{j_3} = 1$ and $\pi_ k = 0$ for any other agent $k \in N \setminus \{j_3\}$, so that the total payoff is $\sum_{k \in N} \pi_k = 1$  (see Figure~\ref{figure:NTUCorePriceability}).
We can immediately show that this assignment is NTU-stable in any market setting $\mathcal{E}_h^r$, including for every $h \in \{1, 2, 3\}$ and $r \in \{\mathrm{buy},\, \mathrm{no}\}$. Observe that both buyers compete for the only unit of item $a$. However, buyer $i_1$ values $\{a\}$ at exactly $v_{i_1}(\{a\}) = 6$, and therefore cannot profitably offer seller $j_1$ a price greater than $p_{j_1} = p_a = 6$ for its purchase. Consequently, no deviating coalition exists.

Consider next the assignment $((\vecc{\tilde S}_I, \vecc{\tilde Z}_J), \vecc{p})$ in which $\vecc{\tilde s}_{i_1j_1} = \vecc{\tilde z}_{i_1j_1} \equiv \{a\}$, while all other bundles are empty. The corresponding total payoff is $\sum_{k \in N} \tilde\pi_k = 0$ (see Figure~\ref{figure:NTUCorePriceability}).
In this case, the assignment is NTU-stable only in the markets $\mathcal{E}_h^r$ with $h \in \{1, 2, 3\}$ and $r \in \{\mathrm{buy},\, \mathrm{no}\}$, since only in these settings there is no coalition in which seller $j_1$ can sell their item to $i_2$ and achieve a positive payoff. Indeed, to get a positive payoff themselves, buyer $i_2$ must acquire the entire bundle $\{a, b, b\}$ for a cumulative price strictly smaller than $v_{i_2}(\{a, b, b\})=14$. However, because of the pricing restrictions, they must pay a price of at least $v_{j_1}(\{a\}) = 6$ to seller $j_1$ and a price of at least $$\max\{v_{j_2}(\{b\}), v_{j_3}(\{b\})\} = 4$$ to both $j_2$ and $j_3$. Therefore, $i_2$'s cumulative price in the deviating coalition would need to be at least $6 + 2 \cdot 4 = 14$, leaving $i_2$ with a non-positive payoff.
\end{example}

\begin{example}\label{example:sRTUCorePriceability}
Consider an economy $\mathcal{E}$ with three buyers $I = \{i_1, i_2, i_3\}$ and two sellers $J = \{j_1, j_2\}$. The sellers' initial endowments are $\vecc\Omega_{j_1} \equiv \{a\}$ and $\vecc\Omega_{j_2} \equiv \{a, b, b\}$, so that the total endowment of the economy is $\vecc\Omega \equiv \{a, a, b, b\}$. The agents' valuations are reported in Table~\ref{table:sRTUCorePriceability}.

\begin{table}[ht]
  \centering
  \begin{tabular}{|c|c|c|c|c|c|c|c|c|c|c|}
    \hline
    $\vecc{\omega}$ & $\{a\}$ & $\{b\}$ & $\{a, a\}$ & $\{a, b\}$ & $\{b, b\}$ & $\{a, a, b\}$ & $\{a, b, b\}$ & $\{a, a, b, b\}$ \\
    \hline
    $v_{i_1}(\vecc{\omega})$ & $6$ & $0$ & $6$ & $6$ & $0$ & $6$ & $6$ & $6$ \\
    \hline
    $v_{i_2}(\vecc{\omega})$ & $0$ & $4$ & $0$ & $4$ & $4$ & $4$ & $4$ & $4$  \\
    \hline
    $v_{i_3}(\vecc{\omega})$ & $0$ & $3$ & $0$ & $3$ & $3$ & $3$ & $3$ & $3$  \\
    \hline
    $v_{j_1}(\vecc{\omega})$ & $6$ & $-$ & $-$ & $-$ & $-$ & $-$ & $-$ & $-$ \\
    \hline
    $v_{j_2}(\vecc{\omega})$ & $12$ & $12$ & $-$ & $12$ & $12$ & $-$ & $12$ & $-$ \\
    \hline
  \end{tabular}

  \caption{Valuations $v_k(\vecc{\omega})$ for each agent $k \in N$ and each bundle $\vecc{\omega} \leq \vecc{\Omega}$ in Example~\ref{example:sRTUCorePriceability}.}
  \label{table:sRTUCorePriceability}
\end{table}

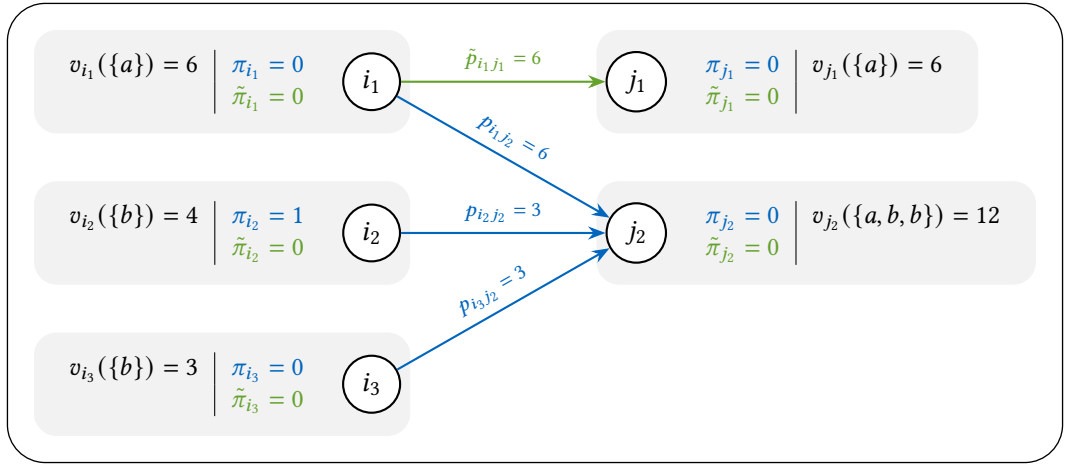
\begin{figure}[ht]
  \centering
  \begin{tikzpicture}[>=Stealth, thick, node/.style={draw, circle, fill=white, minimum size=0.7cm}]
    \begin{scope}[local bounding box=diagram]
      \node[node] (i1) at (0,2) {$i_1$};
      \node[left=0.15cm of i1, font=\small] (text-i1) {
        \begin{tabular}{r|l} $v_{i_1}(\{a\}) = 6$ & \textcolor{TUMBlue}{$\pi_{i_1} = 0$} \\ & \textcolor{TUMGreen}{$\tilde\pi_{i_1} = 0$}
      \end{tabular}};

      \node[node] (i2) at (0,0) {$i_2$};
      \node[left=0.15cm of i2, font=\small] (text-i2) {
        \begin{tabular}{r|l} $v_{i_2}(\{b\}) = 4$ & \textcolor{TUMBlue}{$\pi_{i_2} = 1$} \\ & \textcolor{TUMGreen}{$\tilde\pi_{i_2} = 0$}
      \end{tabular}};

      \node[node] (i3) at (0,-2) {$i_3$};
      \node[left=0.15cm of i3, font=\small] (text-i3) {
        \begin{tabular}{r|l} $v_{i_3}(\{b\}) = 3$ & \textcolor{TUMBlue}{$\pi_{i_3} = 0$} \\ & \textcolor{TUMGreen}{$\tilde\pi_{i_3} = 0$}
      \end{tabular}};

      \node[node] (j1) at (3.5,2) {$j_1$};
      \node[right=0.15cm of j1, font=\small] (text-j1) {
        \begin{tabular}{r|l} \textcolor{TUMBlue}{$\pi_{j_1} = 0$} & $v_{j_1}(\{a\}) = 6$ \\ \textcolor{TUMGreen}{$\tilde\pi_{j_1} = 0$} &
      \end{tabular}};

      \node[node] (j2) at (3.5,0) {$j_2$};
      \node[right=0.15cm of j2, font=\small] (text-j2) {
        \begin{tabular}{r|l} \textcolor{TUMBlue}{$\pi_{j_2} = 0$} & $v_{j_2}(\{a, b, b\}) = 12$ \\ \textcolor{TUMGreen}{$\tilde\pi_{j_2} = 0$} &
      \end{tabular}};

      \begin{scope}[on background layer]
        \foreach \n/\t in {i1/text-i1, i2/text-i2, i3/text-i3, j1/text-j1, j2/text-j2}
        \node[fill=gray!10, rounded corners=3mm, inner sep=1.2mm, fit=(\n) (\t)] {};
      \end{scope}

      \draw[->, TUMGreen] (i1) -- node[midway, above, sloped, font=\scriptsize] {$\tilde p_{i_1j_1} = 6$} (j1);
      \draw[->, TUMBlue] (i1) -- node[midway, above, sloped, font=\scriptsize] {$p_{i_1j_2} = 6$} (j2);
      \draw[->, TUMBlue] (i2) -- node[midway, above, sloped, font=\scriptsize] {$p_{i_2j_2} = 3$} (j2);
      \draw[->, TUMBlue] (i3) -- node[midway, above, sloped, font=\scriptsize] {$p_{i_3j_2} = 3$} (j2);

      \begin{scope}[on background layer]
        \node[draw=black, line width=0.5pt, rounded corners=5mm, inner sep=10pt, fit=(diagram)] {};
      \end{scope}
    \end{scope}
  \end{tikzpicture}

  \caption{Graph representation of the assignments in Example~\ref{example:sRTUCorePriceability}, $((\vecc{S}_I, \vecc{Z}_J), \vecc p)$ with blue arrows and $((\vecc{\tilde S}_I, \vecc{\tilde Z}_J), \vecc{p})$ with green arrows.}
  \label{figure:sRTUCorePriceability}
\end{figure}

First, consider an allocation $(\vecc{S}_I, \vecc{Z}_J)$ in which seller $j_2$ offers all their items to the buyers. In particular, let $\vecc{s}_{i_1j_2} = \vecc{z}_{i_1j_2} \equiv \{a\}$ and $\vecc{s}_{i_2j_2} = \vecc{z}_{i_2j_2} = \vecc{z}_{i_3j_2} = \vecc{z}_{i_3j_2} \equiv \{b\}$, while all the other individual bundles are empty. Let the pricing system $\vecc{p}$ be defined by $p_a = 6$ and $p_b = 3$. Since the allocation is second-order admissible and the pricing system is first-order admissible, the assignment $((\vecc{S}_I, \vecc{Z}_J), \vecc p)$ is first- and second-order admissible. The resulting payoffs are $\pi_{i_2} = 1$ and $\pi_k = 0$ for any other agent $k \in N\setminus\{i_2\}$, so that the total payoff is $\sum_{k \in N} \pi_k = 1$ (see Figure~\ref{figure:sRTUCorePriceability}).
It is immediate to show that this assignment is NTU-stable in any market setting $\mathcal{E}_h^r$, including for every $h \in \{1, 2\}$ and $r \in \{\mathrm{sell},\, \mathrm{no}\}$. Observe that buyer $i_1$ may request one unit of item $a$ from either seller $j_1$ or seller $j_2$. However, since seller $j_1$ cannot profitably receive a price below $v_{j_1}(\{a\}) = 6$, buyer $i_1$ wouldn't improve their payoff by requesting $a$ from $j_1$ instead.

Consider next the assignment $((\vecc{\tilde S}_I, \vecc{\tilde Z}_J), \vecc{p})$ in which $\vecc{\tilde s}_{i_1j_1} = \vecc{\tilde z}_{i_1j_1} \equiv \{a\}$, while all the remaining bundles are empty. The corresponding total payoff is $\sum_{k \in N} \tilde\pi_k = 0$ (see Figure~\ref{figure:NTUCorePriceability}).
In this case, the assignment is NTU-stable only in the markets $\mathcal{E}_h^r$ with $h \in \{1, 2\}$ and $r \in \{\mathrm{sell},\, \mathrm{no}\}$, since in these settings there is no coalition in which buyer $i_1$ can request item $a$ from seller $j_2$ and achieve a positive payoff. To see this, observe that $i_1$ can afford to request $a$ from $j_2$ only if both buyers $i_2$ and $i_3$ each request one unit of item $b$ from $j_2$. This is because $j_2$ must receive a cumulative price of at least $v_{j_2}(\{a, b, b\}) = 12$ for the assignment to be feasible. However, because of the pricing restrictions, both $i_2$ and $i_3$ can only pay a maximum price of
$$\min\{v_{i_2}(\{b\}), v_{i_3}(\{b\})\} = 3.$$
Therefore, $i_1$ must pay at least $12 - 2 \cdot 3 = 6$, which leaves them with a non-positive payoff.
\end{example}

\subsubsection{Different Social Welfare Across the Different Stability Notions}
\begin{example}\label{example:RTUUniquePayoff}
Consider an economy $\mathcal{E}$ with three buyers $I = \{i_1, i_2, i_3\}$ and three sellers $J = \{j_1, j_2, j_3\}$. The sellers' initial endowments are $\vecc\Omega_{j_1} \equiv \{a, a\}$ and $\vecc\Omega_{j_2} = \vecc\Omega_{j_3} \equiv \{b\}$, so that the total endowment of the economy is $\vecc\Omega \equiv \{a, a, b, b\}$. The agents' valuations are reported in Table~\ref{table:RTUUniquePayoff}.

\begin{table}[ht]
  \centering
  \begin{tabular}{|c|c|c|c|c|c|c|c|c|}
    \hline
    $\vecc{\omega}$ & $\{a\}$ & $\{b\}$ & $\{a, a\}$ & $\{a, b\}$ & $\{b, b\}$ & $\{a, a, b\}$ & $\{a, b, b\}$ & $\{a, a, b, b\}$ \\
    \hline
    $v_{i_1}(\vecc{\omega})$ & $9$ & $0$ & $9$ & $9$ & $0$ & $9$ & $9$ & $9$ \\
    \hline
    $v_{i_2}(\vecc{\omega})$ & $4$ & $0$ & $4$ & $4$ & $0$ & $4$ & $4$ & $4$  \\
    \hline
    $v_{i_3}(\vecc{\omega})$ & $0$ & $0$ & $0$ & $0$ & $8$ & $0$ & $8$ & $8$  \\
    \hline
    $v_{j_1}(\vecc{\omega})$ & $10$ &$-$ & $10$ & $-$ & $-$ & $-$ & $-$ & $-$ \\
    \hline
    $v_{j_2}(\vecc{\omega})$ & $-$ & $5$ & $-$ & $-$ & $-$ & $-$ & $-$ & $-$ \\
    \hline
    $v_{j_3}(\vecc{\omega})$ & $-$ & $1$ & $-$ & $-$ & $-$ & $-$ & $-$ & $-$ \\
    \hline
  \end{tabular}

  \caption{Valuations $v_k(\vecc{\omega})$ for each agent $k \in N$ and each bundle $\vecc{\omega} \leq \vecc{\Omega}$ in Example~\ref{example:RTUUniquePayoff}.}
  \label{table:RTUUniquePayoff}
\end{table}

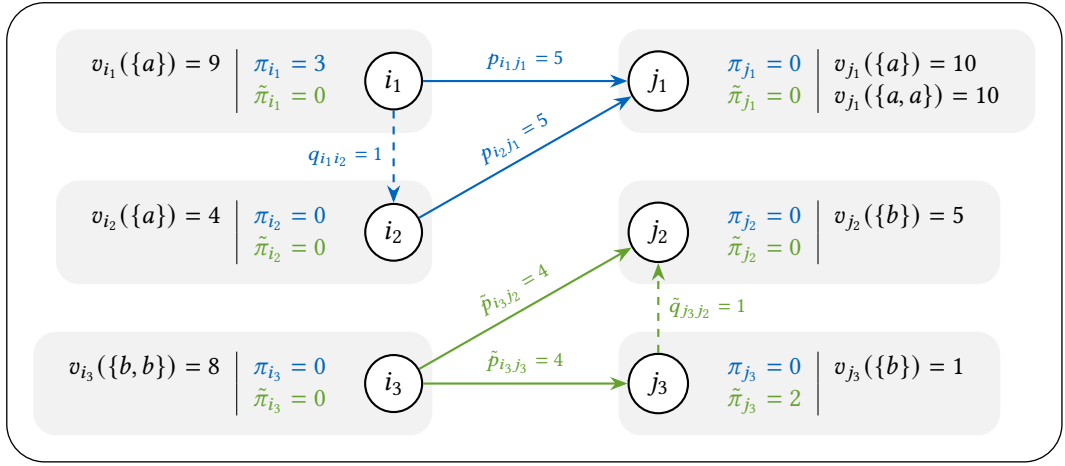
\begin{figure}[ht]
  \centering
  \adjustbox{center}{%
    \begin{tikzpicture}[>=Stealth, thick, node/.style={draw, circle, fill=white, minimum size=0.7cm}]
      \begin{scope}[local bounding box=diagram]
        \node[node] (i1) at (0,0) {$i_1$};
        \node[left=0.15cm of i1, font=\small] (text-i1) {
          \begin{tabular}{r|l} $v_{i_1}(\{a\}) = 9$ & \textcolor{TUMBlue}{$\pi_{i_1} = 3$} \\ & \textcolor{TUMGreen}{$\tilde\pi_{i_1} = 0$}
        \end{tabular}};

        \node[node] (i2) at (0,-2) {$i_2$};
        \node[left=0.15cm of i2, font=\small] (text-i2) {
          \begin{tabular}{r|l} $v_{i_2}(\{a\}) = 4$ & \textcolor{TUMBlue}{$\pi_{i_2} = 0$} \\ & \textcolor{TUMGreen}{$\tilde\pi_{i_2} = 0$}
        \end{tabular}};

        \node[node] (i3) at (0,-4) {$i_3$};
        \node[left=0.15cm of i3, font=\small] (text-i3) {
          \begin{tabular}{r|l} $v_{i_3}(\{b, b\}) = 8$ & \textcolor{TUMBlue}{$\pi_{i_3} = 0$} \\ & \textcolor{TUMGreen}{$\tilde\pi_{i_3} = 0$}
        \end{tabular}};

        \node[node] (j1) at (3.5,0) {$j_1$};
        \node[right=0.15cm of j1, font=\small] (text-j1) {
          \begin{tabular}{r|l} \textcolor{TUMBlue}{$\pi_{j_1} = 0$} & $v_{j_1}(\{a\}) = 10$ \\ \textcolor{TUMGreen}{$\tilde\pi_{j_1} = 0$} & $v_{j_1}(\{a, a\}) = 10$
        \end{tabular}};

        \node[node] (j2) at (3.5,-2) {$j_2$};
        \node[right=0.15cm of j2, font=\small] (text-j2) {
          \begin{tabular}{r|l} \textcolor{TUMBlue}{$\pi_{j_2} = 0$} & $v_{j_2}(\{b\}) = 5$ \\ \textcolor{TUMGreen}{$\tilde\pi_{j_2} = 0$} &
        \end{tabular}};

        \node[node] (j3) at (3.5,-4) {$j_3$};
        \node[right=0.15cm of j3, font=\small] (text-j3) {
          \begin{tabular}{r|l} \textcolor{TUMBlue}{$\pi_{j_3} = 0$} & $v_{j_3}(\{b\}) = 1$ \\ \textcolor{TUMGreen}{$\tilde\pi_{j_3} = 2$} &
        \end{tabular}};

        \begin{scope}[on background layer]
          \foreach \n/\t in {i1/text-i1, i2/text-i2, i3/text-i3, j1/text-j1, j2/text-j2, j3/text-j3}
          \node[fill=gray!10, rounded corners=3mm, inner sep=1.2mm, fit=(\n) (\t)] {};
        \end{scope}

        \draw[->, TUMBlue] (i1) -- node[midway, above, sloped, font=\scriptsize] {$p_{i_1j_1} = 5$} (j1);
        \draw[->, TUMBlue] (i2) -- node[midway, above, sloped, font=\scriptsize] {$p_{i_2j_1} = 5$} (j1);
        \draw[->, TUMBlue, dashed] (i1) -- node[midway, left, font=\scriptsize] {$q_{i_1i_2} = 1$} (i2);
        \draw[->, TUMGreen] (i3) -- node[midway, above, sloped, font=\scriptsize] {$\tilde p_{i_3j_2} = 4$} (j2);
        \draw[->, TUMGreen] (i3) -- node[midway, above, sloped, font=\scriptsize] {$\tilde p_{i_3j_3} = 4$} (j3);
        \draw[->, TUMGreen, dashed] (j3) -- node[midway, right, font=\scriptsize] {$\tilde q_{j_3j_2} = 1$} (j2);

        \begin{scope}[on background layer]
          \node[draw=black, line width=0.5pt, rounded corners=5mm, inner sep=10pt, fit=(diagram)] {};
        \end{scope}
      \end{scope}
    \end{tikzpicture}
  }

  \caption{Graph representation of the assignments in Example~\ref{example:RTUUniquePayoff}, $((\vecc{S}_I, \vecc{Z}_J), \vecc p, \vecc{Q}_I)$ with blue arrows and $((\vecc{\tilde S}_I, \vecc{\tilde Z}_J), \vecc{\tilde p}, \vecc{\tilde Q}_J)$ with green arrows.}
  \label{figure:RTUUniquePayoff}
\end{figure}

\paragraph{No side-payment setting} If no side payments are allowed, the only first- or second-order feasible (and therefore stable) allocation is the empty one, with a total payoff of 0.

\paragraph{Buyers' single side-payment setting} If we introduce buyers' single side payments, the only NTU-stable assignment $((\vecc{S}_I, \vecc{Z}_J), \vecc p, \vecc Q_{I}) \in \mathcal{F}(\mathcal{E}_h^{\mathrm{buy}})$ for $h \in \{1, 2\}$ is the one represented with blue arcs in Figure~\ref{figure:RTUUniquePayoff}, where the non-empty individual bundles are $\vecc{s}_{i_1j_1} = \vecc{z}_{i_1j_1} = \vecc{s}_{i_2j_1} = \vecc{z}_{i_2j_1} \equiv \{a\}$, and the pricing system $\vecc p$ is such that $p_a = 5$ and $p_b = 4$. The total payoff is $\sum_{k \in N} \pi_k = 3$.

\paragraph{Sellers' single side-payment setting} With sellers' single side payments, the only NTU-stable assignment $((\vecc{\tilde S}_I, \vecc{\tilde Z}_J), \vecc{\tilde p}, \vecc{\tilde Q}_{J}) \in \mathcal{F}(\mathcal{E}_h^{\mathrm{buy}})$ for $h \in \{1, 2\}$ is the one represented with green arrows in Figure~\ref{figure:RTUUniquePayoff}, where $\vecc{\tilde s}_{i_3j_2} = \vecc{\tilde z}_{i_3j_2} = \vecc{\tilde s}_{i_3j_3} = \vecc{\tilde z}_{i_3j_3} \equiv \{b\}$, and the pricing system $\vecc{\tilde p}$ is again defined by $\tilde p_a = 5$ and $\tilde p_b = 4$. The total payoff is $\sum_{k \in N} \tilde \pi_k = 2$.

\paragraph{Full and separate side-payment settings} Finally, with full or separate side payments, the only stable assignment is the combination of the two previous ones, with a total payoff of $\sum_{k \in N} \pi_k + \sum_{k \in N} \tilde \pi_k = 5$.
\end{example}

\subsubsection{Markets with No Stable Assignments}
\begin{example}\label{example:emptyCores}
Consider an economy $\mathcal{E}$ with two buyers $I = \{i_1, i_2\}$ and two sellers $J = \{j_1, j_2\}$. The sellers' initial endowments are $\vecc\Omega_{j_1} \equiv \{a\}$ and $\vecc\Omega_{j_2} \equiv \{b\}$, so that the total endowment of the economy is $\vecc\Omega \equiv \{a, b\}$. The agents' valuations are reported in Table~\ref{table:emptyCores}.

\begin{table}[ht]
  \centering
  \begin{tabular}{|c|c|c|c|c|c|c|c|c|}
    \hline
    $\vecc{\omega}$ & $\{a\}$ & $\{b\}$ & $\{a, b\}$ \\
    \hline
    $v_{i_1}(\vecc{\omega})$ & $5.5$ & $6.5$ & $6.5$ \\
    \hline
    $v_{i_2}(\vecc{\omega})$ & $0$ & $0$ & $11$ \\
    \hline
    $v_{j_1}(\vecc{\omega})$ & $4$ & $-$ & $-$ \\
    \hline
    $v_{j_2}(\vecc{\omega})$ & $-$ & $5$ & $-$ \\
    \hline
  \end{tabular}

  \caption{Valuations $v_k(\vecc{\omega})$ for each agent $k \in N$ and each bundle $\vecc{\omega} \leq \vecc{\Omega}$ in Example~\ref{example:emptyCores}.}
  \label{table:emptyCores}
\end{table}

\begin{figure}[ht]
  \centering
  \adjustbox{center}{%
    \begin{tikzpicture}[>=Stealth, thick, node/.style={draw, circle, fill=white, minimum size=0.7cm}]
      \begin{scope}[local bounding box=diagram]
        \node[node] (i1) at (0,1) {$i_1$};
        \node[left=0.15cm of i1, font=\small] (text-i1) {
          \begin{tabular}{r|l} $v_{i_1}(\{a\}) = 5.5$ & \textcolor{TUMBlue}{$0 \leq \pi_{i_1} \leq 1.5$} \\ $v_{i_1}(\{b\}) = 6.5$ & \textcolor{TUMGreen}{$0 \leq \tilde\pi_{i_1} \leq 1.5$} \\ & \textcolor{TUMOrange}{$\hat\pi_{i_1} = 0$}
        \end{tabular}};

        \node[node] (i2) at (0,-1.5) {$i_2$};
        \node[left=0.15cm of i2, font=\small] (text-i2) {
          \begin{tabular}{r|l} $v_{i_2}(\{a, b\}) = 11$ & \textcolor{TUMBlue}{$\pi_{i_2} = 0$} \\ & \textcolor{TUMGreen}{$\tilde\pi_{i_2} = 0$} \\ & \textcolor{TUMOrange}{$0 \leq \hat\pi_{i_2} \leq 2$}
        \end{tabular}};

        \node[node] (j1) at (5,1) {$j_1$};
        \node[right=0.15cm of j1, font=\small] (text-j1) {
          \begin{tabular}{r|l} \textcolor{TUMBlue}{$0 \leq \pi_{j_1} \leq 1.5$} & $v_{j_1}(\{a\}) = 4$ \\ \textcolor{TUMGreen}{$\tilde\pi_{j_1} = 0$} & \\ \textcolor{TUMOrange}{$0 \leq\hat\pi_{j_1} \leq 2$} &
        \end{tabular}};

        \node[node] (j2) at (5,-1.5) {$j_2$};
        \node[right=0.15cm of j2, font=\small] (text-j2) {
          \begin{tabular}{r|l} \textcolor{TUMBlue}{$\pi_{j_2} = 0$} & $v_{j_2}(\{b\}) = 5$ \\ \textcolor{TUMGreen}{$0 \leq \tilde\pi_{j_2} \leq 1.5$} & \\ \textcolor{TUMOrange}{$0 \leq \hat\pi_{j_2} \leq 2 - \hat{\pi}_{j_1}$} &
        \end{tabular}};

        \begin{scope}[on background layer]
          \foreach \n/\t in {i1/text-i1, i2/text-i2, j1/text-j1, j2/text-j2}
          \node[fill=gray!10, rounded corners=3mm, inner sep=1.2mm, fit=(\n) (\t)] {};
        \end{scope}

        \draw[->, TUMBlue] (i1) -- node[midway, above, font=\scriptsize] {$4 \leq p_{a} \leq 5.5$} (j1);
        \draw[->, TUMGreen] (i1) -- node[midway, above, sloped, pos=0.75, font=\scriptsize] {$5 \leq \tilde p_{b} \leq 6.5$} (j2);
        \draw[->, TUMOrange] (i2) -- node[midway, above, sloped, pos=0.25, font=\scriptsize] {$4 \leq \hat p_{a}\leq 6$} (j1);
        \draw[->, TUMOrange] (i2) -- node[midway, below, font=\scriptsize] {$5 \leq \hat p_{b} \leq 11 -\hat{p}_a$} (j2);

        \begin{scope}[on background layer]
          \node[draw=black, line width=0.5pt, rounded corners=5mm, inner sep=10pt, fit=(diagram)] {};
        \end{scope}
      \end{scope}
  \end{tikzpicture}}

  \caption{Graph representation of the three non-trivial feasible assignments in Example~\ref{example:emptyCores}.}
  \label{figure:emptyCores}
\end{figure}
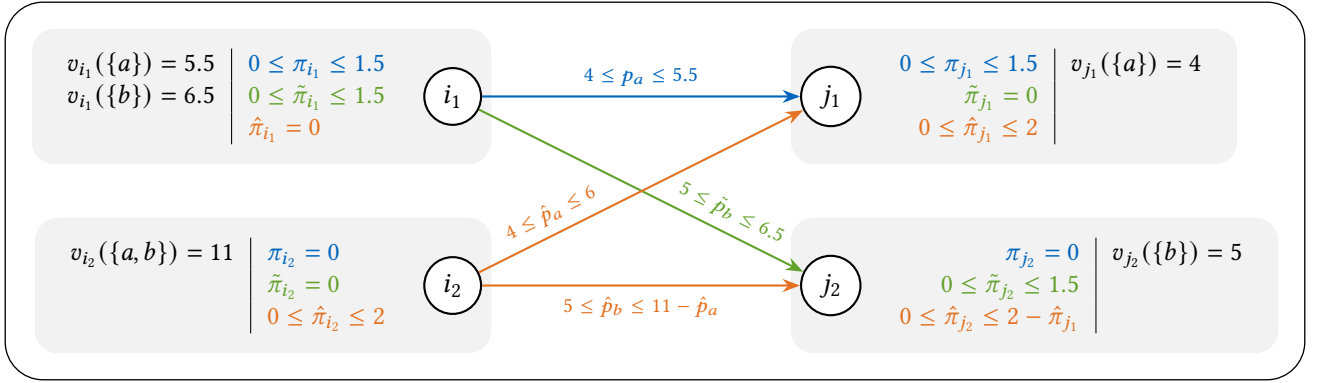

In any market setting $\mathcal{E}_h^r$, with $h \in \{1, \dots, 5\}$ and $r \in \{\mathrm{full}, \, \mathrm{sep}, \, \mathrm{buy}, \, \mathrm{sell}, \, \mathrm{no}\}$, the only non-trivial feasible assignments are the following:
\begin{enumerate}
  \item Buyer $i_1$ purchasing bundle $\vecc{s}_{i_1j_1} = \vecc{z}_{i_1j_1} \equiv\{a\}$ from seller $j_1$ at a price $4 \leq p_a \leq 5.5$.
  \item Buyer $i_1$ purchasing bundle $\vecc{\tilde s}_{i_1j_2} = \vecc{\tilde z}_{i_1j_2} \equiv \{b\}$ from seller $j_2$ at a price $5 \leq \tilde p_b \leq 6.5$.
  \item Buyer $i_2$ purchasing bundle $\vecc{\hat s}_{i_2j_1} = \vecc{\hat z}_{i_2j_1} \equiv \{a\}$ from seller $j_1$ and bundle $\vecc{\hat s}_{i_2j_2} = \vecc{\hat z}_{i_2j_2} \equiv\{b\}$ from seller $j_2$ at a total price of $9 \leq \hat p_a + \hat p_b \leq 11$.
\end{enumerate}
We now show that none of these assignments is stable.
\begin{enumerate}
  \item In the first case, the highest price that buyer $i_1$ can pay is $$p_a = v_{i_1}(\{a\}) = 5.5.$$ But buyer $i_2$ can offer to seller $j_1$ a strictly greater amount of money, so that both agents are strictly better off. For example, this can be achieved through the allocation $(\vecc{\hat S}_I, \vecc{\hat Z}_J)$ with prices $\hat p_a = 5.6 > p_a$ and $\hat p_b = 5.1$.
  \item In the second case, the maximum price that buyer $i_1$ can offer is $\tilde p_b = v_{i_1}(\{b\}) = 6.5$. However, buyer $i_2$ can offer to seller $j_2$ a strictly greater amount of money, for example, by realizing again the allocation $(\vecc{\hat S}_I, \vecc{\hat Z}_J)$ but with prices $\tilde p_a = 4$ and $\tilde p_b = 7 > \tilde p_b$.
  \item In the third case, the maximum cumulative price for buyer $i_2$ is $\hat p_{i_2} = v_{i_2}(\{a, b\}) = 11$. This means that if he pays seller $j_1$ a price $\hat p_{a} > 5.5 = v_{i_1}(\{a\})$ to not make him deviate with buyer $i_1$, then he can pay seller $j_2$ only
    $$\hat p_{b} < v_{i_2}(\{a, b\}) - \hat p_{a} = 5.5 < v_{i_1}(\{b\}),$$
    so $j_2$ will deviate with $i_1$. Similarly, if he pays seller $j_2$ an amount $\hat p_{b} > 6.5 = v_{i_1}(\{b\})$, then he can pay seller $j_1$ only
    $$\hat p_{a} < v_{i_2}(\{a, b\}) - \hat p_{b} = 4.5 < v_{i_1}(\{a\}),$$
    so $j_1$ will deviate with $i_1$.
\end{enumerate}
\end{example}

\subsubsection{Markets with Empty TU-Core, Non-empty NTU-Core, and Arbitrarily Large Social Welfare for RTU-Core}

\begin{example}\label{example:sidePaymentsAreCool}
Consider an economy $\mathcal{E}$ with four buyers $I = \{i_1, \dots, i_4\}$ and three sellers $J = \{j_1, j_2, j_3\}$. The sellers' initial endowments are $\vecc\Omega_{j_1} = \vecc\Omega_{j_2} \equiv \{a, b\}$ and $\vecc\Omega_{j_3} \equiv \{c, c\}$, so that the total endowment of the economy is $\vecc\Omega \equiv \{a, a, b, b, c, c\}$. The agents' valuations are reported in Table~\ref{table:sidePaymentsAreCool}, where $\alpha \gg 0$.

\begin{table}[ht]
  \centering
  \begin{tabular}{|c|c|c|c|c|c|c|c|c|}
    \hline
    $\vecc{\omega}$ & $\{a\}$ & $\{b\}$ & $\{c\}$ & $\{a, a\}$ & $\{a, b\}$ & $\{c, c\}$ & Others \\
    \hline
    $v_{i_1}(\vecc{\omega})$ & $0$ & $0$ & $0$ & $10$ & $0$ & $0$ & $0$ or $10$ \\
    \hline
    $v_{i_2}(\vecc{\omega})$ & $0$ & $5$ & $0$ & $0$ & $5$ & $0$ & $0$ or $5$  \\
    \hline
    $v_{i_3}(\vecc{\omega})$ & $0$ & $0$ & $2\alpha - 1$ & $0$ & $0$ & $2\alpha - 1$ & $0$ or $2\alpha - 1$ \\
    \hline
    $v_{i_4}(\vecc{\omega})$ & $0$ & $0$ & $\alpha - 1$ & $0$ & $0$ & $\alpha - 1$ & $0$ or $\alpha - 1$ \\
    \hline
    $v_{j_1}(\vecc{\omega})$ & $3$ & $3$ & $-$ & $-$ & $20$ & $-$ & $-$ \\
    \hline
    $v_{j_2}(\vecc{\omega})$ & $6$ & $3$ & $-$ & $-$ & $20$ & $-$ & $-$ \\
    \hline
    $v_{j_3}(\vecc{\omega})$ & $-$ & $-$ & $2\alpha$ & $-$ & $-$ & $2\alpha$ & $-$ \\
    \hline
  \end{tabular}

  \caption{Valuations $v_k(\vecc{\omega})$ for each agent $k \in N$ and each bundle $\vecc{\omega} \leq \vecc{\Omega}$ in Example~\ref{example:sidePaymentsAreCool}.}
  \label{table:sidePaymentsAreCool}
\end{table}

\begin{figure}[ht]
  \centering
  \begin{tikzpicture}[>=Stealth, thick, node/.style={draw, circle, fill=white, minimum size=0.7cm}]
    \begin{scope}[local bounding box=diagram]
      \node[node] (i1) at (0,0.5) {$i_1$};
      \node[left=0.15cm of i1, font=\small] (text-i1) {
        \begin{tabular}{r|l} $v_{i_1}(\{a, a\}) = 10$ & \textcolor{TUMBlue}{$\pi_{i_1} = 0$} \\ & \textcolor{TUMGreen}{$\tilde\pi_{i_1} = 0.2$}
      \end{tabular}};

      \node[node] (i2) at (0,-2) {$i_2$};
      \node[left=0.15cm of i2, font=\small] (text-i2) {
        \begin{tabular}{r|l} $v_{i_2}(\{b\}) = 5$ & \textcolor{TUMBlue}{$\pi_{i_2} = 2$} \\ & \textcolor{TUMGreen}{$\tilde\pi_{i_2} = 0$}
      \end{tabular}};

      \node[node] (i3) at (0,-4) {$i_3$};
      \node[left=0.15cm of i3, font=\small] (text-i3) {
        \begin{tabular}{r|l} $v_{i_3}(\{c\}) = 2\alpha - 1$ & \textcolor{TUMBlue}{$\pi_{i_3} = \alpha - 2$} \\ & \textcolor{TUMGreen}{$\tilde\pi_{i_3} = 0$}
      \end{tabular}};

      \node[node] (i4) at (0,-6) {$i_4$};
      \node[left=0.15cm of i4, font=\small] (text-i4) {
        \begin{tabular}{r|l} $v_{i_4}(\{c\}) = \alpha - 1$ & \textcolor{TUMBlue}{$\pi_{i_4} = 0$} \\ & \textcolor{TUMGreen}{$\tilde\pi_{i_4} = 0$}
      \end{tabular}};

      \node[node] (j1) at (3.5,0.5) {$j_1$};
      \node[right=0.15cm of j1, font=\small] (text-j1) {
        \begin{tabular}{r|l} \textcolor{TUMBlue}{$\pi_{j_1} = 0$} & $v_{j_1}(\{a\}) = 3$ \\ \textcolor{TUMGreen}{$\tilde\pi_{j_1} = 0.6$} & $v_{j_1}(\{b\}) = 3$ \\ & $v_{j_1}(\{a, b\}) = 20$
      \end{tabular}};

      \node[node] (j2) at (3.5,-2) {$j_2$};
      \node[right=0.15cm of j2, font=\small] (text-j2) {
        \begin{tabular}{r|l} \textcolor{TUMBlue}{$\pi_{j_2} = 0$} & $v_{j_2}(\{a\}) = 6$ \\ \textcolor{TUMGreen}{$\tilde\pi_{j_2} = 0.2$} & $v_{j_2}(\{b\}) = 3$ \\ & $v_{j_2}(\{a, b\}) = 20$
      \end{tabular}};

      \node[node] (j3) at (3.5,-5) {$j_3$};
      \node[right=0.15cm of j3, font=\small] (text-j3) {
        \begin{tabular}{r|l} \textcolor{TUMBlue}{$\pi_{j_3} = 0$} & $v_{j_3}(\{c, c\}) = 2\alpha$ \\ \textcolor{TUMGreen}{$\tilde\pi_{j_3} = 0$} &
      \end{tabular}};

      \begin{scope}[on background layer]
        \foreach \n/\t in {i1/text-i1, i2/text-i2, i3/text-i3, i4/text-i4, j1/text-j1, j2/text-j2, j3/text-j3}
        \node[fill=gray!10, rounded corners=3mm, inner sep=1.2mm, fit=(\n) (\t)] {};
      \end{scope}

      \draw[->, TUMBlue] (i2) -- node[pos=0.5, sloped, below, font=\scriptsize] {$p_{i_2j_2} = 3$} (j2);
      \draw[->, TUMGreen] (i1) -- node[pos=0.5, sloped, above, font=\scriptsize] {$\tilde p_{i_1j_1} = 4.9$} (j1);
      \draw[->, TUMGreen] (i1) -- node[pos=0.5, sloped, above, font=\scriptsize] {$\tilde p_{i_1j_2} = 4.9$} (j2);
      \draw[->, dashed, TUMGreen] (j1) -- node[pos=0.5, right, font=\scriptsize] {$\tilde q_{j_1j_2} = 1.3$} (j2);
      \draw[->, TUMBlue] (i3) -- node[pos=0.5, sloped, above, font=\scriptsize] {$p_{i_3j_3} = \alpha$} (j3);
      \draw[->, TUMBlue] (i4) -- node[pos=0.5, sloped, below, font=\scriptsize] {$p_{i_4j_3} = \alpha$} (j3);
      \draw[->, dashed, TUMBlue] (i3) -- node[pos=0.5, left, font=\scriptsize] {$q_{i_3i_4} = 1$} (i4);

      \begin{scope}[on background layer]
        \node[draw=black, line width=0.5pt, rounded corners=5mm, inner sep=10pt, fit=(diagram)] {};
      \end{scope}
    \end{scope}
  \end{tikzpicture}

  \caption{Graph representation of the assignments in Example~\ref{example:sidePaymentsAreCool}. With blue arrows is represented an NTU-stable assignment in $\mathcal{E}_h^{\mathrm{buy}}$ with $h \in \{1, 2\}$, whereas with green arrows a possible deviation in $\mathcal{E}_h^r$ with $r \succeq \mathrm{sell}$.}
  \label{figure:sidePaymentsAreCool}
\end{figure}
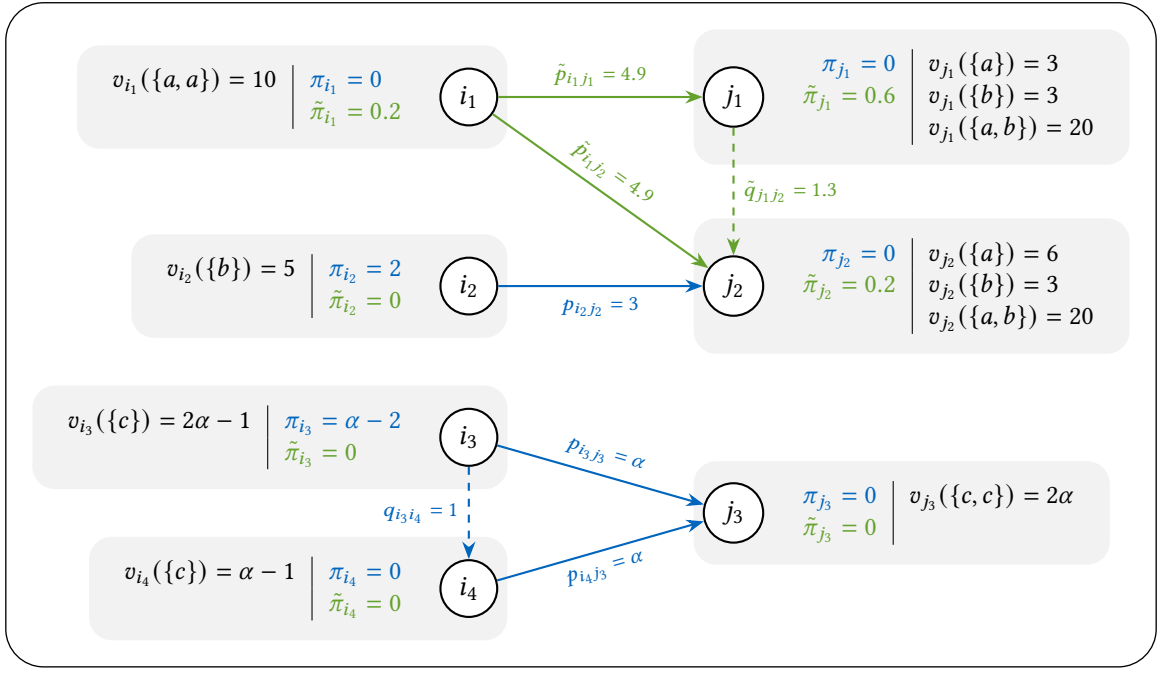

We now analyze NTU-stability in market settings with first- and second-order resale restrictions under decreasing levels of side-payment constraints.

\paragraph{No side-payment setting} If no side payments are allowed, the only non-trivial feasible assignments $((\vecc{S}_I, \vecc{Z}_J), \vecc p)$ are those in which buyer $i_2$ requests bundle $\{b\}$ either from seller $j_1$ or seller $j_2$ at a price $3 \leq p_b \leq 5$, while all the other individual bundles are empty. To see this, observe the following:
\begin{enumerate}
  \item Buyer $i_1$ cannot profitably purchase a single unit of item $a$. Indeed, since $v_{i_1}(\{a\}) = 0$, getting a non-negative payoff would require $p_a = 0$. However, feasibility for the sellers implies that the price $i_1$ must pay satisfies
    $$p_a \geq \min\{v_{j_1}(\{a\}), v_{j_2}(\{a\})\} = 3,$$
    if $i_1$ purchases $\{a\}$ alone, and
    $$p_a \geq \min\{v_{j_1}(\{a, b\}), v_{j_2}(\{a, b\})\} - v_{i_2}(\{b\}) = 15,$$
    if $i_1$ purchases $\{a\}$ jointly with $i_2$ purchasing $\{b\}$. Therefore, this trade cannot be part of any feasible assignment.

  \item Buyer $i_1$ cannot profitably purchase two units of item $a$. In this case, $v_{i_1}(\{a, a\}) = 10$. However, the pricing restrictions imply that $i_1$ must pay a cumulative price of
    $$2p_a \geq 2\max\{v_{j_1}(\{a\}), v_{j_2}(\{a\})\} = 12 > v_{i_1}(\{a, a\}),$$
    if acquiring the cumulative bundle $\{a, a\}$ on their own, and
    $$2p_a \geq 2\max_{j \in \{j_1, j_2\}}\max\{v_{j}(\{a\}), v_{j}(\{a, b\}) - v_{i_2}(\{b\})\} = 30 > v_{i_1}(\{a, a\}),$$
    if acquiring it jointly with $i_2$ purchasing $\{b\}$.
    Thus, in this trade as well, $i_1$ obtains a negative payoff.

  \item Buyers $i_3$ and $i_4$ cannot profitably purchase a single unit of item $c$, either individually or jointly. Individually, any buyer purchasing $\{c\}$ must pay at least
    $$p_c \geq v_{j_3}(\{c\}) = 2\alpha > v_{i_3}(\{c\}) > v_{i_4}(\{c\}),$$
    while if they purchase it simultaneously, pricing restrictions imply that each must pay at least
    $$p_c \geq \frac{v_{j_3}(\{c, c\})}{2} = \alpha > v_{i_4}(\{c\}).$$
    In both cases, at least one agent obtains a negative payoff, so the corresponding assignment is not feasible.
\end{enumerate}

Among the feasible assignments, the NTU-stable ones are those in which the price of item $b$ is $p_b = 3$. Otherwise, the seller $j \in \{j_1, j_2\}$ who does not trade with $i_2$ can offer them a lower price, inducing a deviation in which both $i_2$ and $j$ strictly increase their payoffs.
Note that the total payoff of any such stable assignment is $\sum_{k \in N} \pi_k = 2$.

\paragraph{Buyers' single side-payment setting} If we introduce buyers' single side payments, all NTU-stable assignments still involve buyer $i_2$ trading bundle $\{b\}$ either with $j_1$ or $j_2$ for a price $p_b = 3$, while buyer $i_1$ is not able to afford any non-empty bundle (the analysis is analogous to the previous case, noting that now $i_1$ and $i_2$ could in principle side-pay each other).
However, buyer $i_3$ can now side-pay buyer $i_4$ an amount of money $q_{i_3i_4} \geq 1$, enabling both buyers to jointly pay seller $j_3$ a price
$$\frac{v_{j_3}(\{c, c\})}{2} = \alpha \leq p_c \leq \frac{v_{i_3}(\{c\}) + v_{i_4}(\{c\})}{2} = \frac{3}{2}\alpha - 1$$
for one unit each of item $c$.
Consequently, the total payoff of any stable assignments is now $\sum_{k \in N} \pi_k = \alpha$, which can be made arbitrarily large by the choice of $\alpha$.

\paragraph{Full, separate, and sellers' single side-payment settings} Finally, when at least sellers' single side payments are allowed, the market does not admit any NTU-stable assignment. In particular, the only feasible trades involving agents $i_1$, $i_2$, $j_1$, and $j_2$ are the following:
\begin{enumerate}
  \item Buyer $i_1$ purchasing two copies of bundle $\{a\}$, one from $j_1$ and one from $j_2$, at a price $4.5 \leq p_a \leq 5$ each, and seller $j_1$ side-paying seller $j_2$ an amount $6 - p_a \leq q_{j_1j_2} \leq p_a - 1$.
  \item Buyer $i_2$ purchasing bundle $\{b\}$ from seller $j_1$ at a price of $3 \leq p_b \leq 5$.
  \item Buyer $i_2$ purchasing bundle $\{b\}$ from seller $j_2$ at a price of $3 \leq p_b \leq 5$.
\end{enumerate}
Suppose that we are in the first case and $i_1$ pays the maximum price $p_a = 5$ to both sellers $j_1$ and $j_2$. Then, at least one seller, say $j_1$, obtains a payoff $\pi_{j_1} < 2$, since the maximum total utility the two sellers can distribute among themselves is only $$2p_a - [v_{j_1}(\{a\}) + v_{j_2}(\{a\})] = 3.$$ Therefore, buyer $i_2$ can profitably deviate with seller $j_1$ by offering a price $\tilde p_b$ satisfying $\pi_{j_1} + 3 < \tilde p_b < 5$, which makes both agents strictly better off.

Suppose now $i_2$ trades bundle $\{b\}$ with one of the first two sellers, say $j_1$, at the maximum price $p_b = 5$, so that $j_1$'s payoff is $\pi_{j_1} = 2$. Then every agent in the coalition $C = \{i_1, j_1, j_2\}$ can strictly increase their payoff by deviating to an alternative assignment in which $i_1$ purchases one unit of item $a$ from each seller at a price $4.5 < \tilde p_a < 5$, while $j_1$ side-pays $j_2$ an amount $6 - \tilde p_a < \tilde q_{j_1j_2} < \tilde p_a - 3$.

\end{example}

\begin{example}\label{example:sidePaymentsAreCool2}
Consider an economy $\mathcal{E}$ with three buyers $I = \{i_1, i_2, i_3\}$ and four sellers $J = \{j_1, \dots, j_4\}$. The sellers' initial endowments are respectively $\vecc\Omega_{j_1} \equiv \{a, a\}$, $\vecc\Omega_{j_2} \equiv \{b\}$, and $\vecc\Omega_{j_3} = \vecc\Omega_{j_4} \equiv \{c\}$, so that the total endowment of the economy is $\vecc\Omega \equiv \{a, a, b, c, c\}$. The agents' valuations are reported in Table~\ref{table:sidePaymentsAreCool2}, where $\alpha \gg 0$.

\begin{table}[ht]
  \centering
  \begin{tabular}{|c|c|c|c|c|c|c|c|c|}
    \hline
    $\vecc{\omega}$ & $\{a\}$ & $\{b\}$ & $\{c\}$ & $\{a, a\}$ & $\{a, b\}$ & $\{c, c\}$ & Others \\
    \hline
    $v_{i_1}(\vecc{\omega})$ & $4$ & $6$ & $0$ & $4$ & $6$ & $0$ & $4$ or $6$ \\
    \hline
    $v_{i_2}(\vecc{\omega})$ & $9$ & $6$ & $0$ & $9$ & $9$ & $0$ & $6$ or $9$ \\
    \hline
    $v_{i_3}(\vecc{\omega})$ & $0$ & $0$ & $0$ & $0$ & $0$ & $2\alpha - 1$ & $0$ or $2\alpha - 1$ \\
    \hline
    $v_{j_1}(\vecc{\omega})$ & $10$ & $-$ & $-$ & $10$ & $-$ & $-$ & $-$ \\
    \hline
    $v_{j_2}(\vecc{\omega})$ & $-$ & $4$ & $-$ & $-$ & $-$ & $-$ & $-$ \\
    \hline
    $v_{j_3}(\vecc{\omega})$ & $-$ & $-$ & $1$ & $-$ & $-$ & $-$ & $-$\\
    \hline
    $v_{j_4}(\vecc{\omega})$ & $-$ & $-$ & $\alpha$ & $-$ & $-$ & $-$ & $-$\\
    \hline
  \end{tabular}

  \caption{Valuations $v_k(\vecc{\omega})$ for each agent $k \in N$ and each bundle $\vecc{\omega} \leq \vecc{\Omega}$ in Example~\ref{example:sidePaymentsAreCool2}.}
  \label{table:sidePaymentsAreCool2}
\end{table}

\begin{figure}[ht]
  \centering
  \begin{tikzpicture}[>=Stealth, thick, node/.style={draw, circle, fill=white, minimum size=0.7cm}]
    \begin{scope}[local bounding box=diagram]
      \node[node] (i1) at (0,0) {$i_1$};
      \node[left=0.15cm of i1, font=\small] (text-i1) {
        \begin{tabular}{r|l} $v_{i_1}(\{a\}) = 4$ & \textcolor{TUMBlue}{$\pi_{i_1} = 0$} \\ $v_{i_1}(\{b\}) = 6$ & \textcolor{TUMGreen}{$\tilde\pi_{i_1} = 0.4$}
      \end{tabular}};

      \node[node] (i2) at (0,-2) {$i_2$};
      \node[left=0.15cm of i2, font=\small] (text-i2) {
        \begin{tabular}{r|l} $v_{i_2}(\{a\}) = 9$ & \textcolor{TUMBlue}{$\pi_{i_2} = 0$} \\ $v_{i_2}(\{b\}) = 6$ & \textcolor{TUMGreen}{$\tilde\pi_{i_2} = 2.4$}
      \end{tabular}};

      \node[node] (i3) at (0,-5) {$i_3$};
      \node[left=0.15cm of i3, font=\small] (text-i3) {
        \begin{tabular}{r|l} $v_{i_3}(\{c, c\}) = 2\alpha - 1$ & \textcolor{TUMBlue}{$\pi_{i_3} = 1$} \\ & \textcolor{TUMGreen}{$\tilde\pi_{i_3} = 0$}
      \end{tabular}};

      \node[node] (j1) at (3.5,0) {$j_1$};
      \node[right=0.15cm of j1, font=\small] (text-j1) {
        \begin{tabular}{r|l} \textcolor{TUMBlue}{$\pi_{j_1} = 0$} & $v_{j_1}(\{a, a\}) = 10$ \\ \textcolor{TUMGreen}{$\tilde\pi_{j_1} = 0.2$} &
      \end{tabular}};

      \node[node] (j2) at (3.5,-2) {$j_2$};
      \node[right=0.15cm of j2, font=\small] (text-j2) {
        \begin{tabular}{r|l} \textcolor{TUMBlue}{$\pi_{j_2} = 2$} & $v_{j_2}(\{b\}) = 4$ \\ \textcolor{TUMGreen}{$\tilde\pi_{j_2} = 0$} &
      \end{tabular}};

      \node[node] (j3) at (3.5,-4) {$j_3$};
      \node[right=0.15cm of j3, font=\small] (text-j3) {
        \begin{tabular}{r|l} \textcolor{TUMBlue}{$\pi_{j_3} = \alpha - 3$} & $v_{j_3}(\{c\}) = 1$ \\ \textcolor{TUMGreen}{$\tilde\pi_{j_3} = 0$} &
      \end{tabular}};

      \node[node] (j4) at (3.5,-6) {$j_4$};
      \node[right=0.15cm of j4, font=\small] (text-j4) {
        \begin{tabular}{r|l} \textcolor{TUMBlue}{$\pi_{j_4} = 0$} & $v_{j_4}(\{c\}) = \alpha$ \\ \textcolor{TUMGreen}{$\tilde\pi_{j_4} = 0$} &
      \end{tabular}};

      \begin{scope}[on background layer]
        \foreach \n/\t in {i1/text-i1, i2/text-i2, i3/text-i3, j1/text-j1, j2/text-j2, j3/text-j3, j4/text-j4}
        \node[fill=gray!10, rounded corners=3mm, inner sep=1.2mm, fit=(\n) (\t)] {};
      \end{scope}

      \draw[->, TUMBlue] (i2) -- node[pos=0.5, sloped, below, font=\scriptsize] {$p_{i_1j_2} = 6$} (j2);
      \draw[->, TUMGreen] (i1) -- node[pos=0.5, sloped, above, font=\scriptsize] {$\tilde p_{i_1j_1} = 5.1$} (j1);
      \draw[->, TUMGreen] (i2) -- node[pos=0.5, sloped, below, font=\scriptsize] {$\tilde p_{i_2j_1} = 5.1$} (j1);
      \draw[->, dashed, TUMGreen] (i2) -- node[pos=0.5, left, font=\scriptsize] {$q_{i_2i_1} = 1.5$} (i1);
      \draw[->, TUMBlue] (i3) -- node[pos=0.5, sloped, above, font=\scriptsize] {$p_{i_3j_3} = \alpha - 1$} (j3);
      \draw[->, TUMBlue] (i3) -- node[pos=0.5, sloped, below, font=\scriptsize] {$p_{i_3j_4} = \alpha - 1$} (j4);
      \draw[->, dashed, TUMBlue] (j3) -- node[pos=0.5, right, font=\scriptsize] {$q_{j_3j_4} = 1$} (j4);

      \begin{scope}[on background layer]
        \node[draw=black, line width=0.5pt, rounded corners=5mm, inner sep=10pt, fit=(diagram)] {};
      \end{scope}
    \end{scope}
  \end{tikzpicture}

  \caption{Graph representation of the assignments in Example~\ref{example:sidePaymentsAreCool2}. With blue arrows is represented an NTU-stable assignment in $\mathcal{E}_h^{\mathrm{sell}}$ with $h \in \{1, 2\}$, whereas with green arrows a possible deviation in $\mathcal{E}_h^r$ with $r \succeq \mathrm{buy}$.}
  \label{figure:sidePaymentsAreCool2}
\end{figure}
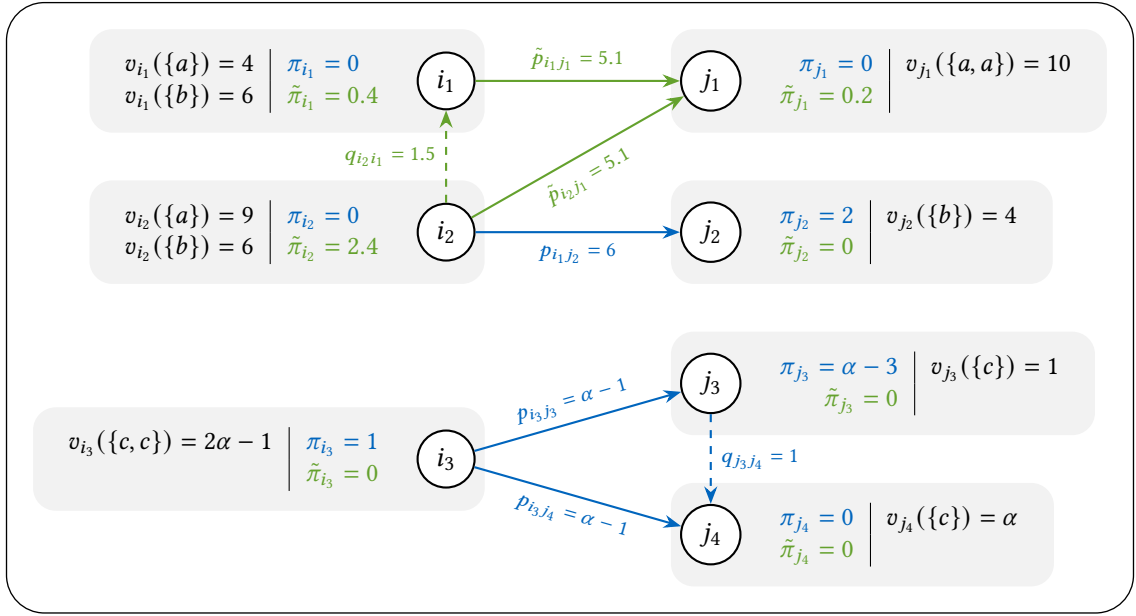

We now analyze NTU-stability in the corresponding markets with first- and second-order resale restrictions and under decreasing levels of side-payment restrictions.

\paragraph{No side-payment setting} If no side payments are allowed, the only non-trivial feasible assignments $((\vecc S_I, \vecc Z_J), \vecc p)$ are those in which either buyer $i_1$ or buyer $i_2$ requests bundle $\{b\}$ from seller $j_2$ at a price $4 \leq p_b \leq 6$, while all the other individual bundles are empty. To see this, observe the following:
\begin{enumerate}
  \item Buyers $i_1$ and $i_2$ cannot profitably purchase a single unit of item $a$, either individually or jointly. Individually, any buyer purchasing $\{a\}$ must pay at least $$p_a \geq v_{j_1}(\{a\}) = 10 > v_{i_2}(\{a\}) > v_{i_1}(\{a\}),$$
    while if they purchase it simultaneously, pricing restrictions imply that each must pay at least $$p_a \geq \frac{v_{j_1}(\{a, a\})}{2} = 5 > v_{i_1}(\{a\}).$$
    In both cases, at least one agent obtains a negative payoff, so the corresponding assignment is not feasible.

  \item Buyers $i_1$ and $i_2$ cannot profitably purchase bundle $\{a, b\}$. If either buyer purchases $\{a, b\}$ individually, they must pay at least
    $$v_{j_1}(\{a\}) + v_{j_2}(\{b\}) = 14 > v_{i_2}(\{a, b\}) > v_{i_1}(\{a, b\}).$$
    If buyer $i_1$ requests bundle $\{a, b\}$ while buyer $i_2$ requests simultaneously bundle $\{a\}$, pricing restrictions imply that $i_1$ must pay at least
    $$\frac{v_{j_1}(\{a, a\})}{2} + v_{j_2}(\{b\}) = 9 > v_{i_1}(\{a,b\}).$$
    Vice versa, if $i_1$ requests bundle $\{a\}$ while $i_2$ requests bundle $\{a, b\}$, $i_1$ must pay at least
    $$\frac{v_{j_1}(\{a, a\})}{2} = 5 > v_{i_1}(\{a\}).$$
    Thus, also in these trades, at least one agent cannot obtain a non-negative payoff.

  \item Buyer $i_3$ cannot profitably purchase either one or two units of item $c$. Indeed, the pricing restrictions imply that purchasing a single unit requires paying at least
    $$p_c \geq \min\{v_{j_3}(\{c\}), v_{j_4}(\{c\})\} = 1 > v_{i_3}(\{c\}),$$
    while purchasing two units requires a total payment satisfying
    $$2p_c \geq 2\max\{v_{j_3}(\{c\}), v_{j_4}(\{c\})\} = 2\alpha > v_{i_3}(\{c, c\}).$$
    Hence, $i_3$ cannot obtain a non-negative payoff in either case.
\end{enumerate}
Among the feasible assignments, the NTU-stable ones are those in which the price of item $b$ is $p_b = 6$. Otherwise, the buyer $i \in \{i_1, i_2\}$ who does not trade with $j_2$ can offer them a higher price, inducing a deviation in which both $i$ and $j_2$ strictly increase their payoffs. Note that the total payoff of any such stable assignment is $\sum_{k \in N} \pi_k = 2$.

\paragraph{Sellers' single side-payment setting} With sellers' single side payments, NTU-stable assignments still involve seller $j_2$ trading bundle $\{b\}$ with one of the first two buyers for a price $p_b = 6$, while seller $j_1$ is still not able to sell any of their items (the arguments from the previous setting extend here with minor modifications). However, buyer $i_3$ can now purchase one unit of item $c$ from each of sellers $j_3$ and $j_4$ at a price satisfying
$$\frac{v_{j_3}(\{c\}) + v_{j_4}(\{c\})}{2} = \frac{\alpha + 1}{2} \leq p_c \leq v_{i_3}(\{c, c\}) = 2\alpha - 1,$$
since $j_3$ can side-pay $j_4$ an amount $$\alpha - p_c \leq q_{j_3j_4} \leq p_c - 1.$$
Therefore, the payoff of any stable assignment is $\sum_{k \in N} \pi_k = \alpha$, which can be made arbitrarily large by the choice of $\alpha$.

\paragraph{Full, separate, and buyers' single side-payment settings} Finally, if at least buyers' single side payments are allowed, the market does not admit any NTU-stable assignment. In particular, the only feasible trades involving agents $i_1$, $i_2$, $j_1$, and $j_2$ are the following:
\begin{enumerate}
  \item Buyers $i_1$ and $i_2$ each purchasing one unit of item $a$ from $j_1$ at a price $5 \leq p_a \leq 6.5$, and $i_2$ side-paying $i_1$ an amount $p_a - 4 \leq q_{i_2i_1} \leq 9 - p_a$.
  \item Buyer $i_1$ purchasing bundle $\{b\}$ from seller $j_2$ at a price of $4 \leq p_b \leq 6$.
  \item Buyer $i_2$ purchasing bundle $\{b\}$ from seller $j_2$ at a price of $4 \leq p_b \leq 6$.
  \item Buyer $i_1$ purchasing bundle $\{a\}$ from seller $j_1$ and bundle $\{b\}$ from seller $j_2$, while buyer $i_2$ purchases bundle $\{a\}$ from seller $j_1$. The corresponding prices are $5 \leq p_a \leq 5.5$ and $4 \leq p_b \leq 15 - 2p_a$, with $i_2$ side-paying $i_1$ and amount $p_a + p_b - 6 \leq q_{i_2i_1} \leq 9 - p_a$.
\end{enumerate}
In the first case, suppose that each buyer pays the minimum price $p_a = 5$ to seller $j_1$. Then the maximum total utility that the buyers $i_1$ and $i_2$ can jointly obtain is
$$v_{i_1}(\{a\}) + v_{i_2}(\{a\}) - 2p_a \leq 3.$$
Therefore, at least one of the buyers, say $i_1$, must obtain a payoff $\pi_{i_1} < 2$. Buyer $i_1$ can then profitably deviate with seller $j_2$ by paying a price $4 < \tilde p_b < 6 - \pi_{i_1}$, which makes both agents strictly better off.

Suppose now that $i_1$ trades bundle $\{b\}$ with $j_2$ at the minimum price $p_b = 4$, so that $i_1$'s payoff is $\pi_{i_1} = 2$. In this case, $i_1$ and $i_2$ can jointly trade bundle $\{a\}$ with $j_1$ at a price $5 < \tilde p_a < 5.5$ together with a side payment $\tilde p_a - 2 < \tilde q_{i_2i_1} < 9 - \tilde p_a$, so that every deviating agent strictly increases their payoff. The same argument holds for the assignments in which buyer $i_2$ trades bundle $\{b\}$ with $j_2$.

Finally, in the fourth case, the total utility that the agents can distribute among themselves is:
$$ v_{i_1}(\{a, b\}) + v_{i_2}(\{a\}) - v_{j_1}(\{a, a\}) - v_{j_2}(\{b\}) = 1.$$
Therefore, any of the previously described assignments is a potential deviation.
\end{example}

\end{document}